\documentclass[preprint,12pt,authoryear,review]{elsarticle}
\usepackage[titletoc]{appendix}

\usepackage{tikz}

\usepackage{amssymb}
\usepackage{amsfonts}
\usepackage{mathrsfs} 
\usepackage{amsmath, amsthm, amssymb}
\usepackage{graphicx}
\usepackage{hyperref, comment}
\usepackage{bbm}

\usepackage[a4paper]{geometry}
\usepackage{pgf}

\usetikzlibrary{fit}					
\usetikzlibrary{backgrounds}	

\newtheorem{theorem}{Theorem}
\newtheorem{lemma}[theorem]{Lemma}

\newtheorem{proposition}[theorem]{Proposition}

\newtheorem{corollary}[theorem]{Corollary}
\newtheorem{definition}[theorem]{Definition}
\newtheorem{example}[theorem]{Example}

\newtheorem{remark}[theorem]{Remark}

\definecolor{darkblack}{rgb}{0, .07, .5}
\definecolor{darkred}{rgb}{0.5,0,0}
 \definecolor{mahogany}{rgb}{0.65, 0., 0.5}

\def\Pr{\text{\rm{Pr}}}

\newcommand{\Var}{\mathsf{Var}}

\newcommand{\E}{\mathbb{E}}

\newcommand{\Val}{\text{Val}}
\newcommand{\RN}[1]{%
  \textup{\uppercase\expandafter{\romannumeral#1}}%
}

\journal{Games and Economic Behavior}

\begin{document}

\begin{frontmatter}


\newcommand\blfootnote[1]{%
  \begingroup
  \renewcommand\thefootnote{}\footnote{#1}%
  \addtocounter{footnote}{-1}%
  \endgroup
}

\title{Playing Games with Bounded Entropy }


\author{Mehrdad Valizadeh, Amin Gohari}

\address{Department of Electrical Engineering, Sharif University of Technology, Tehran, Iran
(valizadeh@ee.sharif.edu, aminzadeh@sharif.edu)}

\begin{abstract}
In this paper, we consider zero-sum repeated games in which the maximizer (player or team) is restricted to strategies requiring no more than a limited amount of randomness. Particularly, we analyze the maxmin payoff of the maximizer in two models: the first model (the Neyman-Okada/Gossner-Vieille model) forces the maximizer to randomize her action in each stage just by conditioning her decision to the outcomes of a given sequence of random source, whereas in the second model (the Gossner and Tomala model), the maximizer is a team of players who are free to privately randomize their corresponding actions but do not have access to any explicit source of shared randomness needed for coordination. The works of Gossner and Vieille, and Gossner and Tomala adopted the method of types to establish their results; however, we utilize the idea of random hashing which is the core of randomness extractors in the information theory literature. In addition, we adopt the well-studied tool of simulation of a source from another source. By utilizing these tools, we are able to simplify the prior results and generalize them as well. We give a full characterization of the maxmin payoff of the maximizer in the repeated games under study. Particularly, the maxmin value of the first model is fully described by the function $\mathcal {J}(\mathtt h)$, where $\mathcal {J}(\mathtt h)$ is the maximum payoff that the maximizer can secure in the one-shot game by choosing mixed strategies of entropy at most $\mathtt h$. In the second part of the paper, we study the computational aspects of $\mathcal {J}(\mathtt h)$, which has not received much attention in the game theory literature. We observe the equivalence of this problem with entropy minimization problems in other scientific contexts. Next, we offer three explicit lower bounds on the entropy-payoff trade-off curve. To do this, we provide and utilize new results for the set of distributions that guarantee a certain payoff for Alice (mixed strategies corresponding to a security level for Alice). In particular, we study how this set of distributions shrinks as we increase the security level. While the use of total variation distance is common in game theory, our derivation indicates the suitability of utilizing the R\'enyi-divergence of order two.\blfootnote{A short version of this paper was presented at the 2017 IEEE Symposium On Information Theory (ISIT 2017).}
\end{abstract}


\end{frontmatter}



\tableofcontents
\section{Introduction}\label{S:intro}

\cite{Nash} proved that if the players of a given one-shot game can randomize on their pure strategies set according to \emph{any} probability distribution, then the game has at least one Nash equilibrium in the mixed strategies. However, if the players have restrictions on implementing their random actions, then implementable Nash equilibria do not necessarily exist. Specifically, assume that one of the players is restricted to constructing her actions as a deterministic function of a given random source. In this setup, Nash equilibria do not necessarily exist. Similarly, in a class of repeated games including two-player zero-sum games, Nash equilibria do not exist if sufficient random bits are not available to the players (See \cite{Hubacek} and \cite{Budinich}). Therefore, when the players are limited to a given amount of randomness, the maximum payoff that each player can guarantee, regardless of what strategies the other players choose, becomes of interest. In this paper, we study the maximum guaranteed payoff in repeated zero-sum games with bounded randomness under two models. In the first model, we consider a finitely repeated version of a two-player zero-sum game in which the private randomness available to one of the players is limited. The second model is a zero-sum game between a team and an adversary player. In this game, the team players are free to randomize their corresponding actions privately but do not have access to any explicit source of shared randomness to coordinate their actions. It is assumed that the adversary player monitors the played actions imperfectly; hence, the history of actions observed by the team players is an implicit and limited random source for their coordination.

A version of the first model was studied by \cite{Gossner2002}. They studied a repeated version of a zero-sum game $G$ between two players Alice and Bob, where Alice was the maximizer and Bob was the minimizer. At each stage, first, Alice observed an independent drawing of a random source $X$ whose distribution was a common knowledge. Both players then played an action which was monitored by the other player. Alice was restricted to choosing the action of each stage as a \emph{deterministic} function of the observed random source up to that stage along with the history of the previous actions, while Bob chose his actions at each stage as a \emph{random} function of the history of the previous actions. Note that the only source of randomization for Alice was the outcomes of random source $X$, while Bob could freely randomize his actions. \cite{Gossner2002} proved that when the number of stages of the repeated game is sufficiently large, the maximum expected average payoff of Alice is specified by the entropy-payoff trade-off curve of the one-shot game $G$. Specifically, for the one-shot game $G$, define $\mathcal{J}(\mathsf h)$ as the maximum expected payoff that Alice can secure (regardless of what Bob plays) by playing mixed actions of entropy at most $\mathsf h$. Let $\mathcal J_{\text{cav}}(.)$ be the upper concave envelope of $\mathcal J(.)$. \cite{Gossner2002} proved that the maximum expected average payoff of Alice in the repeated game converges to $\mathcal J_{\text{cav}}(H(X))$, where $H(X)$ is the entropy of random variable $X$.
This is similar to Shannon's compression formula $H(X)$, which is defined on a single copy of the source (\emph{single letter}), but gives the ultimate compression limit when multiple copies of the source are observed.

Generalizing the model of \cite{Gossner2002}, we study a repeated zero-sum game that considers the possibility of leakage of Alice's random source sequence to Bob, and hence we call it \emph{the repeated game with leaked randomness source}. In other words, we assume that Bob can imperfectly monitor the random source of Alice. More specifically, we assume an i.i.d.\ sequence of pairs $(X_1, Y_1)$, $(X_2, Y_2),\ldots$ distributed according to a given distribution $p(x,y)$. The sequence of $X_1, X_2, \ldots$ is revealed symbol by symbol (causally) to Alice as the game is played out, while the sequence of $Y_1, Y_2, \ldots$ is revealed symbol by symbol to Bob. We can view $Y_i$ as the leakage that Bob obtains about Alice's observation.
As before, Alice cannot randomize freely and is only able to use the randomness in the sequence of $X_1, X_2, \ldots$. We show that $\mathcal{J}_{\text{cav}}(H(X|Y))$ is the maximum payoff that Alice can secure regardless of Bob's actions. This result is a generalization of the result of \cite{Gossner2002}, since when $Y$ is a constant random variable, the conditional entropy $H(X|Y)$ reduces to the unconditional entropy $H(X)$. Furthermore, it is obtained that if a genie provides the values of $Y_1,Y_2,\ldots$ symbol by symbol for Alice, the maximum payoff she can secure remains unchanged. In other words, knowledge of what Bob knows about Alice's observations is not helpful for Alice.

The second model, originally studied by \cite{Gossner_Tomala}, takes into account the limited access of a team players to shared randomness needed for coordination against an adversary. This model, which is called \emph{secret correlation in repeated games with imperfect monitoring}, consists of a zero-sum game $G$ played repeatedly between team $A$ and player $B$. Team $A$, consisting of $m$ players, is the maximizer and player $B$ is the minimizer. At each stage, first, all players choose an action from their corresponding set of actions, and then the players of team $A$ observe all of the actions played, while player $B$ observes a noisy version of the action profile of team $A$. To choose the action of each stage, each player can privately randomize her action and utilize the history of her observations up to that stage. Since the action profile of team $A$ is revealed to player $B$ just through a noisy channel, the players of team $A$ could extract shared random bits from the profile of actions, where the extracted bits are \emph{almost} independent of the observations of player $B$. These shared bits can be utilized by players of team $A$ to coordinate their actions in the upcoming stages. Although the players of team $A$ can randomize their actions privately and extract shared random bits, the extracted shared bits are limited. Consequently, in this setup, the set of implementable strategies of team $A$ is constrained by the amount of shared randomness they could extract and use.

In the model of \cite{Gossner_Tomala}, at each stage $t$, player $B$ monitors the action profile of team $A$ through a noisy channel with output signal $S_t$. They assumed that $S_t$ was not only seen by player $B$ but also by the players of team $A$. Generalizing the result of \cite{Gossner_Tomala}, we remove the assumption that the players of team $A$ observe $S_t$, and show that this causes no reduction in the payoff of team $A$.

The above generalizations of the results of \cite{Gossner2002} and \cite{Gossner_Tomala} are immediate from our different proof technique. To explain this, consider that in order to construct the optimal strategies for the above two models of repeated games, we need to simulate random actions from the available source of randomness. In the first model, the source of randomness is the random sequence of $X_1,X_2,\ldots$; in the second model, the source of randomness is the history of (the imperfectly monitored) played actions. \cite{Gossner2002} and \cite{Gossner_Tomala} introduced a new notion of \emph{absolute Kullback distance}, which has not been used in the information theory literature. They took the absolute Kullback distance as the measure of the accuracy of the simulation, and utilized the \emph{method of types} to simulate the desired random actions from the source of randomness. On the other hand, we consider the total variation distance as the measure of accuracy, and utilize the standard method of \emph{random hashing}. The method of types alone is insufficient for obtaining the generalizations discussed in this paper, and hashing is necessary. We separate the simulation of random actions from the source of randomness into two steps: first, we extract sufficient private random bits from the randomness source, and then simulate the desired actions using the extracted random bits. Therefore, we need to use two tools for ``randomness extraction" and ``simulation of a source from another source". These tools are available in the literature of information theory for stationary or nonstationary/ergodic or non-ergodic sources (\emph{the information spectrum methods}). Randomness extraction is reviewed in \ref{sec:randomness_extraction} of this paper for completeness. Through utilizing these tools, we are able to simplify the proofs and generalize the results of \cite{Gossner2002} and \cite{Gossner_Tomala}.

As stated above, the maximum guaranteed payoff of Alice in the repeated game with leaked randomness source (first model) is characterized in terms of $\mathcal{J}(\mathsf h)$, which is the maximum payoff that Alice can secure in the one-shot game by choosing mixed actions with entropy at most $\mathsf h$. In this paper, we also study the computational aspects of $\mathcal{J}(\mathsf h)$. Equivalently, we study the inverse function of $\mathcal{J}(\mathsf h)$ denoted by $F(w)$. $F(w)$ is the minimum entropy of the randomness consumed by Alice to guarantee payoff $w$ in the one-shot game. We call $F(w)$ the \emph{min-entropy function}. \footnote{One should not confuse our ``min-entropy function" with the term ``min-entropy" commonly used to denote the R\'enyi entropy of order infinity.}

To compute $F(w)$, first, we need to consider the set of distributions on the action of Alice that would secure a payoff $w$ for her. This set will be a polytope in the space of all probability distributions. Then, we should solve an entropy minimization problem over this polytope in the space of probability distributions. In fact, minimizing and maximizing entropy arises in a wide range of contexts. Computing maximum entropy under a set of linear constraints is a well-studied problem with a wide range of applications, \emph{e.g.,} see \cite{Fang}, \cite[p.367]{Cover} and the principle of maximum entropy. \cite{MinMax} have shown that computing the minimum entropy can be also quite significant, and \cite{Watanabe} has shown that many algorithms for clustering and pattern recognition are essentially solving entropy minimization problems. An special case of the entropy minimization problem (with its own applications) is that of finding a joint probability of minimum entropy given its marginal distributions (the marginal distribution is a linear constraint on the joint probability distribution); see \cite{Kocaoglu, Cicalese, Kovacevic}. In addition, \cite{Shor} has shown that the quantum version of the entropy minimization problem is closely related to a number of noticeable problems in quantum information theory.

\cite{Kovacevic2} have shown that entropy minimization problem is an NP-hard non-convex optimization problem. Since the entropy is a concave function over the probability simplex, its minimum occurs at a vertex of the feasible domain. As a result, computation of $F( w )$ leads to a search problem over an exponentially large set. \cite{Computation} have proposed an algorithm to solve the entropy minimization problem (and hence, can be used to compute $F(w)$), but it has no guarantee of finding the global minimum for all polytopes.

$F(w)$ provides a game theoretic interpretation of the entropy minimization problem. In Section~\ref{S:prob_state}, we study the properties of function $F(w)$, and utilize \emph{probabilistic} tools to obtain a number of easy-to-compute bounds on the value of $F( w )$. While the literature on game theory makes extensive use of the total variation distance between distributions, we use the $\chi^2$--divergence (or the Tsallis divergence of order two) to derive a lower bound for $F(w)$. This lower bound is strictly tighter than the bound derived using the total variation distance showing the applicability of $\chi^2$--divergence in the context of game theory.

The rest of this paper is organized as follows: In Section~\ref{sec:pre}, we introduce the notation of this paper, and present a concise discussion of Shannon entropy. In Section~\ref{sec:rep_0}, we study a version of our first model studied by \cite{Gossner2002}, and simplify the proof. The complete version of our first model namely the repeated game with leaked randomness source will be studied in Section~\ref{sec:5op}. In Section~\ref{sec:imperfect_monitoring}, we investigate the problem of secret correlation in repeated games with imperfect monitoring, simplify the proof of \cite{Gossner_Tomala}, and extend their results. Section~\ref{S:prob_state} is devoted to the computational aspects of the min-entropy function and in Section~\ref{S:proofs}, we provide the proofs. Some of the details including a discussion on randomness extraction are left for the appendices.

\section{Preliminaries}\label{sec:pre}
\subsection{Notation} \label{def-sec-nkl4}
In this paper, we use the notation $x^j$ to represent a sequence of variables $(x_1,x_2,\ldots,x_j)$. The same notation is used to represent sequences of random variables, \emph{i.e.,} $X^j=(X_1,X_2,\ldots,X_j)$. Note that this notation is used for sequences that have two subscripts the same way \emph{i.e.,} $X_k^j=(X_{k,1},X_{k,2},\ldots,X_{k,j}).$ Calligraphic letters such as $\mathcal{X},\mathcal{Y}, \mathcal{A}, \mathcal{B}, \dots$ represent finite sets, and $|\mathcal X|$ denotes the cardinality of the finite set $\mathcal X$. Real vectors are represented by bold lower case letters, and bold uppercase letters are used to represent random vectors. For example the probability mass function (pmf) of a random variable $X$ with finite sample space $\mathcal{X}=\{1,2,\dots,n\}$, beside the representation $p_X(x)$, will be also denoted by a vector $\mathbf{p}=(p_1,\dots,p_n)$. When it is obvious from the context, we drop the subscript and use $p(x)$ instead of $p_X(x)$. We say that $X^n$ is drawn i.i.d.\ from $p(x)$ if
$$p(x^n)=\prod_{i=1}^np(x_i).$$
We use $\Delta (\mathcal{A})$ to denote the probability simplex on alphabet $\mathcal{A}$, \emph{i.e.,} the set of all probability distributions on the finite set $\mathcal{A}$.
The total variation distance between pmfs $p_X$ and $q_X$ is denoted by $d_1(p_X, q_X)$ or $\|p_X-q_X\|_{TV}$ and is defined as:
$$d_1(p_X,q_X)=\|p_X-q_X\|_{TV}\triangleq \frac 12 \sum_{x\in \mathcal X}|p_X(x)-q_X(x)|.$$
When the pmfs are represented by vectors $\mathbf{p}$ and $\mathbf{q}$, the total variation distance between them is represented by $d_1(\mathbf{p},\mathbf{q})$ or $\|\mathbf{p}-\mathbf{q}\|_{TV}$. Some of the properties of the total variation distance are summarized in the following lemma.
\begin{lemma}\label{lemma:tv_p}
The following properties hold for the total variation distance:
\begin{description}
\item[\quad\textbf{Property 1:}] $\|p_{E}p_{F|E}-q_{E}p_{F|E}\|_{TV}=\|p_{E}-q_{E}\|_{TV}$;
\item[\quad\textbf{Property 2:}] $\|p_{E}p_{F|E}-q_{E}q_{F|E}\|_{TV}\geq \|p_{E}-q_{E}\|_{TV}$;
\item[\quad\textbf{Property 3:}] $\|p_{E_1}p_{F_1}-p_{E_2}q_{F_2}\|_{TV}\leq \|p_{E_1}-p_{E_2}\|_{TV}+\|p_{F_1}-q_{F_2}\|_{TV}$;
\item[\quad\textbf{Property 4:}] For an arbitrary deterministic function $f$ on the sample space of $E_1$ and $E_2$ we have $\|p_{f(E_1)}-p_{f(E_2)}\|_{TV}\leq \|p_{E_1}-p_{E_2}\|_{TV}$;
\item[\quad\textbf{Property 5:}] $\|p_{E_1}p_{F}-p_{E_2}p_{F}\|_{TV} = \|p_{E_1}-p_{E_2}\|_{TV}$.
\end{description}
\end{lemma}
\subsection{Entropy function}
Let $X\in \mathcal{X}$ and $Y\in \mathcal Y$ be two random variables with joint probability distribution $p_{X,Y}$ and respective marginal distributions $p_X$ and $p_Y$. The Shannon entropy (or simply the entropy) of the random variable $X$ is defined to be:
$$H(X)=\sum_{x\in \mathcal X} -p_X(x) \log(p_X(x)),$$
where $0\log(0)=0$ by continuity and all logarithms in this paper are in base two. Since the entropy is a function of the pmf $p_X$, we sometimes write $H(p_X)$ (or $H(\mathbf p)$ when the pmf is denoted by probability vector $\mathbf p$) instead of $H(X)$.

The conditional Shannon entropy (or simply the conditional entropy) of $X$ given $Y$ is defined as:
\begin{align*}
H(X|Y)&=\sum_{(x,y)\in\mathcal{X}\times \mathcal{Y}} -p_{XY}(x,y) \log(p_{X|Y}(x|y))\\
&=\sum_{y\in \mathcal Y} p_Y(y)H(X|Y=y),
\end{align*}
where $H(X|Y=y)=\sum_{x\in \mathcal X} -p_{X|Y}(x|y) \log(p_{X|Y}(x|y))$.

Using the definition of entropy and conditional entropy one can check that
$$H(X,Y)=H(Y)+H(X|Y).$$
Utilizing the above property iteratively, the chain rule for a sequence of random variables is obtained:
$$H(X^n)=H(X_1)+\sum_{i=2}^n H(X_i|X^{i-1}).$$
Furthermore, the following properties hold for the entropy function:
\begin{itemize}
\item $H(X)\geq 0.$
\item $H(X_1,X_2)\geq H(X_1).$
\item Let $f(x)$ be an arbitrary deterministic function then $H(f(X))\leq H(X)$.
\end{itemize}

\section{Repeated games with a bounded randomness source}\label{sec:rep_0}
In this section, we revisit the repeated game studied by \cite{Gossner2002}. The high-level picture of the proof of \cite{Gossner2002} is as follows: it uses the so-called ``block-Markov" proof technique of information theory, where we divide time into a number of blocks and Alice uses her observations in each block to produce the actions for the next block. Our goal is to show that the high-level picture of the proof presented in \cite{Gossner2002} can be made precise in an easier manner using standard information theory tools.

\subsection{Problem definition}\label{subsec:gossner}
Let us first begin with reviewing the definition of the repeated game with a bounded randomness source studied by \cite{Gossner2002}. Consider a $T$ stage repeated zero-sum game between players Alice($A$) and Bob($B$) with respective finite action sets $\mathcal{A}=\{1,\dots,n\}$ and $\mathcal{B}=\{1,\dots,n'\}$, where $n$ and $n'$ are natural numbers. Let $X^T=(X_{1},X_2,\dots,X_{T})$ be a sequence of random variables drawn i.i.d.\ from a sample space $\mathcal{X}$ with law $p_{X}$.
 In every stage $t\in\{1,2,\dots,T\}$, first, Alice observes random source $X_{t}$ privately, and then, Alice and Bob choose actions $A_t\in \mathcal{A}$ and $B_t\in\mathcal{B}$, respectively. At the end of stage $t$, both players observe the chosen actions $A_t$ and $B_t$ and Alice gets stage payoff $u_{A_t,B_t}$ from Bob. In order to choose actions at stage $t$, players make use of the history of their observations up to stage $t$, which is denoted by $\mathsf H_1^t=(X_{1},A_1,B_1, \dots,X_{t-1},A_{t-1},B_{t-1},X_{t})$ for Alice, and $\mathsf H_2^t=(A_1,B_1,\dots,A_{t-1},B_{t-1})$ for Bob. Let $\sigma_t:(\mathcal{A}\times\mathcal{B})^{t-1}\times \mathcal{X}^t\to \mathcal{A}$ and $\tau_t:(\mathcal{A}\times\mathcal{B})^{t-1}\to \mathcal{B}$ be the functions mapping the history of observations of Alice and Bob to actions at stage $t$, so $A_t=\sigma_t(\mathsf H_1^t)$ and $B_t=\tau_t(\mathsf H_2^t)$. Note that Alice does not have access to any private source of randomness except $\mathsf H_1^t$, so she has to use the deterministic function $\sigma_t(\cdot)$, while Bob can utilize the random function $\tau_t(\cdot)$. We call the $T$-tuples $\sigma=(\sigma_1,\sigma_2,\dots,\sigma_T)$ and $\tau=(\tau_1,\tau_2,\dots,\tau_T)$ the strategies of Alice and Bob, respectively. The expected average payoff for Alice up to stage $T$ induced by strategies $\sigma$ and $\tau$ is denoted by $\lambda_T(\sigma,\tau)$, which is
\begin{equation}\label{eq:exp_avg}
\lambda_T(\sigma,\tau)=\E_{\sigma,\tau}\left[\frac{1}{T}\sum_{t=1}^T u_{A_t,B_t}\right],
\end{equation}
where $\E_{\sigma,\tau}$ denotes the expectation with respect to the distribution induced by i.i.d.\ repetitions of $p_{X}$ and strategies $\sigma$ and $\tau$. Alice wishes to maximize $\lambda_T(\sigma,\tau)$, and Bob wishes to minimize it.

\begin{definition}\label{def:max-min value}
Let $v$ be an arbitrary real value:
\begin{itemize}
\item Alice can secure $v$ if there exists a strategy $\sigma^*$ for Alice such that for all strategy $\tau$ of Bob we have $\liminf_{T\to \infty} \lambda_T(\sigma^*,\tau) \geq v$.
\item Bob defends $v$ if given an arbitrary strategy $\sigma$ for Alice, there exists a strategy $\tau^*$ for Bob such that $\limsup_{T\to \infty} \lambda_T(\sigma,\tau^*) \leq v$.
\item $v$ is the maxmin value of the repeated game, if Alice can secure $v$ and Bob can defend $v$.
\end{itemize}
\end{definition}
\begin{theorem}[\cite{Gossner2002}]\label{T:maxmin_val_gossner}
The maxmin value of the repeated game with bounded randomness source defined in Section~\ref{subsec:gossner} is $\mathcal J_{\text{cav}} (H(X))$, where $\mathcal J_{\text{cav}} (\mathsf h)$ is the upper concave envelope of
\begin{equation}\label{eq:j}
\mathcal{J}(\mathsf h) = \max_{\mathbf{p}\in \Delta(\mathcal{A}), H(\mathbf{p})\leq \mathsf{h}} \min_{b\in\mathcal{B}} \E_{\mathbf{p}}[u_{A,b}],
\end{equation}
where $\E_{\mathbf{p}}$ denotes the expectation with respect to $\mathbf{p}$.
\end{theorem}
In Section~\ref{subsec:gossner_proof}, we give a simplified proof of \cite{Gossner2002} that explains how Alice can secure $\mathcal J_{\text{cav}} (H(X))$. The complete proof is left for Section~\ref{sec:5op}, where we present a generalized version of the above result.

\subsection{Simplifying the proof of \cite{Gossner2002}}\label{subsec:gossner_proof}

As mentioned in \cite{Gossner2002}, the upper concave envelope $\mathcal J_{\text{cav}}(\cdot)$ at $\mathsf{h}=H(X)$ can be expressed as the convex combination
$$\gamma \mathcal J(H(p_A^{(1)}))+(1-\gamma)\mathcal J(H(p_A^{(2)})),$$
for some $\gamma\in[0,1]$ and pmfs $p_A^{(1)}$ and $p_A^{(2)}$ in $\Delta(\mathcal{A})$, where $p_A^{(1)}$ and $p_A^{(2)}$ secure respective payoffs $\mathcal{J}(H(p_A^{(1)}))$ and $\mathcal{J}(H(p_A^{(2)}))$ in the one-shot game and the following equality is satisfied:
$$\gamma H(p_A^{(1)}) + (1-\gamma)H(p_A^{(2)})=H(X).$$
As a result, it suffices to show that for any $p_A^{(1)}$ and $p_A^{(2)}$ satisfying
$$\gamma H(p_A^{(1)}) + (1-\gamma)H(p_A^{(2)})< H(X),$$
Alice can use the source of randomness $X$ to secure payoffs arbitrarily close to
$$\gamma \min_{b\in \mathcal{B}}\E_{p_A^{(1)}}[u_{A,b}] + (1-\gamma) \min_{b\in \mathcal{B}}\E_{p_A^{(2)}}[u_{A,b}],$$
which is the weighted average of the payoffs that input distributions $p_A^{(1)}$ and $p_A^{(2)}$ secure.

Let the game be played for $T$ stages. Take some $T$ of the form $T=NL$, and divide the total $T$ stages into $N$ blocks of length $L$. Excluding the first block, we generate the action sequence of each block as a function of the random source observed during the previous block (we ignore the payoff of the first block throughout the discussion, since by taking the number of blocks $N$ large enough, the contribution of the first block in Alice's net average payoff becomes negligible). More specifically, excluding the first block, each block is further divided into two subblocks, where the first subblock takes up $\gamma$ fraction of the block as illustrated in Fig.~\ref{F:blocks}. Alice aims to use her observed random source during the previous block to play \emph{almost} i.i.d.\ according to $p_A^{(1)}$ during the first subblock and according to $p_A^{(2)}$ during the second subblock. In addition, Alice wants her action at any given stage to be also \emph{almost} independent of Bob's observations up to that stage. Observe that if Alice could produce actions that were \emph{perfectly} i.i.d.\ according to $p_A^{(1)}$ in the first subblock and \emph{perfectly} i.i.d.\ according to $p_A^{(2)}$ in the second subblock, both independent of Bob's observations, then in $\gamma$ fraction of the stages in a block, she would secure the average payoff of
$\min_{b\in \mathcal{B}}\E_{p_A^{(1)}}[u_{A,b}]$ per action, and in the remaining $1-\gamma$ fraction she would secure the payoff of $\min_{b\in \mathcal{B}}\E_{p_A^{(2)}}[u_{A,b}]$. This would give Alice a total payoff of
$$\gamma \min_{b\in \mathcal{B}}\E_{p_A^{(1)}}[u_{A,b}] + (1-\gamma) \min_{b\in \mathcal{B}}\E_{p_A^{(2)}}[u_{A,b}].$$
We will now show that regarding the total variation distance as the measure of accuracy, Alice can play almost i.i.d.\ according to $p_A^{(1)}$ during the first subblock of each block, and almost i.i.d.\ according to $p_A^{(2)}$ during the second subblock of each block. Therefore, Alice can secure payoffs arbitrarily close to $\gamma \min_{b\in \mathcal{B}}\E_{p_A^{(1)}}[u_{A,b}] + (1-\gamma) \min_{b\in \mathcal{B}}\E_{p_A^{(2)}}[u_{A,b}]$ and this will complete the proof.

\setlength{\unitlength}{1cm}
\begin{figure}
\centering
\scalebox{1}[1]{
\begin{picture}(15,4)
\put(0,0){\includegraphics[width=15cm]{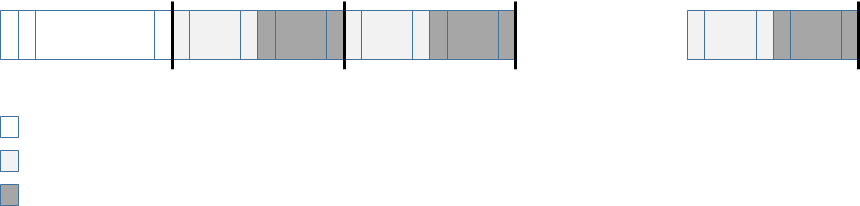}}

\small
\put(.7,0){{Time slots in which Alice draws a random action with law $p_A^{(2)}$}}
\put(.7,.61){Time slots in which Alice draws a random action with law $p_A^{(1)}$}
\put(.7,1.22){Time slots of the first block in which Alice plays the pure action $1\in\mathcal{A}$}
\normalsize

\put(1.5,2){$1$}
\put(4.5,2){$2$}
\put(7.5,2){$3$}
\put(10.5,2){$\dots$}
\put(13.5,2){$N$}

\tiny
\put(.1,3.5){$1$}
\put(.4,3.5){$2$}
\put(1.5,3.5){$\dots$}
\put(2.75,3.5){$L$}

\put(3.1,3.5){$1$}
\put(3.5,3.5){$\dots$}
\put(4,3.5){$\lceil \gamma L \rceil$}
\put(5,3.5){$\dots$}
\put(5.7,3.5){$L$}

\put(6.1,3.5){$1$}
\put(6.5,3.5){$\dots$}
\put(7,3.5){$\lceil \gamma L \rceil$}
\put(8,3.5){$\dots$}
\put(8.7,3.5){$L$}

\put(10.55,3.5){$\dots$}

\put(12.1,3.5){$1$}
\put(12.5,3.5){$\dots$}
\put(13,3.5){$\lceil \gamma L \rceil$}
\put(14,3.5){$\dots$}
\put(14.7,3.5){$L$}
\normalsize
\end{picture}
}
\caption{Illustration of the block Markov strategy}
\label{F:blocks}
\end{figure}
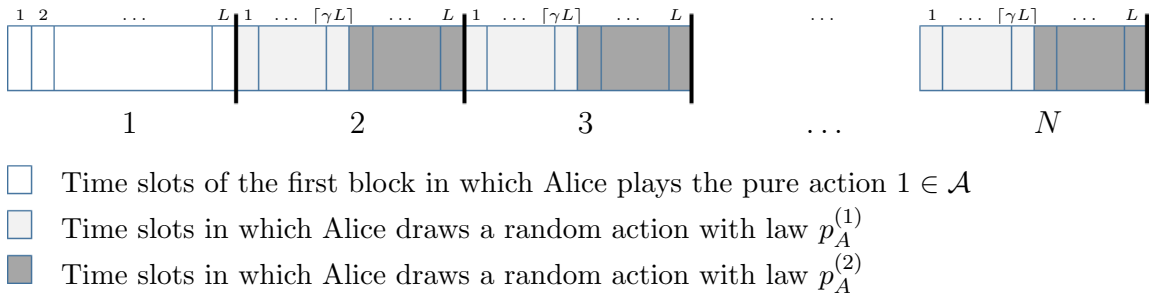

Remember that $\gamma H(p_A^{(1)}) + (1-\gamma)H(p_A^{(2)})< H(X).$ Intuitively speaking, since Alice's observation from a block has entropy $LH(X)$, which is larger than the entropy of her intended action in the next block $L\gamma H(p_A^{(1)}) + L(1-\gamma)H(p_A^{(2)})$, she should be able to find a proper mapping to produce her actions in the current block from her observations in the preceding block. \cite{Gossner2002} propose such a mapping and put effort to prove its correctness. This is where most of the effort is spent. The essential problem here is to simulate a source from another source (here the observations from one block to actions in the next block). This problem is solved in the information theory literature. As stated below in Lemma~\ref{L:han}, which is adopted from \cite[p. 110]{han}, the measure of relevance to compare between the two sources is \emph{inf-entropy} and \emph{sup-entropy}, rather than entropy. However, for i.i.d.\ sources (or concatenation of i.i.d.\ sources), inf-entropy and sup-entropy reduce to the normal Shannon entropy, as discussed in Remark~\ref{remark:iid} below.
\begin{definition}
Sup-entropy and inf-entropy of a random source $Z_1,Z_2,\ldots$ are denoted by $\underline{H}(Z)$ and $\overline{H}(Z)$, respectively, and defined as follows:
$$\underline{H}(Z)=\text{p-}\liminf_{L\to\infty}\frac{1}{L}\log\frac{1}{p_{Z^L}(Z^L)},$$
$$\overline{H}(Z)=\text{p-}\limsup_{L\to\infty}\frac{1}{L}\log\frac{1}{p_{Z^L}(Z^L)},$$
where for a random sequence $\{W_t\}$,
$$\text{p-}\liminf_{t\to\infty}W_t=\sup \left\{\beta|\lim_{t\to\infty}\Pr[W_t<\beta]=0\right\},$$
$$\text{p-}\limsup_{t\to\infty}W_t=\inf \left\{\beta|\lim_{t\to\infty}\Pr[W_t>\beta]=0\right\}.$$
\end{definition}
\begin{remark}\label{remark:iid}
If $Z^L=(Z_1,Z_2,\ldots,Z_L)$ is a concatenation of two sequences of i.i.d.\ random variables on $\mathcal{Z}$ with distribution
$$p_{Z^L}(z^L)=\prod_{i=1}^{\lceil \gamma L\rceil}p^{(1)}(z_i) \prod_{i=\lceil \gamma L\rceil+1}^L p^{(2)}(z_i),$$
where $\lceil a \rceil$ is the smallest integer greater than or equal to $a$, then, we have
\begin{align*}
\lim_{L\to \infty} \frac{1}{L}\log&\frac{1}{p_{Z^L}(Z^L)} = \lim_{L\to \infty} \frac{1}{L}\left(\sum_{i=1}^{\lceil \gamma L\rceil} \log \frac{1}{p^{(1)}(Z_i)}+\sum_{i=\lceil \gamma L\rceil+1}^L \log \frac{1}{p^{(2)}(Z_i)}\right) \\
&=\gamma \sum_{z\in \mathcal{Z}} p^{(1)}(z) \log \frac{1}{p^{(1)}(z)}+(1-\gamma) \sum_{z\in \mathcal{Z}}p^{(2)}(z) \log \frac{1}{p^{(2)}(z)} &\textrm{with probability 1}\\
&= \gamma H(p^{(1)})+(1-\gamma)H(p^{(2)}) &\textrm{with probability 1},
\end{align*}
where the first equality results from the independence of the random variables $Z_1,Z_2,\ldots$ and the second equality follows from the weak law of large numbers. Thus
$$\text{p-}\liminf_{L\to\infty}\frac{1}{L}\log\frac{1}{p_{Z^L}(Z^L)}=\gamma H(p^{(1)})+(1-\gamma)H(p^{(2)}),$$
$$\text{p-}\limsup_{L\to\infty}\frac{1}{L}\log\frac{1}{p_{Z^L}(Z^L)}=\gamma H(p^{(1)})+(1-\gamma)H(p^{(2)}),$$
hence $\underline{H}(Z)=\overline{H}(Z)=\gamma H(p^{(1)})+(1-\gamma)H(p^{(2)})=\lim_{L\to \infty}1/L H(Z^L)$.
\end{remark}
\begin{lemma}[Simulation of a source from another source]\label{L:han}
For each natural number $L$, consider an arbitrary distribution $p^{(L)}_{X^L}$ on sequences $(x_1, x_2, \ldots, x_L)\in\mathcal{X}^L$. Similarly, for each natural number $L$, consider an arbitrary distribution $q^{(L)}_{A^L}$ on sequences $(a_1, a_2, \ldots, a_L)\in\mathcal{A}^L$.
If $\underline{H}(X)>\overline{H}(A)$, then, for each natural number $L$, there exists a mapping $\varphi_L:\mathcal{X}^L\to \mathcal{A}^L$ such that the total variation distance between the distributions of $\varphi_L(X^L)$ and $A^L$ vanishes asymptotically, \emph{i.e.,} $$\lim_{L\to \infty}\| q^{(L)}_{A^L}-p^{(L)}_{\varphi_L(X^L)} \|_{TV}=0.$$
\end{lemma}
The proof of Lemma~\ref{L:han} can be found in \cite[p. 110]{han}.
\begin{remark}
Observe that $\lim_{L\to \infty}\| p_{A^L}-p_{\varphi_L(X^L)} \|_{TV}=0$ in the above lemma implies that $\varphi_L(X^L)$ is almost statistically indistinguishable from $p_{A^L}$ for large values of $L$. In other words, there is \underline{no statistical test} that can distinguish between $p_{A^L}$ and $p_{\varphi_L(X^L)}$ with a non-negligible probability.
\end{remark}
Take some $\epsilon>0$. Since $\gamma H(p_A^{(1)})+(1-\gamma)H(p_A^{(2)})<H(X)$, Lemma~\ref{L:han} and Remark~\ref{remark:iid} imply that there exist mappings $\varphi_L:\mathcal{X}^L\to \mathcal{A}^L$ such that for large $L$, the pmf of $A^L=\varphi_L(X^L)$ is approximated (in total variation distance) as
\begin{align}\left\|p_{A^L}(a^L) - \prod_{t=1}^{\lceil \gamma L\rceil}p_A^{(1)}(a_t) \prod_{t=\lceil \gamma L\rceil+1}^{L}p_A^{(2)}(a_t)\right\|_{TV}\leq \epsilon.\label{eqn:neq1}\end{align}
Therefore, by dividing the $T$ stages into $N$ blocks of $L$ stages, we construct a strategy $\sigma$ for Alice as follows: in each block (excluding the first block) Alice chooses action sequence $A_c^L=\varphi_L(X_p^L)$, where $X_p^L$ is the randomness source observed during the {previous block}; \emph{i.e.,} $A_c^L$ is for the current block, but $X_p^L$ is for the previous block.

The ideal distribution $\prod_{t=1}^{\lceil \gamma L\rceil}p_A^{(1)}(a_t) \prod_{t=\lceil \gamma L\rceil+1}^{L}p_A^{(2)}(a_t)$ gives Alice a payoff of
$$ \lceil \gamma L\rceil \min_{b\in \mathcal{B}}\E_{p_A^{(1)}}[u_{A,b}] + (L-\lceil \gamma L\rceil ) \min_{b\in \mathcal{B}}\E_{p_A^{(2)}}[u_{A,b}].$$
Alice's actual distribution is within $\epsilon$ total variation distance of the ideal distribution. By relating the total variation distance to the payoff differences, we obtain that difference between the payoff under the actual distribution and the ideal one is at most $\epsilon$ times $2L\mathsf M$, where $\mathsf M$ is the maximum absolute value entry of the payoff table ($L\mathsf M$ is the maximum absolute value entry of the $L$ repetitions of the game). Thus, the distance between the average payoff of the block differs by at most $2\epsilon \mathsf M$ from the average payoff under the ideal distribution. This completes the proof.

\begin{remark}
We constructed the strategy $\sigma$ for the case $T=NL$. In general, for $T=NL+\delta$, where $\delta<L$, one can extend the first block to contain $L+\delta$ stages and choose $N$ large enough to diminish the effect of the first block on the average payoff.
\end{remark}

\subsection{Generalized model: leakage of the randomness source}\label{sec:5op}
We generalize the model of the repeated game of Section~\ref{subsec:gossner} and define the repeated game with leaked randomness source as follows: Let $X^T=(X_{1},X_2,\dots,X_{T})$ and $Y^T=(Y_{1},Y_{2},\dots,Y_{T})$ be two sequences of random variables such that $(X_t, Y_t)$, for $t\in\{1,2,\dots,T\}$, are drawn i.i.d.\ from a sample space $\mathcal X\times \mathcal Y$ with law $p_{XY}$. In every stage $t\in\{1,2,\dots,T\}$, Alice and Bob observe respective random sources $X_{t}$ and $Y_{t}$ privately and choose actions $A_t\in \mathcal{A}$ and $B_t\in\mathcal{B}$. Thus the only difference of this model with the model of Section~\ref{subsec:gossner} is that in every stage $t$, Bob observes $Y_t$ which is related to the observation of Alice $X_t$; hence, the history of observations of Bob is $\mathsf H_2^t=(Y_{1},A_1,B_1,\dots,Y_{t-1},A_{t-1},B_{t-1},Y_{t})$ and $B_t=\tau_t(\mathsf H_2^t)$, where $\tau_t:(\mathcal{A}\times\mathcal{B})^{t-1}\times \mathcal{Y}^t\to \mathcal{B}$. As before, the history of observations of Alice is $\mathsf H_1^t=(X_{1},A_1,B_1, \dots,X_{t-1},A_{t-1},B_{t-1},X_{t})$. Note that Alice does not have access to any private sources of randomness except $\mathsf H_1^t$ and hence, she has to use the deterministic function $\sigma_t(\cdot)$, while Bob can utilize the random function $\tau_t(\cdot)$. Alice (respectively Bob) wishes to maximize (respectively minimize) the expected average payoff $\lambda_T(\sigma,\tau)$ defined in Equation \eqref{eq:exp_avg}. Note that if we set $Y_t$ to be constant random variables, the above repeated zero-sum game reduces to the one considered in Section~\ref{subsec:gossner}.

The maxmin value of the repeated game with leaked randomness source is defined as in Definition~\ref{def:max-min value} and characterized as follows:
\begin{theorem}\label{T:maxmin_val}
The maxmin value of the repeated game with leaked randomness source defined in Section~\ref{sec:5op} is $ \mathcal J_{\text{cav}} (H(X|Y))$, where $\mathcal J_{\text{cav}} (.)$ is defined as in Theorem~\ref{T:maxmin_val_gossner}.
\end{theorem}
\color{black}
To prove Theorem~\ref{T:maxmin_val}, in Section~\ref{S:secure}, we show that Alice can secure $ \mathcal J_{\text{cav}} (H(X|Y))$ and in Section~\ref{S:defend}, we show that Bob can defend $ \mathcal J_{\text{cav}} (H(X|Y))$. Therefore, by definition, $ \mathcal J_{\text{cav}} (H(X|Y))$ is the maxmin value of the repeated game. Before going to Sections \ref{S:secure} and \ref{S:defend}, we present a corollary followed by an example.
\color{black}
\begin{corollary}
Consider the case in which besides the sequence $X_1,X_2,\dots$, Alice also observes the sequence $Y_1,Y_2,\dots$; therefore, $X_t'=(X_t,Y_t)$ is the observation of Alice at stage $t$. Since $H(X'|Y)=H(X|Y)$, the maxmin value of the game is still $\mathcal{J}_{\text{cav}}(H(X|Y))$ which is equal to the maxmin value of the game in which Alice does not observe the sequence $Y_1,Y_2,\dots$. Therefore, Alice's access to $Y_1,Y_2,\dots$ does not help her. This parallels the classical result of \cite{slepian1973noiseless} in information theory.
\end{corollary}
\begin{example} (Matching pennies)
Consider the matching pennies game with $\mathcal{A}=\mathcal{B}=\{0,1\}$ and payoffs:
$$u_{0,0}=u_{1,1}=1,\quad u_{0,1}=u_{1,0}=0.$$
Assume that Alice observes (symbol by symbol) an i.i.d.\ sequence of binary random variables $X_1, X_2, \ldots$, where $$p_{X_i}(0)=q,\quad p_{X_i}(1)=1-q.$$
Each $X_i$ is revealed to Bob with probability $\alpha\in[0,1]$. Thus, Bob observes (symbol by symbol) a sequence $Y_1, Y_2, \ldots$ such that with probability $\alpha$, $Y_i=X_i$ and with probability $1-\alpha$, $Y_i=\text{null}$.

Roughly speaking, $\alpha$ fraction of the bits of Alice are revealed to Bob. If Alice is made aware of which of her bits are compromised and leaked to Bob, she can drop the compromised bits and keep the $1-\alpha$ fraction of the secret bits. In this way, she can distill an average of $(1-\alpha)H(q)$ random bits per observation. It is known from \cite{Gossner2002} that for the matching pennies game, $\mathcal{J}_{\text{cav}}(\mathsf{h})=\mathsf{h}/2$. Thus, using $(1-\alpha)H(q)$ bits per stage, Alice can secure payoff $(1-\alpha)H(q)/2$. The above theorem says that Alice can secure the same payoff even without knowing which of her bits are compromised.
\end{example}

\subsubsection{Proof of Theorem \ref{T:maxmin_val}: Alice can secure $ \mathcal J_{\text{cav}} (H(X|Y))$}\label{S:secure}
To show that Alice can secure any payoff less than $\mathcal J_{\text{cav}} (H(X|Y))$, we extend and simplify the proof of \cite{Gossner2002}. As before, the upper concave envelope $\mathcal J_{\text{cav}}(\cdot)$ at $\mathsf{h}=H(X|Y)$ can be expressed as the convex combination
$$\gamma \mathcal J(H(p_A^{(1)}))+(1-\gamma)\mathcal J(H(p_A^{(2)}))$$
for some $\gamma\in[0,1]$ and pmfs $p_A^{(1)}$ and $p_A^{(2)}$ in $\Delta(\mathcal{A})$ where $p_A^{(1)}$ and $p_A^{(2)}$ secure respective payoffs $\mathcal{J}(H(p_A^{(1)}))$ and $\mathcal{J}(H(p_A^{(2)}))$ in the one-shot game, and the following equality is satisfied:
$$\gamma H(p_A^{(1)}) + (1-\gamma)H(p_A^{(2)})=H(X|Y).$$
As a result, it suffices to show that for any $p_A^{(1)}$ and $p_A^{(2)}$ satisfying
$$\gamma H(p_A^{(1)}) + (1-\gamma)H(p_A^{(2)})< H(X|Y),$$
Alice can secure payoffs arbitrarily close to
$$ \gamma \min_{b\in \mathcal{B}}\E_{p_A^{(1)}}[u_{A,b}] + (1-\gamma) \min_{b\in \mathcal{B}}\E_{p_A^{(2)}}[u_{A,b}].$$

Again, the idea is to utilize the block-Markov proof technique. Take some $T$ of the form $T=NL$, and divide the total $T$ stages into $N$ blocks of length $L$. Each block is divided into two subblocks as before. Excluding the first block,
Alice wants to play almost i.i.d.\ according to $p_A^{(1)}$ during the first subblock, and almost i.i.d.\ according to $p_A^{(2)}$ during the second subblock.

By symmetry, we only need to consider the payoff that Alice gets in one of the blocks. For notational simplicity, we denote the observations of Alice and Bob in the \emph{previous} block by $X_p^{L}$ and $Y_p^L$ respectively, and use $A_c^L$ and $B_c^L$ to denote their actions in the \emph{current} block. Instead of Lemma~\ref{L:han} in the previous section, the proof relies on the following proposition whose proof is given in Section~\ref{S:mapping}:
\begin{proposition}\label{Pro:mapping}
Let $(X_p^L, Y_p^L)$ be drawn i.i.d.\ from $p_{XY}(x_p,y_p)$, $\gamma\in [0,1]$ be an arbitrary real number, and $p_A^{(1)}$ and $p_A^{(2)}$ be arbitrary distributions on $\mathcal{A}$ such that
$$\gamma H(p_A^{(1)}) + (1-\gamma)H(p_A^{(2)}) < H(X|Y).$$
Then, for any $\epsilon>0$, there exists a natural number $\tilde L$ and mappings $\psi_L:\mathcal{X}^L\to \mathcal{A}^L$ such that for all $L\geq \tilde L$ and $A_c^L=\psi_L (X_p^L)$,
\begin{align}\left\|p_{A_c^L,Y_p^L}(a_c^L,y_p^L)-p_{Y_p^L}(y_p^L) \prod_{t=1}^{\lceil \gamma L\rceil}p_A^{(1)}(a_{c,t}) \prod_{t=\lceil \gamma L\rceil+1}^{L}p_A^{(2)}(a_{c,t}) \right\|_{TV} <\epsilon.\label{eqn:propo1}\end{align}
\end{proposition}
Note that $p_{Y_p^L}(y_p^L) \prod_{t=1}^{\lceil \gamma L\rceil}p_A^{(1)}(a_{c,t}) \prod_{t=\lceil \gamma L\rceil+1}^{L}p_A^{(2)}(a_{c,t})$ is the distribution of the ideal independent actions desired by Alice, whereas $p_{A_c^L,Y_p^L}(a_c^L,y_p^L)$ is the real joint distribution.

Let Alice use the actions $A_c^L=(A_{c,1}, A_{c,2}, \ldots, A_{c,L})$ in the current block, where $A_c^L=\psi_L (X_p^L)$ and $\psi_L$ is the mapping of Proposition~\ref{Pro:mapping}. Because the $X$-source is i.i.d.,\ the $X$-source for different blocks are independent, and random variable $A_c^L$ (a function of $X_p^L$) is independent of Alice's action in all the previous blocks. As a result, even though Bob has access to the entire past history of the game and his $Y$-source observations, he obtains information about $X_p^L$ only through his source $Y_p^L$ and Alice's prior actions in the current block. In other words, $B_{c,t}$, Bob's action at the $t$-th stage of the current block, is conditionally independent of $X_p^L$ given $Y_p^L, A_c^{t-1}, B_c^{t-1}$. Since $A_c^L=\psi_L (X_p^L)$, $B_{c,t}$ is also conditionally independent of $A_{c,t}, A_{c,t+1}, \cdots, A_{c,L}$ given $Y_p^L, A_c^{t-1}, B_c^{t-1}$. Thus,
$$p(b_c^L|y_p^L, a_c^L)=\prod_{t=1}^L p(b_{c,t}|y_p^L, a_c^{t-1}, b_c^{t-1}).$$
Then, utilizing the first property of total variation in Lemma~\ref{lemma:tv_p} for random variables $E=(A_c^L,Y_p^L)$ and $F=B_c^L$ we conclude from \eqref{eqn:propo1} that
\begin{align}\Big\|p&_{A_c^L,Y_p^L}(a_c^L,y_p^L)\prod_{t=1}^Lp(b_{c,t}|y_p^L, a_c^{t-1}, b_c^{t-1})\label{eqn:propo2}
\\&-p_{Y_p^L}(y_p^L) \prod_{t=1}^{\lceil \gamma L\rceil}p_A^{(1)}(a_{c,t})p(b_{c,t}|y_p^L, a_c^{t-1}, b_c^{t-1}) \prod_{t=\lceil \gamma L\rceil+1}^{L}p_A^{(2)}(a_{c,t})p(b_{c,t}|y_p^L, a_c^{t-1}, b_c^{t-1}) \Big\|_{TV} <\epsilon.\nonumber\end{align}
Utilizing the second property of total variation in Lemma~\ref{lemma:tv_p} for random variables $E=(A_c^L,B_c^L)$ and $F=Y_p^L$, we conclude that the distance between the actual actions $p_{A_c^L,B_c^L}(a_c^L,b_c^L)$ and the ideal one is less than or equal to $\epsilon$, \emph{i.e.,}
\begin{align}\Big\|p_{A_c^L,B_c^L}(a_c^L,b_c^L)-\prod_{t=1}^{\lceil \gamma L\rceil}p_A^{(1)}(a_{c,t})p(b_{c,t}|a_c^{t-1}, b_c^{t-1}) \prod_{t=\lceil \gamma L\rceil+1}^{L}p_A^{(2)}(a_{c,t})p(b_{c,t}|a_c^{t-1}, b_c^{t-1}) \Big\|_{TV} <\epsilon.\label{eqn:propo3}\end{align}
As before, by relating the total variation distance to the payoff differences, we obtain that the average payoff differs by at most $2\epsilon \mathsf{M}$ from the average payoff under the ideal distribution. This completes the proof.

\subsubsection{Proof of Theorem \ref{T:maxmin_val}: Bob can defend $ \mathcal J_{\text{cav}} (H(X|Y))$} \label{S:defend}
This is an extension of the proof given in \cite{Neyman2000, Gossner2002}. Let $\sigma$ be an arbitrary strategy for Alice and generate strategy $\tau$ for Bob as follows: given $h_2^t$, an arbitrary history of observations of Bob until stage $t$, $\tau_t(h_2^t)$ is the best choice of Bob that minimizes the expected payoff at stage $t$, \emph{i.e.,}
$$ \tau_t(h_2^t) \in \underset{b\in \mathcal{B}}{\arg \min} \E_{\sigma} \left[u_{A_t,b}|\mathsf{H}_2^t=h_2^t\right],$$
where $\E_{\sigma}$ denotes the expectation with respect to the probability distribution induced by $\sigma$ and i.i.d.\ repetitions of $p_{X,Y}$. Note that conditional on the observation $h_2^t$ of Bob, Alice's action $A_t$ has entropy $H(A_t|\mathsf{H}_2^t=h_2^t)$, thus,
\begin{equation} \label{E:max_payoff}
\E_{\sigma,\tau}\left[u_{A_t,B_t}|\mathsf{H}_2^t=h_2^t\right] \leq \mathcal{J}\left(H(A_t|\mathsf{H}_2^t=h_2^t)\right)\leq \mathcal{J}_{\text{cav}}\left(H(A_t|\mathsf{H}_2^t=h_2^t)\right).
\end{equation}
Therefore, we have
\begin{align}
\E_{\sigma,\tau}[u_{A_t,B_t}]
& = \sum_{h_2^t\in (\mathcal{A}\times\mathcal{B})^{t-1}\times \mathcal{Y}^t} \Pr\left[\mathsf{H}_2^t=h_2^t\right]\E_{\sigma,\tau}\left[u_{A_t,B_t}|\mathsf{H}_2^t=h_2^t\right] \notag \\
&\leq \sum_{h_2^t\in (\mathcal{A}\times\mathcal{B})^{t-1}\times \mathcal{Y}^t} \Pr\left[\mathsf{H}_2^t=h_2^t\right]\mathcal{J}_{\text{cav}}\left(H(A_t|\mathsf{H}_2^t=h_2^t)\right) \notag \\
& \leq \mathcal{J}_{\text{cav}}\left(\sum_{h_2^t\in (\mathcal{A}\times\mathcal{B})^{t-1}\times \mathcal{Y}^t} \Pr\left[\mathsf{H}_2^t=h_2^t\right] H(A_t|\mathsf{H}_2^t=h_2^t)\right) \notag \\
&= \mathcal{J}_{\text{cav}}\left(H(A_t|\mathsf{H}_2^t)\right), \label{E:jensen1}
\end{align}
where the second inequality is implied by applying Jensen's inequality to concave function $\mathcal{J}_{\text{cav}}(\cdot)$. By definition of $\lambda_T(\sigma,\tau)$ and using \eqref{E:jensen1} we have
\begin{align}
\lambda_T(\sigma,\tau)
& \leq \sum_{t=1}^T \frac{1}{T} \mathcal{J}_{\text{cav}}\left(H(A_t|\mathsf{H}_2^t)\right) \notag \\
& \leq \mathcal{J}_{\text{cav}}\left(\sum_{t=1}^T \frac{1}{T} H(A_t|\mathsf{H}_2^t)\right) \label{second_inequality} \\
& =\mathcal{J}_{\text{cav}}\left(\sum_{t=1}^T \frac{1}{T} H(A_t|Y^t,A^{t-1},B^{t-1})\right) \notag \\
& =\mathcal{J}_{\text{cav}}\left(\sum_{t=1}^T \frac{1}{T} H(A_t|Y^T,A^{t-1})\right) \label{second_equality} \\
& =\mathcal{J}_{\text{cav}}\left(\frac{1}{T} H(A^T|Y^T)\right) \notag \\
& \leq \mathcal{J}_{\text{cav}}\left(\frac{1}{T} H(X^T|Y^T)\right) \label{final_inequality} \\
& = \mathcal{J}_{\text{cav}}\left( H(X|Y)\right), \notag
\end{align}
where \eqref{second_inequality} follows from applying Jensen's inequality to concave function $\mathcal{J}_{\text{cav}}$. Note that given the strategy $\tau$ for Bob, $B^{t-1}$ is a deterministic function of $Y^t$ and $A^{t-1}$, thus $ H(A_t|Y^t,A^{t-1},B^{t-1})= H(A_t|Y^t,A^{t-1})$; furthermore, $A_t$ and $B_t$ are causally generated from i.i.d.\ sequences $X^T$ and $Y^T$, so $(A_t,A^{t-1},Y^t)$ is independent of $(Y_{t+1},Y_{t+2},\dots,Y_{T})$ and hence $ H(A_t|Y^t,A^{t-1})= H(A_t|Y^T,A^{t-1})$, which implies \eqref{second_equality}. To prove \eqref{final_inequality}, consider that given strategies of the players and the sequence $Y^T$, $A^T$ is a deterministic function of $X^T$ and hence $H(X^T|Y^T)\geq H(A^T|Y^T)$; therefore, since $\mathcal{J}_{\text{cav}}(\cdot)$ is an increasing function, \eqref{final_inequality} follows.

Thus, for all strategy $\sigma$ of Alice there exists a strategy $\tau$ for Bob such that $\limsup_{T\to \infty} \lambda_T(\sigma,\tau) \leq \mathcal{J}_{\text{cav}}\left( H(X|Y)\right)$. Hence, Bob can defend $\mathcal{J}_{\text{cav}}\left( H(X|Y)\right)$.

\section{Secret correlation in repeated games with imperfect monitoring}\label{sec:imperfect_monitoring}
In this section, we revisit the repeated game with imperfect monitoring studied by \cite{Gossner_Tomala}. The main contribution of this part is to simplify the proof of \cite{Gossner_Tomala} and generalize their results.
\subsection{Problem definition}
Consider a $T$ stage repeated zero-sum game between team $A$ and player $B$, where team $A$ consists of $m$ players $\{1,2,\ldots,m\}$ with respective finite actions sets $\mathcal{A}_1, \mathcal{A}_2,\ldots,\mathcal{A}_m$. Let $\mathcal{A}=\mathcal{A}_1\times \mathcal{A}_2\times\cdots\times \mathcal{A}_m$ and $\mathcal{B}$ denote the finite actions sets of team $A$ and player $B$, respectively. In every stage $t\in\{1,2,\dots,T\}$, all players choose an action from their corresponding sets of actions. Let $\bold{A}_t$ denote the action profile of team $A$ and $B_t$ denote the action of player $B$ at stage $t$. At the end of stage $t$, team $A$ gets payoff $u_{\bold{A}_t,B_t}$ from player $B$ and all players of team $A$ observe the chosen actions $\bold{A}_t$ and $B_t$, while player $B$ observes $B_t$ and $S_t\in \mathcal S$, where $S_t$ is a noisy version of $\bold{A}_t$ with conditional distribution $\Pr[S_t=s_t|\bold{A}_t=\bold a_t]=p_{S|\bold{A}}(s_t|\bold a_t)$. In order to choose the actions of stage $t$, players make use of the history of their observations until stage $t$, which is denoted by $\mathsf H_1^t=(\bold{A}^{t-1},B^{t-1})$ for players of team $A$ and $\mathsf H_2^t=(B^{t-1},S^{t-1})$ for player $B$. Let $\sigma_{i,t}:(\mathcal A\times \mathcal B)^{t-1}\to \mathcal{A}_i$ and $\tau_t:(\mathcal B \times \mathcal S)^{t-1}\to \mathcal B$ be the \emph{random functions} mapping the history of observations of arbitrary player $i$ in team $A$ and player $B$ to their actions at stage $t$, thus, $\bold{A}_t=(\sigma_{1,t}(\mathsf H_1^t),\sigma_{2,t}(\mathsf H_1^t),\dots,\sigma_{m,t}(\mathsf H_1^t))$ and $B_t=\tau_t(\mathsf H_2^t)$. Let $\sigma_t=(\sigma_{1,t},\sigma_{2,t},\dots,\sigma_{m,t})$. We call the $T$-tuples $\sigma=(\sigma_1,\sigma_2,\dots,\sigma_T)$ and $\tau=(\tau_1,\tau_2,\dots,\tau_T)$ the strategies of team $A$ and player $B$, respectively. The expected average payoff for team $A$, $\lambda_T(\sigma,\tau)$, and the maxmin value of the game are defined in a similar way as in Section~\ref{subsec:gossner}.
\begin{remark}
In the definition of the repeated game with imperfect monitoring, \cite{Gossner_Tomala} assumed that the signals $S_t$ were also observed by the players of team $A$, whereas we assume that the signals $S_t$ are \underline{not} observed by the players of team $A$. We generalize the result of \cite{Gossner_Tomala} by showing that the maxmin value of the game remains the same with this change in assumption.
\end{remark}
\begin{definition} \label{def:Q} Let $\Pi$ be a subset of $\Delta(\mathcal A)$ containing the distributions of independent random actions on $\mathcal{A}=\mathcal{A}_1\times \mathcal{A}_2\times\cdots\times \mathcal{A}_m$, \emph{i.e.,}
$$\Pi=\left\{p_{\bold A}=p_{A_1}p_{A_2}\dots p_{A_m}|p_{A_i}\in \Delta(\mathcal A_i), \forall i=1,\dots,m\right\}.$$
Furthermore, given random variables $R\in \mathcal R$ and $\bold{A}\in \mathcal{A}$ with joint distribution $p_{\bold AR}$, \color{black} functions $\pi(\bold{A}|R=r)$ (the security level of $\bold{A}$ given $R=r$) and $\pi(\bold{A}|R)$ (the security level of $\bold{A}$ given $R$) are defined as follows:
$$\pi(\bold{A}|R=r)=\min_{b\in\mathcal B}\E[u_{\bold A,b}|R=r]$$
and $$\pi(\bold A|R)=\sum_{r\in\mathcal R}p_R(r)\pi(\bold{A}|R=r).$$
Observe that $\pi(\bold{A}|R)$ is in terms of an average over values of $R$ and we have chosen the notation $\pi(\bold{A}|R)$ as its definition resembles the way conditional entropy is defined in information theory.
\end{definition}
\color{black}

\begin{theorem}\label{theorem:secret_corr}
The maxmin value of the repeated game with imperfect monitoring is
$$w=\max \pi(\bold{A}|R),$$
where the maximization is over all random variables $\bold{A}\in \mathcal A$, $S \in \mathcal S$, $R\in \mathcal R =\{0,1\}$ and $Q\in \mathcal Q=\{1,2,3,\cdots, 2|\mathcal{A}|\}$, satisfying
\begin{align}
p_{S\bold ARQ}(s,\bold a,r,q)&=p_{RQ}(r,q)p_{\bold A|Q}(\bold a|q)p_{S|\bold A}(s|\bold a),\label{eqnE1}\\
p_{\bold A|Q}(\bold a|q)&\in\Pi,\label{eqnE2}\\
H(Q\bold A|SR)&\geq H(Q|R).\label{eqnE3}
\end{align}
The set $\Pi$ and the function $\pi(\bold{A}|R)$ are defined in Definition~\ref{def:Q} and $p_{S|\bold A}$ is the fixed conditional distribution that generates the monitoring signals $S_t$ of the repeated game with imperfect monitoring.
\end{theorem}
\begin{remark}
The statement of the above theorem has a different (but equivalent) form than the one given in \cite{Gossner_Tomala}. In particular, \cite{Gossner_Tomala} have expressed the solution as an optimization problem over the set of ``distributions of distributions". \color{black} The computational aspects of the solution given by \cite{Gossner_Tomala} has been studied by \cite{Gossner_Tomala2}. \color{black}In our characterization, because all the variables have finite alphabet sets, the above form is {computable} and is expressed as a maximization problem over a compact and bounded set of probability distribution.
\end{remark}

In Section~\ref{sec:monitoring_secure}, we show that team $A$ can secure $w$, and in Section~\ref{sec:defend_w}, we show that player $B$ can defend $w$ even when players of team $A$, in addition to $\bold A_t$ and $B_t$, observe the signals $S_t$. Thus, $w$ is the maxmin value of the repeated game with imperfect monitoring regardless of whether team $A$ observe $S_t$ or not.

\subsection{Proof of Theorem \ref{theorem:secret_corr}: Payoff secured by team $A$} \label{sec:monitoring_secure}
\color{black}We begin with investigating the special case in which player $B$ perfectly monitors the actions of team $A$; \emph{i.e.,} there exists a deterministic function $f:\mathcal S\to \mathcal A$ such that for all stages $t$, we have $\bold A_t=f(S_t)$. Let $\bold{A},S,R,Q$ be arbitrary random variables with joint pmf $p_{RQ\bold AS}$ in the feasible set of the maximization in the statement of the theorem. Considering that $\bold A=f(S)$, we have:
$$H(Q\bold A|RS)-H(Q|R)=H(Q|RS)-H(Q|R)\geq 0,$$
where the inequality follows from \eqref{eqnE3}. On the other hand, we have $H(Q|RS)\leq H(Q|R)$. Hence, $H(Q|RS)= H(Q|R)$. Since $H(Q|RS)-H(Q|R)=H(S|RQ)-H(S|R)$, we obtain that $H(S|RQ)=H(S|R)$. It is known that $H(S|RQ)\leq H(S|R)$ (conditioning reduces entropy), and equality $H(S|RQ)=H(S|R)$ holds only if for all $r\in \mathcal R$ and $q\in \mathcal Q$ such that $p_{RQ}(r,q)>0$, we have $p_{S|R=r,Q=q}=p_{S|R=r}$. 
 Since $\bold A$ is a deterministic function of $S$, we also conclude that for all $r\in \mathcal R$ and $q\in \mathcal Q$ such that $p_{RQ}(r,q)>0$, we have $p_{\bold A|R=r,Q=q}=p_{\bold A|R=r}$. On the other hand, Equations \eqref{eqnE1} and \eqref{eqnE2} imply that $p_{\bold A|R=r,Q=q}=p_{\bold A|Q=q}\in \Pi$; thus, for all $r$ with positive probability, we have $p_{\bold A|R=r}\in \Pi$.

Let $T$ be the total number of the stages of the game. Construct a strategy $\sigma$ for team $A$ as follows: in the first $\lceil p_R(0)T\rceil$ stages, the team players play i.i.d.\ according to $p_{\bold A|R=0}$, and in the remaining $T-\lceil p_R(0)T\rceil$ stages, they play i.i.d.\ according to $p_{\bold A|R=1}$. Note that since for all $r$ with positive probability, $p_{\bold A|R=r}$ belongs to $\Pi$, the team players can implement the above strategy distributively without the need for shared randomness. When $T$ tends to infinity, the expected average payoff secured by $\sigma$ converges to $\pi(\bold A|R)$, and hence, they can secure $\pi(\bold A|R)$.

Next, we assume that player $B$ monitors the actions of the team players imperfectly, \emph{i.e.,} Bob is unable to compute $\bold{A}$ as a deterministic function of $S$. \color{black} Let $\bold{A}\in \mathcal A$, $S \in \mathcal S$, $R\in \mathcal R =\{0,1\}$ and $Q\in \mathcal Q=\{1,2,3,\cdots, 2|\mathcal{A}|\}$ be arbitrary random variables with joint pmf $p_{RQ\bold AS}$ satisfying \eqref{eqnE1}, \eqref{eqnE2} and the strict form of \eqref{eqnE3}, \emph{i.e.,}
\begin{equation} \label{eq:cond1}
H(Q\bold A|RS)>H(Q|R).
\end{equation}
The boundary case, $H(Q\bold A|RS)=H(Q|R)$, will be addressed in Remark~\ref{rem:boundary_case}. We now show that team $A$ can secure payoff $\pi(\bold{A}|R)$. Since
$$\pi(\bold A|R) = p_R(0) \min_{b\in \mathcal B} \E[u_{\bold A,b}|R=0] + p_R(1) \min_{b\in \mathcal B} \E [u_{\bold A,b}|R=1],$$
 it suffices to show that players of team $A$ can choose their actions almost i.i.d.\ according to $p_{\bold A|R}(\bold a|0)$ in $\gamma = p_R(0)$ portion of stages and almost i.i.d\ according to $p_{\bold A|R}(\bold a|1)$ in $1-\gamma = p_R(1)$ portion of stages, while the action of each stage is almost independent of the history of observations of player $B$ until that stage.

In the rest of the proof, the ideal joint distribution $p_{RQ\bold AS}$ is assumed to be given and fixed. Random variable $\bold A$ should not be confused with $\bold A_1$ or $\bold A^T=(\bold A_1, \bold A_2, \cdots, \bold A_T)$ which denote the action profile of players of team $A$ at the first stage and the $T$ time instances, respectively.

\textbf{The block-Markov technique:}
As in the previous section, team $A$ utilizes the block-Markov strategy. We assume that the game is played over
one block of length $k L$ (the first block) followed by $N$ blocks of length $L$. Therefore, the total number of repetitions of the game is $T=k L + NL$. The first block is of length $kL$, which will be specified later; this block is sufficiently long block to provide enough randomness to initialize the block-Markov strategy for the $N$ blocks of length $L$. Excluding the first block, each block is divided into two subblocks of length $\lceil \gamma L\rceil$ and $L-\lceil \gamma L\rceil$. Team $A$ aims to play almost i.i.d.\ according to $p_{\bold A|R}(\bold a|0)$ in the first subblock and according to $p_{\bold A|R}(\bold a|1)$ in the second subblock, while the action of each stage is almost independent of the history of observations of player $B$ until that stage. To do so, in each block, the players of team $A$ extract a sequence of shared randomness, which is almost independent of the observations of player $B$, to correlate their actions in the next block. For arbitrary block $i\geq3$ (the \emph{current} block), let $Q_c^L \in \mathcal Q^L$ denote the sequence of shared randomness extracted in previous block and $\bold A_c^L$ denote the sequence of actions of team $A$ played in the current block. Players of team $A$ produce their actions in the current block, $\bold A_{c}^L$, only from the shared randomness $Q_c^L$ and their private randomness in the manner that is described in details below. With team $A$ playing $\bold A_{c,t}$ in the $t$-th stage of the current block, player $B$ gets the signal $S_{c,t}$ constructed from $\bold A_{c,t}$ with the conditional distribution $p_{S|\bold A}$.

For simplicity of the notation, the sequence of shared randomness, the actions of team $A$, the actions of player $B$ and the signals observed by player $B$ are denoted by $Q_c^L,\bold A_c^L,B_c^L,S_c^L$ for the current block ($i$-th block) and by $Q_p^L,\bold A_p^L,B_p^L,S_p^L$ for the previous block ($(i-1)$-th block), where $i\geq 3$. Similarly, the sequence of shared randomness, the actions of team $A$, the actions of player $B$ and the signals observed by player $B$ in the second block are denoted by $Q_s^L,\bold A_s^L,B_s^L,S_s^L$. For the first block we use $\bold A_f^{k L},B_f^{k L}$ and $S_f^{k L}$ to denote these random variables.

\textbf{First block:}
In the first block, players of team $A$ start off without any shared randomness. They choose their actions independently and i.i.d.\ according to some $p_{\bold A}^{(0)}\in \Pi$ satisfying $H(\bold A^{(0)}|S^{(0)})>0$, where $(\bold A^{(0)},S^{(0)})$ are some random variables with joint pmf $p_{\bold AS}^{(0)}(\bold a,s)=p_{\bold A}^{(0)}(\bold a)p_{S|\bold A}(s|\bold a)$. Note that the distribution $p_{\bold A}^{(0)}$ with the above specifications exists because player $B$ does not have a perfect monitoring of the actions of the team players. Furthermore, since $p_{\bold A}^{(0)} \in \Pi$ and the players of team $A$ have access to private randomness, they can implement it distributively. Thus,
$$p_{\bold A_f^{kL},S_f^{kL}}(\bold a_f^{kL},s_f^{kL})=\prod_{t=1}^{kL} p_{\bold A}^{(0)}(\bold a_{f,t})p_{S|A}(s_{f,t}|\bold a_{f,t}).$$

The length of the first block is $kL$, where $L$ will be specified later; here $k$ is a natural number satisfying
\begin{equation}\label{eq:cond2}
k>\frac{H(Q|R)}{H(\bold A^{(0)}|S^{(0)})}.
\end{equation}

\textbf{Second block:} Equation \eqref{eq:cond2} implies that $kH(\bold A^{(0)}|S^{(0)})>H(Q|R)$. Therefore, according to Proposition~\ref{Pro:mapping}, for arbitrary $\epsilon>0$, there exist mappings $\varphi_L:\mathcal A^{kL}\to \mathcal{Q}^L$ such that if we take $Q_s^L=\varphi_L(\bold A_f^{kL})$, then, for sufficiently large $L$ we have:
\begin{equation}\label{eq:second_block}
\left|\left|p_{Q_s^L,S_f^{kL}}(q_s^L,s_f^{kL})-p_{S_f^{kL}}(s_f^{kL})\prod_{t=1}^{\lceil\gamma L \rceil} p_{Q|R}(q_{s,t}|0)\prod_{t=\lceil\gamma L \rceil+1}^L p_{Q|R}(q_{s,t}|1)\right|\right|_{TV}\leq\epsilon.
\end{equation}
This implies that each of the players of team $A$ can extract shared randomness $Q_s^{L}$ which is almost independent of the observation of player B in the first block.
Let $\mathsf H_{2,s}=(S_f^{kL},B_f^{kl})$ denote the history of observations of player $B$ before starting the second block. Since $B_f^{kL}$ is produced locally by player B from $S_f^{kL}$, random variable $B_f^{kL}$ is conditionally independent of $Q_s^L$ given $S_f^{kL}$. Thus using the first property of total variation in Lemma~\ref{lemma:tv_p} for random variables $E=(S_f^{kL},Q_s^L)$ and $F=B_f^{kL}$ we have
\begin{equation}\label{eq:init}
\left|\left|p_{Q_s^L,\mathsf H_{2,s}}(q_s^L,h_{2,s})-p_{\mathsf H_{2,s}}(h_{2,s})\prod_{t=1}^{\lceil\gamma L \rceil} p_{Q|R}(q_{s,t}|0)\prod_{t=\lceil\gamma L \rceil+1}^L p_{Q|R}(q_{s,t}|1)\right|\right|_{TV}\leq\epsilon.
\end{equation}
In the second block, the players of team $A$ (all of whom know $Q_s^L$) choose $\bold A_{s,t}$ with distribution $p_{\bold A_{s,t}|Q_{s,t}}=p_{\bold A|Q}$ from $\mathcal A$, that is
$$p_{\bold A_s^L|Q_s^L,\mathsf H_{2,s}}(\bold a_s^L|q_s^L,h_{2,s})=\prod_{t=1}^Lp_{\bold A|Q}(\bold a_{s,t}|q_{s,t}).$$
 Because $p_{\bold A|Q}\in \Pi$, actions of players of team $A$ are mutually independent given $Q_s^L$ and they can realize it using private randomness.
Player $B$ gets signal $S_{s,t}$ constructed from $\bold A_{s,t}$ with the conditional distribution $p_{S|\bold A}$, \emph{i.e.,}
$$p_{S_s^L|\bold A_s^L,Q_s^L,\mathsf H_{2,s}}(s_s^L|\bold a_s^L,q_s^L,h_{2,s})=\prod_{t=1}^Lp_{S|\bold A}(s_{s,t}|\bold a_{s,t}).$$

\textbf{In the $i$-th block for $i\geq 3$:} Let $Q_p^L,\bold A_p^L,B_p^L,S_p^L$ be the variables of the previous block, \emph{i.e.,} the $(i-1)$-th block. Note that $Q_p^L,\bold A_p^L$ is available to all players of team $A$. The idea is to construct the current block's shared randomness $Q_c^L$ from $(Q_p^L,\bold A_p^L)$ of the previous block. To do this, we need to find a suitable mapping $\psi_L: \mathcal Q^L\times \mathcal A^L \to \mathcal Q^L$ and set $Q_c^L=\psi_L(Q_p^L,\bold A_p^L)$. Once the shared randomness $Q_c^L$ is constructed in the current block, players of team $A$ construct their actions in the current block solely based on $Q_c^L$ using the conditional distribution
$p_{\bold A_{c,t}|Q_{c,t}}=p_{\bold A|Q}$
as follows:
$$ p_{\bold A_c^L|Q_c^L}(\bold a_c^L|q_c^L)=\prod_{t=1}^L p_{\bold A|Q}(\bold a_{c,t}|q_{c,t}).$$
In other words, while players of team $A$ observe the entire past history of the actions, their actions in stage $t$ of the current block depends only on $Q_{c,t}$. As before, because $p_{\bold A|Q}\in\Pi$, distributed implementation of this conditional distribution is feasible with the private randomness of the team players.

Considering \eqref{eq:init} as the induction basis, suppose that in the previous block ($(i-1)$-th block) we have
\begin{equation}\label{eq:asmp}
\left|\left|p_{Q_p^L,\mathsf H_{2,p}}(q_p^L,h_{2,p})-p_{\mathsf H_{2,p}}(h_{2,p})\prod_{t=1}^{\lceil\gamma L \rceil} p_{Q|R}(q_{p,t}|0)\prod_{t=\lceil\gamma L \rceil+1}^L p_{Q|R}(q_{p,t}|1)\right|\right|_{TV}\leq3^{(i-3)}\epsilon,
\end{equation}
where $\mathsf H_{2,p}$ is the observation of player $B$ before starting the previous block. Our first goal is to identify an appropriate mapping $\psi_L: \mathcal Q^L\times \mathcal A^L \to \mathcal Q^L$ to construct the shared randomness of current block $Q_c^L$ such that
\begin{equation}\label{eq:asmpN}
\left|\left|p_{Q_c^L,\mathsf H_{2,c}}(q_c^L,h_{2,c})-p_{\mathsf H_{2,c}}(h_{2,c})\prod_{t=1}^{\lceil\gamma L \rceil} p_{Q|R}(q_{c,t}|0)\prod_{t=\lceil\gamma L \rceil+1}^L p_{Q|R}(q_{c,t}|1)\right|\right|_{TV}\leq3^{(i-2)}\epsilon,
\end{equation}
where $\mathsf H_{2,c}=(\mathsf H_{2,p},S_p^L,B_p^L)$ is the observation of player $B$ before starting the current block.

Random variables $\mathbf A_p^L$ and $S_p^L$ are produced in the previous block as follows:
$$p_{\bold A_p^L,S_p^L|Q_p^L}(\bold a_p^L, s_p^L|q_p^L)=\prod_{t=1}^Lp_{\bold A|Q}(\bold a_{p,t}|q_{p,t})p_{S|\bold A}(s_{p,t}|\bold a_{p,t}).$$
Using the first property of total variation in Lemma~\ref{lemma:tv_p} for $E=(Q_p^L,\mathsf H_{2,p})$ and $F=(\mathbf A_p^L,S_p^L)$ along with Equation \eqref{eq:asmp} we obtain
\begin{align}\label{eq:asmp2}
\Bigg\|p_{Q_p^L,\mathbf A_p^L,S_p^L,\mathsf H_{2,p}}&(q_p^L,\mathbf a_p^L,s_p^L,h_{2,p})-p_{\mathsf H_{2,p}}(h_{2,p})\times\notag\\
&\prod_{t=1}^{\lceil\gamma L \rceil} p_{Q\bold AS|R}(q_{p,t},\mathbf a_{p,t},s_{p,t}|0)\prod_{t=\lceil\gamma L \rceil+1}^L p_{Q\bold AS|R}(q_{p,t},\mathbf a_{p,t},s_{p,t}|1)\Bigg\|_{TV}\leq3^{(i-3)}\epsilon.
\end{align}

We use the following proposition which is a generalization of Proposition~\ref{Pro:mapping}.
\begin{proposition}\label{Pro:mapping2}
Let $(U,X,Y)\in \{0,1\}\times\mathcal X \times \mathcal Y$ and $(V,Z)\in \{0,1\}\times\mathcal Z$ be arbitrary random variables with respective distributions $p_{UXY}$ and $p_{VZ}$ with finite supports such that
$$H(X|YU)>H(Z|V).$$
We also define random variables $(X_1, \cdots, X_L, Y_1, \cdots, Y_L, W)$ with the joint distribution $p_{X^LY^LW}$ such that
$$\left|\left|p_{X^LY^LW}(x^L,y^L,w)-p_W(w)\prod_{t=1}^{\lceil p_U(0)L\rceil} p_{XY|U}(x_t,y_t|0)\prod_{t=\lceil p_U(0) L\rceil+1}^{L} p_{XY|U}(x_t,y_t|1)\right|\right|_{TV}\leq \delta_1.$$
Then, for arbitrary $\delta_2>0$, there exist mappings $\psi_L:\mathcal X^L\to \mathcal Z^L$ and natural number $\bar{L}$ such that for all $L\geq \bar{L}$ and $Z^L=\psi_L(X^L)$ we have
$$\left|\left| p_{Z^LY^LW}(z^L,y^L,w)-p_W(w)p_{Y^L}(y^L)\prod_{t=1}^{\lceil p_V(0) L \rceil}p_{Z|V}(z_t|0)\prod_{t=\lceil p_V(0) L \rceil+1}^{L}p_{Z|V}(z_t|1)\right|\right|_{TV}\leq2\delta_1+\delta_2.$$
\end{proposition}
\begin{proof}
The proof of Proposition~\ref{Pro:mapping2} is provided in Section~\ref{S:mapping2}.
\end{proof}

We use Proposition~\ref{Pro:mapping2} with the choice of $X=(Q,\bold A)$, $Y=S$, $U=R$, $Z=Q$, $V=R$, $X^L=(Q_p^L,\bold A_p^L)$, $Y^L=S_{p}^L$, $Z^L=Q_c^L$ and $W=\mathsf H_{2,p}$. The assumptions of Proposition~\ref{Pro:mapping2} are satisfied by Equations \eqref{eq:cond1} and \eqref{eq:asmp2}. Therefore, we obtain that for sufficiently large $L$, there exists a mapping $\psi_L: \mathcal Q^L\times \mathcal A^L \to \mathcal Q^L$ such that for $Q_c^L=\psi_L(Q_p^L,\bold A_p^L)$ we have
\begin{align}\label{eq:current_block}
&\left\|p_{Q_c^L,S_p^L,\mathsf H_{2,p}}(q_c^L,s_p^L,h_{2,p})-p_{S_p^L}(s_p^L)p_{\mathsf H_{2,p}}(h_{2,p})\prod_{t=1}^{\lceil\gamma L \rceil} p_{Q|R}(q_{c,t}|0)\prod_{t=\lceil\gamma L \rceil+1}^L p_{Q|R}(q_{c,t}|1)\right\|_{TV}\notag\\
&\qquad\qquad\qquad\qquad\qquad\qquad\qquad\qquad\qquad\qquad\leq2\times 3^{(i-3)}\epsilon+\epsilon \leq 3^{(i-2)}\epsilon.
\end{align}
Since $B_p^L$ is conditionally independent of $Q_c^L$ given $(\mathsf H_{2,p},S_p^L)$, using the first property of total variation in Lemma~\ref{lemma:tv_p} for $E=(Q_c^L,S_p^L,\mathsf H_{2,p})$ and $F=B_p^L$ we have
\begin{equation}\label{eq:sentence}
\left|\left|p_{Q_c^L,\mathsf H_{2,c}}(q_c^L,h_{2,c})-p_{\mathsf H_{2,c}}(h_{2,c})\prod_{t=1}^{\lceil\gamma L \rceil} p_{Q|R}(q_{c,t}|0)\prod_{t=\lceil\gamma L \rceil+1}^L p_{Q|R}(q_{c,t}|1)\right|\right|_{TV}\leq3^{(i-2)}\epsilon,
\end{equation}
where we utilized the notation $\mathsf H_{2,c}=(\mathsf H_{2,p},S_p^L,B_p^L)$. Therefore, by induction, Equation \eqref{eq:sentence} holds for arbitrary $i\geq3$.

\textbf{Calculation of the payoff of the $i$-th block for $i\geq 2$:}
Actions of team $A$ and the signals $S_c^L$ were produced according to
$$ p_{\bold A_c^L,S_c^L|\mathsf H_{2,c},Q_c^L}(\bold a_c^L|h_{2,c},q_c^L)=\prod_{t=1}^L p_{\bold A|Q}(\bold a_{c,t}|q_{c,t})p_{S|\bold A}(s_{c,t}|\bold a_{c,t}).$$
Player $B$ constructs his action $B_{c,t}$ from the observations available to him at stage $t$ of the current block, that is
$$p_{B_c^L|\mathsf H_{2,c},Q_c^L,\bold A_c^L,S_c^L}(b_c^L|h_{2,c},q_c^L,\bold a_c^L,s_c^L)=\prod_{t=1}^Lp_{B_{c,t}|\mathsf H_{2,c},S_c^{t-1},B_c^{t-1}}(b_{c,t}|h_{2,c}, s_c^{t-1}, b_c^{t-1}).$$
Therefore, using the first property of total variation in Lemma~\ref{lemma:tv_p} for random variables $E=(\mathsf H_{2,c},Q_c^L)$ and $F=(\bold A_c^L,B_c^L,S_c^L)$ along with Equation \eqref{eq:sentence}, we obtain
\begin{align*}
\Bigg\|&p_{Q_c^L,\bold A_c^L,S_c^L,\mathsf H_{2,c},B_c^L}(q_c^L,\bold a_c^L,s_c^L,h_{2,c},b_c^L)-p_{\mathsf H_{2,c}}(h_{2,c})\times \\
&\qquad\quad\prod_{t=1}^{\lceil\gamma L \rceil} p_{Q|R}(q_{c,t}|0)p_{\bold A|Q}(\bold a_{c,t}|q_{c,t})p_{S|\bold A}(s_{c,t}|\bold a_{c,t})p_{B_{c,t}|\mathsf H_{2,c},S_c^{t-1},B_c^{t-1}}(b_{c,t}|h_{2,c},s_c^{t-1},b_c^{t-1})\times\\
&\qquad\prod_{t=\lceil\gamma L \rceil+1}^L p_{Q|R}(q_{c,t}|1)p_{\bold A|Q}(a_{c,t}|q_{c,t})p_{S|\bold A}(s_{c,t}|\bold a_{c,t})p_{B_{c,t}|\mathsf H_{2,c},S_c^{t-1},B_c^{t-1}}(b_{c,t}|h_{2,c},s_c^{t-1},b_c^{t-1})\Bigg\|_{TV}\\
&\qquad\qquad\qquad\qquad\qquad\qquad\qquad\qquad\qquad\qquad\qquad\qquad\qquad\qquad\qquad\leq3^{(i-2)}\epsilon.
\end{align*}
Then, by using the second property of total variation in Lemma~\ref{lemma:tv_p} for random variables $E=(\bold A_c^L,B_c^L)$ and $F=(Q_c^L,\mathsf H_{2,c},S_c^L)$ and considering that $i-2\leq N$ we have
\begin{align} \label{eq:actions_dist}
\Bigg\|p_{\bold A_c^L,B_c^L}(\bold a_c^L,b_c^L)- \prod_{t=1}^{\lceil\gamma L \rceil} &p_{\bold A|R}(\bold a_{c,t}|0)p_{B_{c,t}|\bold A_c^{t-1},B_c^{t-1}}(b_{c,t}|\bold a_c^{t-1},b_c^{t-1})\times \notag \\ &\prod_{t=\lceil\gamma L \rceil+1}^L p_{\bold A|R}(\bold a_{c,t}|1)p_{B_{c,t}|\bold A_c^{t-1},B_c^{t-1}}(b_{c,t}|\bold a_c^{t-1},b_c^{t-1})\Bigg\|_{TV} \leq 3^N\epsilon.
\end{align}

Note that by repeating the above arguments, \eqref{eq:asmpN} concludes that the inequality \eqref{eq:actions_dist} holds for the second block as well, \emph{i.e.,} substituting the subscripts $`c`$ with $`s`$, the inequality \eqref{eq:actions_dist} still holds. Let $\mathsf{M}$ be the maximum absolute value of the payoff table. By relating the total variation distance to the payoff, Equation \eqref{eq:actions_dist} implies that the payoff of team $A$ at the second and the current block (in the $i$-th block for any $i\geq 2$) is at most in $2\mathsf{M}3^N\epsilon$ distance of the payoff of team $A$ when they played i.i.d.\ according to $p_{\bold A|R}(\bold a|0)$ in $\gamma$ portion of stages and according to $p_{\bold A|R}(\bold a|1)$ in the remaining $1-\gamma$ portion of stages. Thus, it suffices to take $N$ large enough so that the effect of the first block in the average payoff of total stages is diminished, and then take $\epsilon$ small enough (consequently take $L$ large enough such that Equations \eqref{eq:second_block} and \eqref{eq:current_block} are satisfied) to make $2\mathsf{M}3^N\epsilon$ as small as desired.

\begin{remark}
We have shown that fixing some $k$ and $N$, for sufficiently large $L$, we can achieve the desired payoff in $T=L(k+N)$ stages. If $T$ is not divisible by $k+N$, we can make the first block slightly longer but because $k+N\ll T$ for sufficiently large $T$, this has negligible effect on the achieved payoff.
\end{remark}
\color{black}\begin{remark} \label{rem:boundary_case}
When player $B$ monitors the actions of the team players imperfectly, we showed that the team players can secure $\pi(\bold A|R)$ for $p_{RQ\bold AS}$ satisfying \eqref{eqnE1}, \eqref{eqnE2} and the strict form of \eqref{eqnE3}. For the boundary case when \eqref{eqnE3} holds with equality, namely when $H(Q\bold A|RS)=H(Q|R)$, the proof given above needs a slight modification: let $\epsilon'>0$ be an arbitrarily small real number. We add to each block (other than the first block) another subblock of length $\lceil \epsilon'L \rceil$ which we call ``the shared randomness banking subblock''. In this subblock, the team players play i.i.d.\ according to $p_{\bold A}^{(0)}$, as they did in the first block. The remaining stages of the block is divided into two subblocks of lengths $\lceil (1-\epsilon')\gamma L\rceil$ and $L-\lceil \epsilon'L \rceil-\lceil (1-\epsilon')\gamma L\rceil$, where the team players play i.i.d.\ according to $p_{\bold A|R}(\bold a|0)$ and $p_{\bold A|R}(\bold a|1)$, respectively. To do so, the team players need $(1-\epsilon')H(Q|R)$ per stage shared randomness, while they can distill $\epsilon' H(\bold A^{(0)}|S^{(0)})+(1-\epsilon')H(\bold AQ|SR)$ per stage shared randomness. Note that since $H(Q\bold A|RS)=H(Q|R)$ and $H(\bold A^{(0)}|S^{(0)})>0$, the distilled randomness is strictly more than the consumed randomness. Therefore, using a similar argument as we utilized for the strict case of $H(Q\bold A|RS)>H(Q|R)$, we can prove that the team players can secure payoffs arbitrarily close to $-\epsilon' \mathsf{M}+(1-\epsilon')\pi(\bold A|R)$. Since $\epsilon'$ is arbitrary, we conclude that the team players can secure $\pi(\bold A|R)$.
\end{remark}\color{black}
\subsection{Proof of Theorem \ref{theorem:secret_corr}: Player $B$ can defend $w$}\label{sec:defend_w}
Since we have given the solution expression in a different form than the one given by \cite{Gossner_Tomala}, we adapt the proof of \cite{Gossner_Tomala} to our solution form.

In this subsection we assume that at the end of each stage $t$, the players of team $A$ in addition to $\mathbf{A}_t$ and $B_t$, observe the signal $S_t$. We show that even with the more information available to players of team $A$, player $B$ can still defend $w$. Let $\sigma$ be an arbitrary strategy for team $A$. We generate strategy $\tau$ for player $B$ as follows: given $h_2^t$, an arbitrary history of observations of player $B$ until stage $t$, $\tau_t(h_2^t)$ is the best choice of player $B$ that minimizes the expected payoff at stage $t$, \emph{i.e.},
$$ \tau_t(h_2^t) \in \underset{b\in \mathcal{B}}{\arg \min} \E_{\sigma} \left[u_{\mathbf{A}_t,b}|\mathsf{H}_2^t=h_2^t\right],$$
where $\E_{\sigma}$ denotes the expectation with respect to the probability distribution induced by $\sigma$. Let team $A$ and player $B$ play with respective strategies $\sigma$ and $\tau$ and let $\mathbf{A}^T$, $B^T$ and $S^T$ denote the sequence of actions and signals generated during the $T$ stages of the game. Define:
$$R=(I,B^{I-1},S^{I-1}),\quad Q=(I,\bold A^{I-1},B^{I-1},S^{I-1}), \quad \tilde{\mathbf A}=\mathbf A_I \quad \textrm{and}\quad \tilde S=S_I,$$
where $I$ is a uniformly distributed random variable on $\mathcal I=\{1,2,\ldots,T\}$ and independent of $(\mathbf A^T, B^T, S^T)$. Random variable $I$ is the so-called time sharing random variable. Note that $R$ is a function of $Q$. Therefore,
$$p_{\tilde S,\tilde{\bold A},R,Q}(s,\bold a,r,q)=p_{RQ}(r,q)p_{\tilde{\bold A}|Q}(\bold a|q)p_{S|\bold A}(s|\bold a).$$
Since the action profile $\bold A_t$ is implemented distributively by conditioning on $(\bold A^{t-1},B^{t-1},S^{t-1})$, the conditional distribution $p_{\tilde{\bold A}|Q}(\bold a|q)$ belongs to $\Pi$, \emph{i.e.}, $p_{\tilde{\bold A}|Q}(\bold a|q)\in \Pi.$ Next, we need to show that $H(Q\tilde{\bold A}|\tilde SR)\geq H(Q|R)$.
Note that
\begin{align}
H(Q\tilde{\bold A}|\tilde SR)-H(Q|R)=H(\tilde{\bold A} \tilde S|QR)-H(\tilde S|R).\label{eqn:equivalent-form}
\end{align}
The inequality $H(\tilde{\bold A} \tilde S|QR)-H(\tilde S|R)\geq 0$ follows from
\begin{align}
H(\tilde{\bold A} \tilde S|QR)-&H(\tilde S|R) =
H(\mathbf A_I S_I|I,\bold A^{I-1},B^{I-1},S^{I-1})-H(S_I|I,B^{I-1},S^{I-1}) \nonumber
\\&=
\sum_{t=1}^T \frac 1T\left(H(\bold A_t,S_t|\bold A^{t-1},B^{t-1},S^{t-1}) - H(S_t|B^{t-1},S^{t-1})\right)\label{eqn:t1}\\
&=\sum_{t=1}^T \frac 1T \left(H(\bold A^t|B^{t},S^{t})-H(\bold A^{t-1}|B^{t-1},S^{t-1})\right)\label{eqn:t2}\\
&=\frac 1T H(\bold A^T|B^T,S^T)\geq 0,\nonumber
\end{align}
where \eqref{eqn:t1} follows from the fact that $I$ is uniform and independent of $(\mathbf A^T, B^T, S^T)$, and \eqref{eqn:t2} follows from
\begin{align}
H(&\bold A^t|B^t,S^t)-H(\bold A^{t-1}|B^{t-1},S^{t-1})\nonumber\\
&=H(\bold A^{t-1}|B^t,S^t)+H(\bold A_t|B^t,S^t,\bold A^{t-1})-H(\bold A^{t-1}|B^{t-1},S^{t-1})\nonumber\\
&=H(\bold A^{t-1},B_t,S_t|B^{t-1},S^{t-1})-H(B_t,S_t|B^{t-1},S^{t-1})+H(\bold A_t,B_t,S_t|\bold A^{t-1},B^{t-1},S^{t-1})\nonumber\\ &\qquad \qquad\qquad-H(B_t,S_t|\bold A^{t-1},B^{t-1},S^{t-1})-H(\bold A^{t-1}|B^{t-1},S^{t-1})\nonumber\\
&=H(\bold A_t,B_t,S_t|\bold A^{t-1},B^{t-1},S^{t-1})-H(B_t,S_t|B^{t-1},S^{t-1})\nonumber\\
&=H(\bold A_t,S_t|\bold A^{t-1},B^{t-1},S^{t-1})- H(S_t|B^{t-1},S^{t-1})+H(B_t|\bold A^{t},B^{t-1},S^{t})-H(B_t|B^{t-1},S^{t})\nonumber\\
&=H(\bold A_t,S_t|\bold A^{t-1},B^{t-1},S^{t-1})- H(S_t|B^{t-1},S^{t-1}),\label{eqn:t3}
\end{align}
where \eqref{eqn:t3} holds because the random variable $B_t$ is conditionally independent of $(\bold A^t,S_t)$ given $(B^{t-1},S^{t-1})$.

Now we relate the payoff of $\sigma$ and $\tau$ to $R$ and $\bold{\tilde A}:$
\begin{align}\label{eq:utility_pi}
\lambda_T(\sigma,\tau)&=\sum_{t=1}^T \frac 1T \sum_{b^{t-1},s^{t-1}}p_{B^{t-1},S^{t-1}}(b^{t-1},s^{t-1})\min_{b\in \mathcal B}\E\left[u_{\bold A_t,b}|B^{t-1}=b^{t-1},S^{t-1}=s^{t-1}\right]\notag \\
&=\sum_{r} p_R(r)\min_{b\in \mathcal B}\E\left[u_{\tilde{\bold A},b}|R=r\right] = \pi (\tilde{\bold A}|R).
\end{align}
We have identified random variables $(R, Q, \tilde{\bold A}, \tilde{S})$ satisfying the constraints of the problem, except for the cardinality bounds on $R$ and $Q$. Cardinality of $R$ and $Q$ can be reduced using the standard arguments such as the support lemma of \cite[Appendix C]{elgamal} or the Fenchel-Bunt extension to the Caratheodory's theorem. \color{black} We leave the argument on the reduction of the cardinality of $R$ and $Q$ to Proposition~\ref{pro:cardinality}. According to Proposition~\ref{pro:cardinality}, given the random variables $(R, Q, \tilde{\bold A}, \tilde{S})$ satisfying the constraints of the problem except for the cardinality bounds on $R$ and $Q$, we can identify other random variables $(R', Q', \bold A', S')$ such that they satisfy all the constraints of the problem.
\begin{proposition} \label{pro:cardinality}
Let $\bold{A}\in \mathcal A$, $S \in \mathcal S$, $R\in \mathcal R$ and $Q\in \mathcal Q$ have a joint distribution satisfying Equations~\eqref{eqnE1}-\eqref{eqnE3} along with
\begin{equation}\label{eq:pi}
\pi(\bold{A}|R) \geq \beta,
\end{equation}
where $\mathcal R$ and $\mathcal Q$ are finite sets with arbitrary cardinalities and $\beta$ is an arbitrary real number. There exist random variables $\bold{A}'\in \mathcal A$, $S' \in \mathcal S$, $R'\in \mathcal R'$ and $Q'\in \mathcal Q'$ such that $|\mathcal R'|=2$, $|\mathcal Q'|=2|\mathcal A|$ and the joint distribution on $(\bold{A}',S',R',Q')$ satisfies Equations~\eqref{eqnE1}-\eqref{eqnE3} and \eqref{eq:pi}.
\end{proposition}
The proof of Proposition~\ref{pro:cardinality} is provided in \ref{sec:cardinality}. \color{black}

\section{Computation of min-entropy function}\label{S:prob_state}
Consider a one-shot zero-sum game between players Alice (maximizer) and Bob (minimizer) and let $\mathcal J(\mathtt h)$ be the maximum expected payoff that Alice can secure (regardless of what Bob plays) by playing mixed actions of entropy at most $\mathtt h$ (as defined in Equation \eqref{eq:j}). As stated in Theorems~\ref{T:maxmin_val_gossner} and \ref{T:maxmin_val}, $\mathcal J(\mathtt h)$ characterizes the maxmin values of the repeated game with bounded entropy (Section~\ref{subsec:gossner}) and the repeated game with leaked randomness (Section~\ref{sec:5op}). The main goal of this section is to study the computational aspects of $\mathcal J(\mathtt h)$.

Let $F(w)$ denote the inverse function of $\mathcal J(\mathtt h)$. $F(w)$ is the minimum entropy that Alice needs to secure payoff $w$; thus, we call it the \emph{min-entropy function}. To compute $F(w)$, one has to minimize the concave function of entropy over the polytope of mixed actions that secure payoff $w$ for Alice. The problem of entropy minimization over a polytope is a standard problem and is known to be NP-hard (see \cite{Kovacevic2}). Thus, the computation of $F(w)$ (or $\mathcal J(\mathtt h)$) is NP-hard. In this section, we study $\mathcal J(\mathtt h)$ thorough its inverse, $F(w)$, and provide some computationally efficient upper and lower bounds for it.

\subsection{Problem statement}
Consider a zero-sum game between players Alice ($A$) and Bob ($B$) with respective pure strategies sets $\mathcal{A}=\{1,\dots,n\}$ and $\mathcal{B}=\{1,\dots,n'\}$, where $n$ and $n'$ are natural numbers. The payoff matrix is denoted by $\mathtt{U}=[u_{i,j}]$, where $u_{i,j}$ is the real valued payoff that player $A$ gets from player $B$ when $i\in \mathcal{A}$ and $j\in \mathcal{B}$ are played. Player $A$ (player $B$) wishes to maximize (minimize) the expected payoff. The set of all randomized strategies of players $A$ and $B$ are denoted by $\Delta (\mathcal{A})$ and $\Delta (\mathcal{B})$ respectively, which are the probability simplexes on sample spaces $\mathcal{A}$ and $\mathcal{B}$ respectively. Thus, Alice's strategy corresponds to a pmf $\mathbf{p}=(p_{1},p_{2},\dots ,p_{n})$, which can be also illustrated as a column vector with non-negative entries that add up to one.

Assume that player $A$ uses randomized strategy $\mathbf{p}$. Then, the payoff of Alice if Bob plays $j\in\mathcal{B}$ is $\sum_{i}p_iu_{i,j}$. We say that Alice secures payoff $w$ with randomized strategy $\mathbf{p}$ (regardless of the action of player $B$) if $\sum_{i}p_iu_{i,j}\geq w$ for all $j\in\mathcal{B}$. Thus, the set of all distributions that guarantee payoff $ w $ for player $A$ can be expressed as
\begin{equation}\label{E:P_set}
\mathcal{P}_{\mathtt U}(w)=\{\mathbf{p}\in \Delta(\mathcal{A}): {\mathbf{p}^T\mathtt U}\geq w \mathbf{1}^T\},\end{equation}
where $\mathbf{p}^T$ is the transpose of the column vector $\mathbf{p}$, $\mathbf{1}$ is a column vector of all ones and $\mathbf{v}_1\geq \mathbf{v}_2$ means any element of $\mathbf{v}_1$ is greater than or equal to the corresponding element at $\mathbf{v}_2$.

We define
\begin{equation}\label{E:min_entropy2}
F_{\mathtt U}(w) \triangleq \underset{\mathbf{p}\in \mathcal{P}_{\mathtt U}( w )}{\min} H(\mathbf{p}).
\end{equation}
If the set $\mathcal{P}_{\mathtt U}( w )$ is empty, we set $F_{\mathtt U}(w)=+\infty$.

\textbf{A remark on notation:}
Two-player zero-sum games are completely characterized by their payoff matrix. Hence, for the sake of simplicity, we will call two-player zero-sum games with their payoff matrix. Thus, game $\mathtt U$ refers to a game with payoff table $\mathtt U$.

\begin{definition}\label{D:parameters} Given a game $\mathtt U=[u_{i,j}]$, parameters $\underbar{m}$, $\overline{m}$, ${v}$ and $ w ^*$ are defined as:
\begin{itemize}
\item
$\underline{m}$ ($\overline{m}$) is the minimum (maximum) element of matrix $\mathtt U$ : $\underbar{m}=\min_{i,j}u_{i,j}$ ($\overline{m}=\max_{i,j}u_{i,j}$).
\item
${v}$ is the maximum payoff secured by pure actions (\emph{pure-strategy security level}): ${v}=\max_i \min_j u_{i,j}$.
\item
$ w ^*$ is the value of the game $\mathtt U$, which is the maximum guaranteed payoff with unlimited access to random sources:
\begin{equation}\label{E:Nash_val}
w ^*=\underset{w:~\mathcal{P}_{\mathtt U}( w )\neq \emptyset}{\max} w.
\end{equation}
\end{itemize}
Note that by definition, $\underbar{m}\leq v\leq w^*\leq \overline{m}$.
\end{definition}

According to Definition~\ref{D:parameters}, ${v}$ is the payoff that is guaranteed without consumption of any randomness, whereas $ w ^*$ is the maximum guaranteed payoff when unlimited randomness is available. Thus, it is interesting to consider the min-entropy function $F_{\mathtt U}( w )$ in the domain ${v}\leq w \leq w ^*$. If $ w \leq {v}$, then $F_{\mathtt U}( w )=0$; if $ w > w ^*$ the feasible set of optimization problem in~\ref{E:min_entropy2} is empty and $F_{\mathtt U}( w )=+\infty$. When ${v}\leq w \leq w ^*$, the function $F_{\mathtt U}( w )$ is not necessarily convex or concave as a function of $w$: it is strictly
increasing and piecewise concave \cite[p. 241]{Neyman2000}. \footnote{This property is stated in \cite{Neyman2000} in terms of the function $J(\cdot)$.}

\subsection{On the set $\mathcal{P}_{\mathtt U}(w)$}\label{sec:Pu}
The function $F_{\mathtt U}(w)$ is defined in \eqref{E:min_entropy2} using $\mathcal{P}_{\mathtt U}(w)$, the set of all distributions that guarantee a security level $ w $ for player $A$. Observe that
\begin{equation}\mathcal{P}_{\mathtt U}(w)=\{\mathbf{p}\in \Delta(\mathcal{A}): {\mathbf{p}^T\mathtt U}\geq w \mathbf{1}^T\}\label{eqn:PUdef}\end{equation}
is a polytope defined via some linear constraints. As the matrix $\mathtt U$ is completely arbitrary, with a change of variables, one can convert it to many different equivalent polytopes.

\setlength{\unitlength}{1cm}
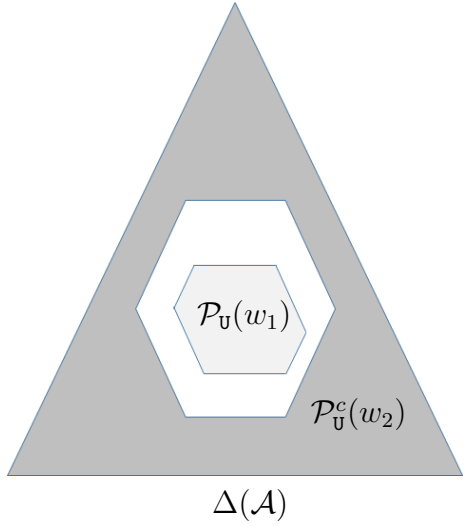
\begin{figure}
\centering
\begin{picture}(6,6)
\put(0,.5){\includegraphics[width=6cm]{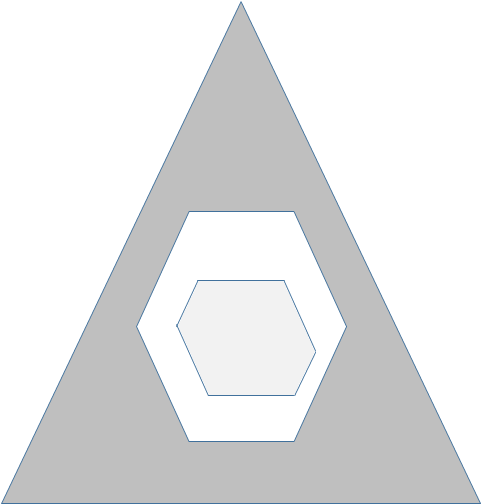}}
\put(2.7,0){$\Delta(\mathcal{A})$}
\put(2.5,2.5){$\mathcal{P}_{\mathtt U}(w_1)$}
\put(4,1.2){$\mathcal{P}_{\mathtt U}^{c}(w_2)$}
\end{picture}
\caption{Illustration of $\mathcal{P}_{\mathtt U}(w_1)$, $\mathcal{P}_{\mathtt U}^{c}(w_2)$ and $\Delta(\mathcal{A})$}
\label{F:simplex}
\end{figure}

We only need to study $\mathcal{P}_{\mathtt U}(w)$ for ${v} \leq w \leq w ^*$. It is immediate from the definition of $\mathcal{P}_{\mathtt U}(w)$ that this set is decreasing in $w$, \emph{i.e.,} for any $w_1\geq w_2$,
\begin{align}\label{eqn:P1P2com}\mathcal{P}_{\mathtt U}(w_1)\subseteq \mathcal{P}_{\mathtt U}(w_2).\end{align}
We are interested to see if the inclusion in \eqref{eqn:P1P2com} is strict, and if yes, quantify to what extent it is. To do this, we look at the distance between the set $\mathcal{P}_{\mathtt U}(w_1)$ and the compliment of $\mathcal{P}_{\mathtt U}(w_2)$ (that is $\mathcal{P}^c_{\mathtt U}(w_2)=\Delta(\mathcal A)-\mathcal{P}_{\mathtt U}(w_2)$). The sets $\mathcal{P}_{\mathtt U}(w_1)$, $\mathcal{P}_{\mathtt U}^{c}(w_2)$ and $\Delta(\mathcal{A})$ are illustrated in Figure~\ref{F:simplex}. The distance between any two sets can be defined as
$$d(\mathcal S_1, \mathcal S_2)\triangleq\inf_{\mathbf{p}\in \mathcal S_1, \mathbf{q}\in \mathcal S_2}d(\mathbf{p},\mathbf{q}),$$
where $d(\mathbf{p},\mathbf{q})$ can be any arbitrary distance measure. The standard option is the total variation distance
$$d_1(\mathbf{p},\mathbf{q})=\frac 12 \sum_{i}|p_i-q_i|.$$
With this choice of the distance, we have
\begin{theorem}\label{thm:TVw1w2}For any $w_1$ and $w_2$ satisfying ${v} \leq w_2\leq w_1 \leq w ^*$, we have
$$d_1(\mathcal{P}^c_{\mathtt U}(w_2), \mathcal{P}_{\mathtt U}(w_1))\geq\frac{|w_1-w_2|}{|\overline{m}-\underbar{m}|}.$$
\end{theorem}
Observe that the quantities $\underline{m}$ and $\overline{m}$ can be simply computed from $\mathtt{U}$. The idea of the proof of Theorem~\ref{thm:TVw1w2} is standard (e.g., see \cite[eq. (4)]{Gossner2002} for a similar derivation), but is included in Section~\ref{proof:mainw1w2} for completeness.

In this paper, we propose the use of the R\'enyi divergence of order two between $\mathbf{p}$ and $\mathbf{q}$ to quantify the distance between two distributions:\footnote{In this definition, we set $p_i^2/q_i$ to be zero if $p_i=q_i=0$, and infinity if $p_i>0$ while $q_i=0$. We have that $d_2(\mathbf{p},\mathbf{q})\geq 0$, and $d_2(\mathbf{p},\mathbf{q})=0$ if and only if $\mathbf{p}=\mathbf{q}$.}
$$d_2(\mathbf{p},\mathbf{q})\triangleq \log(\sum_i\frac{p_i^2}{q_i})=\log(1+\chi^2(\mathbf{p}, \mathbf{q})).$$

Our first result gives the following bound:
\begin{theorem}\label{thm:mainw1w2}For any $w_1, w_2$ satisfying ${v} \leq w_2\leq w_1 \leq w ^*$, we have
$$d_2(\mathcal{P}^c_{\mathtt U}(w_2), \mathcal{P}_{\mathtt U}(w_1))\geq\log\left(1+\frac{( w_1 -w_2)^2}{( w_1 -\underbar{m})(\overline{m}- w_1 )}\right).$$
\end{theorem}
The proof can be found in Section~\ref{proof:mainw1w2}.
The above result is derived by a probabilistic approach, which we believe is novel in the context of linear programming.

Our second result is less crucial, but still useful. It gives a compact formula for the supporting hyperplanes of $\mathcal{P}_{\mathtt U}(w)$ in terms of the Nash equilibrium of a game. We need a definition:

\begin{definition}
$\Val(\mathtt{U})$ denotes the value of a two-player zero-sum game with payoff table $\mathtt{U}=[u_{i,j}]$.
\end{definition}
$\Val(\mathtt{U})$ is the maximum value that Alice can guarantee using arbitrary mixed strategies. It is known that in a zero-sum game, while the game might have multiple Nash equilibria, the value of Alice in all of the equilibria is the same (See e.g., \cite{Narahari}, p.145).

Given values $a_1, a_2, \ldots, a_n\in \mathbb{R}$, let us construct a new table whose $(i,j)$ entry is $\tilde{u}_{ij}=u_{i,j}+a_i$. In other words, we add $a_i$ to the entries in the $i$-th row of $\mathtt{U}$. The new table can be expressed as
$\tilde{\mathtt{U}}=\mathtt{U}+\mathbf{a}\cdot\mathbf{1}^T$, where $\mathbf{a}$ is a column vector whose entries are $a_1, a_2, \ldots, a_n$ and $\mathbf{1}^T$ is a row vector of all ones.
Observe that the table $\tilde{\mathtt{U}}$ can be intuitively understood as giving an additional incentive $a_i$ to Alice for playing her $i$-th action (it is actually a disincentive or ``tax" if $a_i<0$).
\begin{theorem}\label{theorem-sec4-1}The set $\mathcal{P}_{\mathtt U}(w)$ can be characterized as follows:
$$\mathcal{P}_{\mathtt U}(w)=\left\{\mathbf{p}\in\Delta(\mathcal{A})\Big\vert \sum_{i}a_ip_i\leq \Val(\mathtt U+\mathbf{a}\mathbf{1}^T)-w, \quad\forall \mathbf{a}\right\}.$$
\end{theorem}
\begin{remark}
Note that $\max_{\mathbf{p}\in\mathcal{P}_{\mathtt U}(w)}\left(\sum_{i}a_ip_i\right)$ is simply a linear program. The Equivalence of linear program and Nash equilibria is known in the literature (\cite{Dantzig51,Adler}). However, our construction of the game $\tilde{\mathtt{U}}$ based on incentive or tax is new to best of our knowledge. \end{remark}
In Section~\ref{proof-theorem-sec4-1}, we give the proof of Theorem~\ref{theorem-sec4-1} as well as a geometric picture of the Nash equilibrium strategies of Alice.

\subsubsection{On the min-entropy function}\label{S:results}

We begin with a property of the min-entropy function. To state the property, we need the following definition:

\begin{definition}
Consider two games with payoff matrixes $\mathtt{U}_1$ and $\mathtt{U}_2$. Let $\mathtt{U}_3=\mathtt{U}_1\oplus \mathtt{U}_2$ be the direct-sum of $\mathtt{U}_1$ and $\mathtt{U}_2$. $\mathtt{U}_3$ defines a new game in which players simultaneously play one instance of $\mathtt{U}_1$ and one instance of $\mathtt{U}_2$ and the resulting payoff is the sum of payoffs from $\mathtt{U}_1$ and $\mathtt{U}_2$.
\end{definition}

\begin{theorem}\label{T:direct_sum} For every game $\mathtt{U}$, $F_{\mathtt U\oplus \mathtt U}( w )=F_{\mathtt U}( w/2)$. Similarly, for every natural number $k$, $F_{\oplus^k \mathtt U}( w )=F_{\mathtt U}( w/k)$, where $\oplus^k\mathtt U$ is $k$ times direct sum of $\mathtt{U}$.
\end{theorem}
Proof of Theorem~\ref{T:direct_sum} is provided in Section~\ref{S:direct_sum}.

An application of the above theorem is that given an expression $G_{\mathtt U}(w)$ that bounds $F_{\mathtt U}(w)$ from below for all $w$ and $\mathtt{U}$, we can conclude that $G_{\oplus^k \mathtt U}( kw )\leq F_{\oplus^k \mathtt U}(kw)=F_{\mathtt U}( w )$. Thus, $\max_kG_{\oplus^k \mathtt U}( kw )$ is also a (potentially better) lower bound to $F_{\mathtt U}(w)$. As a result, we expect that a ``good" lower (or upper bound) on $F_{\mathtt U}(w)$ should have the correct scaling behavior as we simultaneously play more and more copies of the game.

\subsection{Lower and upper bounds on the min-entropy function}
The min-entropy function $F_{\mathtt U}(w)$ in \eqref{E:min_entropy2} is the minimum of a concave function on a polytope $\mathcal{P}_{\mathtt U}(w)$. This minimum occurs at a vertex of $\mathcal{P}_{\mathtt U}(w)$. This leads to a search in the exponentially large set of vertexes of the polytope $\mathcal{P}_{\mathtt U}(w)$, which is computationally hard. We desire to find bounds on $F_{\mathtt U}(w)$ that are either explicit, or else can be computed in polynomial time. Observe that
\begin{equation}\label{E:min_entropy4}
F_{\mathtt U}(w) = \underset{\mathbf{p}\in \mathcal{P}_{\mathtt U}( w )}{\min} H(\mathbf{p}) = \log(|\mathcal A|)-
\underset{\mathbf{p}\in \mathcal{P}_{\mathtt U}( w )}{\max} D(\mathbf{p}\|\mathbf{p}^{\mathsf u}),
\end{equation}
where $\mathbf{p}^{\mathsf u}$ is the uniform distribution over $\mathcal A$, and $D(\mathbf{p}\|\mathbf{q})=\sum_i p_i\log({p_i}/{q_i})$ is the Kullback--Leibler (KL) divergence. Thus, we are interested in finding the vertex of $\mathcal{P}_{\mathtt U}(w)$ which has maximum distance from the uniform distribution (with respect to KL divergence).

\textbf{Lower bound:} To prove lower bounds for
\begin{equation}
F_{\mathtt U}(w) = \underset{\mathbf{p}\in \mathcal{P}_{\mathtt U}( w )}{\min} H(\mathbf{p}),\label{eqn:Fw-eq}
\end{equation}
one idea is to replace the entropy function with a smaller function and compute the minimum over $\mathcal{P}_{\mathtt U}( w )$. The second idea is to minimize the min-entropy function $F_{\mathtt U}(w)$ over all payoff tables with given properties such as $\underline{m}$, $\overline{m}$ and $v$. Another idea is to relax the set of distributions $\mathcal{P}_{\mathtt U}( w )$ and replace it with a potentially bigger set. We proceed with the first idea, then elaborate on the second idea and finally comment on the third idea.

Using the fact that the R\'enyi entropy is decreasing in its order, we obtain that for any $\alpha>1$,
\begin{equation}\label{E:min_entropy3}
F_{\mathtt U}(w) = \underset{\mathbf{p}\in \mathcal{P}_{\mathtt U}( w )}{\min} H(\mathbf{p}) \geq \underset{\mathbf{p}\in \mathcal{P}_{\mathtt U}( w )}{\min} H_\alpha(\mathbf{p}),
\end{equation}
where $H_\alpha(\mathbf{p})$ is the R\'enyi entropy of order $\alpha$:
$$H_\alpha(\mathbf{p}) = \frac{1}{1-\alpha}\log_2\Bigg(\sum_i p_i^\alpha\Bigg).$$
The case of $\alpha=2$ is related to the Euclidean norm and results in an optimization problem similar to the one given in \eqref{E:min_entropy4} for the Euclidean norm instead of the KL divergence, which is still not tractable. However, the case $\alpha=\infty$ relates to the maximum norm and results in the following lower bound:
$$G_{\mathtt U}^{(1)}(w)=-\log_2\left(\max_{i\in\mathcal A}\underset{\mathbf{p}\in \mathcal{P}_{\mathtt U}( w )}{\max} p_i\right).$$
For each $i$, the problem of finding the maximum of $p_i$ over $\mathbf{p}\in \mathcal{P}_{\mathtt U}( w )$ is a linear program.
From Theorem~\ref{theorem-sec4-1}, we can find an upper bound on the value of this linear program, yielding
\begin{align}G_{\mathtt U}^{(1)}(w)\geq -\log_2\max_{i}
\left(\Val(\mathtt U+\mathbf{e}_i\mathbf{1}^T)-w
\right),\label{eqn:relaxG}\end{align}
where $\mathbf{e}_i$ is a vector of length $|\mathcal{A}|=n$ whose $i$-th coordinate is one, and all its other coordinates are zero. The lower bound $G_{\mathtt U}^{(1)}(w)$ or its relaxed version in \eqref{eqn:relaxG} can be found in polynomial time, even though they are not in explicit forms.

To obtain an explicit lower bound, observe that
$\log_2(1/p_i)=d_2(\mathbf{e}_i,\mathbf{p}).$ Note that the vector $\mathbf{e}_i$ is a probability vector associated to a \emph{deterministic} random variable that chooses $i$ with probability one. Take some $\epsilon>0$. By definition, deterministic strategies cannot secure a payoff of more than $v$. Hence, $\mathbf{e}_i\in\mathcal{P}^c_{\mathtt{U}}(v+\epsilon)$, so we have
\begin{align}G_{\mathtt U}^{(1)}(w)&=-\log_2\left(\underset{\mathbf{p}\in \mathcal{P}_{\mathtt U}( w )}{\max} \max_{i\in\mathcal A}p_i\right)\nonumber\\
&=\underset{\mathbf{p}\in \mathcal{P}_{\mathtt U}( w )}{\min} \min_{i\in\mathcal A}\log_2(\frac{1}{p_i})\nonumber\\
&=\underset{\mathbf{p}\in \mathcal{P}_{\mathtt U}( w )}{\min} \min_{i\in\mathcal A}d_2(\mathbf{e}_i,\mathbf{p})\nonumber
\\&\geq \underset{\mathbf{p}\in \mathcal{P}_{\mathtt U}( w )}{\min} \min_{\mathbf q\in\mathcal \mathcal{P}^c_{\mathtt{U}}(v+\epsilon)}d_2(\mathbf{q},\mathbf{p})\nonumber
\\&\geq\log_2\left(1+\frac{( w -v-\epsilon)^2}{( w -\underline{m})(\overline{m}- w )}\right)\label{eqn:def1},
\end{align}
where \eqref{eqn:def1} follows from Theorem~\ref{thm:mainw1w2}. Letting $\epsilon\rightarrow 0$, we obtain
$$F_{\mathtt U}(w)\geq G_{\mathtt U}^{(2)}(w)\triangleq\log_2\left(1+\frac{( w -{v})^2}{( w -\underline{m})(\overline{m}- w )}\right) \quad \forall w : {v} \leq w \leq w ^*.$$

With a similar argument and using Theorem~\ref{thm:TVw1w2} along with the fact that $p_i=1-d_1(\mathbf{e}_i,\mathbf{p})$, we obtain the following lower bound:
$$F_{\mathtt U}(w)\geq G_{\mathtt U}^{(3)}(w)\triangleq-\log_2\left(1-\frac{w-v}{\overline{m}-\underline{m}}\right) \quad \forall w : {v} \leq w \leq w ^*.$$

Observe that when $v=\underline{m}$, $G^{(2)}_{\mathtt U}(w)$ equals $G^{(3)}_{\mathtt U}(w)$. When $v\neq \underline{m}$, a simple calculation shows that $G^{(2)}_{\mathtt U}(w)\geq G^{(3)}_{\mathtt U}(w)$ if and only if $w\geq (\overline{m}+v)/2$. \begin{example}\label{ex:g2_g3}
Consider a game with payoff matrix:
$$\mathtt U=\left[\begin{matrix}-1&1&1\\1&0.5&1\\1&1&0.5\end{matrix}\right].$$
From the payoff matrix we have $v=.5$, $\underline{m}=-1$, $\overline{m}=1$ and $w^*=0.778$. Therefore, for $w\geq 0.75$, $G^{(2)}_{\mathtt U}(w)$ gives a better lower bound than $G^{(3)}_{\mathtt U}(w)$ on $F_{\mathtt U}(w)$.
\end{example}

\begin{remark} One can inspect that just like the min-entropy function, the explicit lower bounds $G_{\mathtt U}^{(2)}(w)$ and $G_{\mathtt U}^{(3)}(w)$ satisfy
$$G_{\mathtt U\oplus \mathtt U}^{(i)}( w )=G_{\mathtt U}^{(i)}( w/2),~~~i=2,3,$$
thus, have the correct scaling behavior. Additionally, by replacing entropy with R\'enyi entropy in the proof of Theorem~\ref{T:direct_sum}, one also obtains that
$$G_{\mathtt U\oplus \mathtt U}^{(1)}( w )=G_{\mathtt U}^{(1)}( w/2).$$

The function $G_{\mathtt U}^{(2)}(w)$ has second derivative for all $ {v} \leq w \leq w ^*$, while the second derivative of the piecewise concave function $F_{\mathtt U}(w)$ is defined everywhere except for a finite number of kink points. The second derivative of the function $G_{\mathtt U}^{(2)}(w)$ may be positive or negative, while $F_{\mathtt U}(w)$ is piecewise concave. On the other hand, the function $G^{(1)}_{\mathtt U}(w)$ is piecewise convex. The reason is that if $\max_{i\in\mathcal A} \max_{\mathbf{p}\in \mathcal{P}_{\mathtt U}( w )}p_i$ is attained by $i^*$ and a particular vertex of $\mathcal{P}_{\mathtt U}( w )$ for $w\in [w_1,w_2]$, in this interval $\max_{i\in\mathcal A} \max_{\mathbf{p}\in \mathcal{P}_{\mathtt U}( w )}p_{i}$ varies linearly in $w$. Then, convexity of $-\log_2(\cdot)$ results in convexity of $G^{(1)}_{\mathtt U}(w)$ in the interval $[w_1,w_2]$.
\end{remark}
\begin{remark}
Inequality \eqref{eqn:def1} shows that
\begin{equation}\label{LP_bound}
\underset{\mathbf{p}\in \mathcal{P}_{\mathtt U}( w )}{\max} p_i \leq \left(1+\frac{( w -{v})^2}{( w -\underline{m})(\overline{m}- w )}\right)^{-1} \quad \forall w : {v} \leq w \leq w ^* \text{and } \forall i=1,\dots,n,
\end{equation}
which gives an upper bound for the linear programming of $\max_{\mathbf{p}\in \mathcal{P}_{\mathtt U}( w )}p_i$. Note that $\mathcal{P}_{\mathtt U}( w )$ is a very generic polytope, parameterized by a variable $w$. By a change of variables (scaling and shifting), one can convert $\max_{\mathbf{p}\in \mathcal{P}_{\mathtt U}( w )}p_i$ to a wide class of linear programs (with no immediate connection to the probability simplex), and then use the bound given in \eqref{LP_bound}.
\end{remark}

Let $\underline m$, $\overline m$ and $v$ be the minimum entry, maximum entry and pure strategy security level of the payoff table $\mathtt U$. The lower bounds $G_{\mathtt U}^{(2)}(w)$ and $G_{\mathtt U}^{(3)}(w)$, just rely on the parameters $\underline{m}$, $\overline{m}$ and $v$. We seek to answer the following question: given that we just know the parameters $\underline{m}$, $\overline{m}$ and $v$ from payoff table $\mathtt U$, what is the tightest lower bound for the min-entropy function? To compute the tightest lower bound it suffices to minimize the min-entropy function over all payoff tables (of arbitrary size) with parameters $\underline m$, $\overline m$ and $v$. Let denote this lower bound by $G_{\mathtt U}^{(4)}(w)$, then,
\begin{equation}\label{eq:second_idea1}
G_{\mathtt U}^{(4)}(w)= \min_{\mathtt U'} \min_{\mathbf{p}\in \mathcal P_{\mathtt U'}(w)} H(\mathbf p),
\end{equation}
where the first minimization is computed over all payoff tables $\mathtt U'$ (of arbitrary size) with minimum entry $\underline m$, maximum entry $\overline m$ and pure-strategy security level $v$.
We simplify the expression of Equation \eqref{eq:second_idea1} in the following theorem:
\begin{theorem}\label{theorem:second_idea}
Let $\underline m$, $\overline m$,$v$ and $w$ be some known real numbers such that $\underline m\leq v \leq w\leq \overline m$. We have
\begin{align*}
G_{\mathtt U}^{(4)}(w)=\min_{\mathtt U'} \min_{\mathbf{p}\in \mathcal P_{\mathtt U'}(w)} H(\mathbf p) = -\left\lfloor\frac{\overline m-v}{\overline m-w}\right\rfloor &\frac{\overline m-w}{\overline m-v}\log \left(\frac{\overline m-w}{\overline m-v}\right)- \\
& \left(1-\left\lfloor\frac{\overline m-v}{\overline m-w}\right\rfloor \frac{\overline m-w}{\overline m-v}\right)\log\left(1-\left\lfloor\frac{\overline m-v}{\overline m-w}\right\rfloor \frac{\overline m-w}{\overline m-v}\right),
\end{align*}
where the first minimization is computed over all payoff tables $\mathtt U'$ (of arbitrary size) with minimum entry $\underline m$, maximum entry $\overline m$ and pure-strategy security level $v$, and $\lfloor a \rfloor$ is the greatest integer smaller than or equal to $a$.
\end{theorem}
 The proof of the above theorem is provided in Section~\ref{sec:second_idea}.
Theorem~\ref{theorem:second_idea} gives another explicit form lower bound for the min-entropy function which is optimal in the sense that it is the tightest lower bound that utilizes just the information of minimum entry, maximum entry and pure-strategy security level of the payoff table.
\color{black}
\begin{remark}
The bound $G_{\mathtt U}^{(4)}(w)$ does not depend on the minimum entry $\underline m$. This is because in the minimization problem of Equation~\eqref{eq:second_idea1}, we do not restrict the dimension of the payoff table. The key observation is as follows: given a table $\mathtt U$ with maximum entry $\overline m$ and pure-strategy security level $v$, construct another payoff table $\mathtt U'$ as follows: first, replace with $v$ the entries of $\mathtt U$ that have values less than $v$. Then, add a new row of all $\underline m$s to the resulting payoff table. Consider that $\mathtt U'$ still has pure strategy security level of $v$ and maximum entry of $\overline m$. Furthermore, any strategy that guarantees payoff $w$ in game $\mathtt U$ also guarantees payoff $w$ in game $\mathtt U'$.
\end{remark}
\color{black}

Let us now turn to the third idea to prove a lower bound for $ F_{\mathtt U}(w)$ in \eqref{eqn:Fw-eq}, namely replacing the set of distributions $\mathcal{P}_{\mathtt U}( w )$ with a potentially bigger set. As mentioned earlier, minimization of the entropy over the set
\begin{equation*}\mathcal{P}_{\mathtt U}(w)=\{\mathbf{p}\in \Delta(\mathcal{A}): {\mathbf{p}^T\mathtt U}\geq w \mathbf{1}^T\}\end{equation*}
can be difficult. However, it could be possible to solve it (or find good lower bounds for it) for special choices of the matrix $\mathtt U$.\footnote{For instance, if each row of $\mathtt U$ has only one non-zero element, the set of constraints will be on the individual coordinates of the vector $\mathbf{p}$ and minimizing entropy for such constraints is tractable.} We show how a result for an special case of $\mathtt U$ can be utilized to find a bound (computable in polynomial time) for an arbitrary $\mathtt U$. Assume that we have a way to minimize entropy over the set
\begin{equation*}\mathcal Q(\mathbf{r})\triangleq\{\mathbf{p}\in \Delta(\mathcal{A}):
{\mathbf{p}^T{\mathtt U}^*}\geq \mathbf{r}^T\}\end{equation*}
for some given matrix ${\mathtt U}^*$, and any arbitrary column vector $\mathbf{r}$. We are interested in a value for $\mathbf{r}$ such that
\begin{equation}\mathcal{P}_{\mathtt U}(w)\subseteq \mathcal Q(\mathbf{r}).\label{eqn:relax-eq}
\end{equation}
We can relax the minimization of the entropy over the set $\mathcal{P}_{\mathtt U}(w)$ by computing its minimum over the bigger set of $\mathcal Q(\mathbf{r})$. Note that an appropriate $\mathbf{r}$ in \eqref{eqn:relax-eq} can be found by solving a number of linear programs: the product ${\mathbf{p}^T{\mathtt U}^*}$ consists of a number of linear functions of coordinates of $\mathbf{p}$, and the minimum of each linear function over the set $\mathcal{P}_{\mathtt U}(w)$ is a linear program (see also Theorem~\ref{theorem-sec4-1}).

\textbf{Upper bound:}
It is clear that $F_{\mathtt U}(w) \leq H(\mathbf{p})$ for any arbitrary choice of $\mathbf{p}\in \mathcal{P}_{\mathtt U}( w )$. The following theorem gives a number of upper bounds each of which are obtained by identifying $\mathbf{p}\in \mathcal{P}_{\mathtt U}( w )$ in different ways.

\begin{theorem}\label{T:bounds}
Consider a game with payoff matrix $\mathtt{U}$ and parameters ${v}$, $\underbar{m}$, $\overline{m}$ and $ w ^*$ defined in Definition~\ref{D:parameters}. Let $\mathsf{h}^*$ be the entropy of a Nash strategy of player $A$. Define
\begin{align*}
Q_{\mathtt U}^{(1)}( w ) & = \min \left\{\mathsf{h}^*,\frac{ w -{v}}{ w ^*-{v}}\mathsf{h}^* +h\left(\frac{ w -{v}}{ w ^*-{v}}\right) \right\}, \\
Q_{\mathtt U}^{(2)}( w ) & = \min_{i\in\mathcal{A}} H(\mathbf{p}^*_{\text{max},i}),\quad \mathbf{p}^*_{\text{max},i} \in \underset{\mathbf{p}\in \mathcal{P}( w )}{\arg \max}\quad p_{i},\\
Q_{\mathtt U}^{(3)}( w ) & = \underset{j\in\mathcal{B}}{\min}\quad\underset{\mathbf{p}\in\mathcal{P}_{\mathtt U}( w ): \sum_{i\in \mathcal A}p_iu_{i,j}=w}{\max}H(\mathbf{p}),
\end{align*}
where $h(\alpha)=-\alpha \log(\alpha)-(1-\alpha)\log(1-\alpha)$ and $Q_{\mathtt U}^{(2)}( w )$ can be defined with any choice of $\mathbf{p}^*_{\text{max},i}$ from the argmax set (if there are multiple possible choices). We have
$$ F_{\mathtt U}( w ) \leq Q_{{\mathtt U}}^{(r)}( w ),\quad r=1,2,3.$$
\end{theorem}
Proof of Theorem~\ref{T:bounds} can be found in Section~\ref{proofTbound}.

\begin{remark}
The second derivative of $Q^{(1)}_{\mathtt U}(w)$ is negative for $v\leq w \leq w^*$. Thus, $Q^{(1)}_{\mathtt U}(w)$ is a concave function of $w$ and one can readily inspect that:
$$Q^{(1)}_{\mathtt U\oplus \mathtt U}(w)=Q^{(1)}_{\mathtt U}(\frac{w}{2}).$$
As $\arg \max_{\mathbf{p}\in \mathcal{P}( w )} p_{i}$ may contain multiple elements, $Q^{(2)}_{\mathtt U}(w)$ is not a well defined function of $w$. The function $Q^{(3)}_{\mathtt U}(w)$ is not necessarily scalable for game $\mathtt U\oplus \mathtt U$.
\end{remark}

\begin{example}\label{Ex:Game}
Consider two games with the following payoff tables:
$$
\mathtt U=\left[
	\begin{matrix}
	3&1 &0 &-2 &0 &-2 \\
	1&3 &-2 &0 &-2 &0 \\
	0&-2 &3 &1 &0 &-2 \\
	-2&0 &1 &3 &-2 &0 \\
	0&-2 &0 &-2 &3 &1 \\
	-2&0 &-2 &0 &1 &3
	\end{matrix}
 \right],
 \qquad
 \mathtt U'=\left[
	\begin{matrix}
	0&1 &1 &.5 \\
	1&0&.5 &1 \\
	1&.5 &0 &1 \\
	.5&1 &1 &0
	\end{matrix}
 \right].
$$
Figures~\ref{F:bounds} and \ref{F:tight_bounds} illustrate the behavior of the bounds for the games $\mathtt U$ and $\mathtt U'$, respectively. In this examples, since $v=\underline{m}$, $G^{(3)}_{\mathtt U} (w)$ coincides with $G^{(2)}_{\mathtt U} (w)$, thus, it has not been depicted in the figures.
\end{example}

\begin{figure}
\centering
\begin{picture}(17,7)
\put(0,.1){\includegraphics[width=17cm]{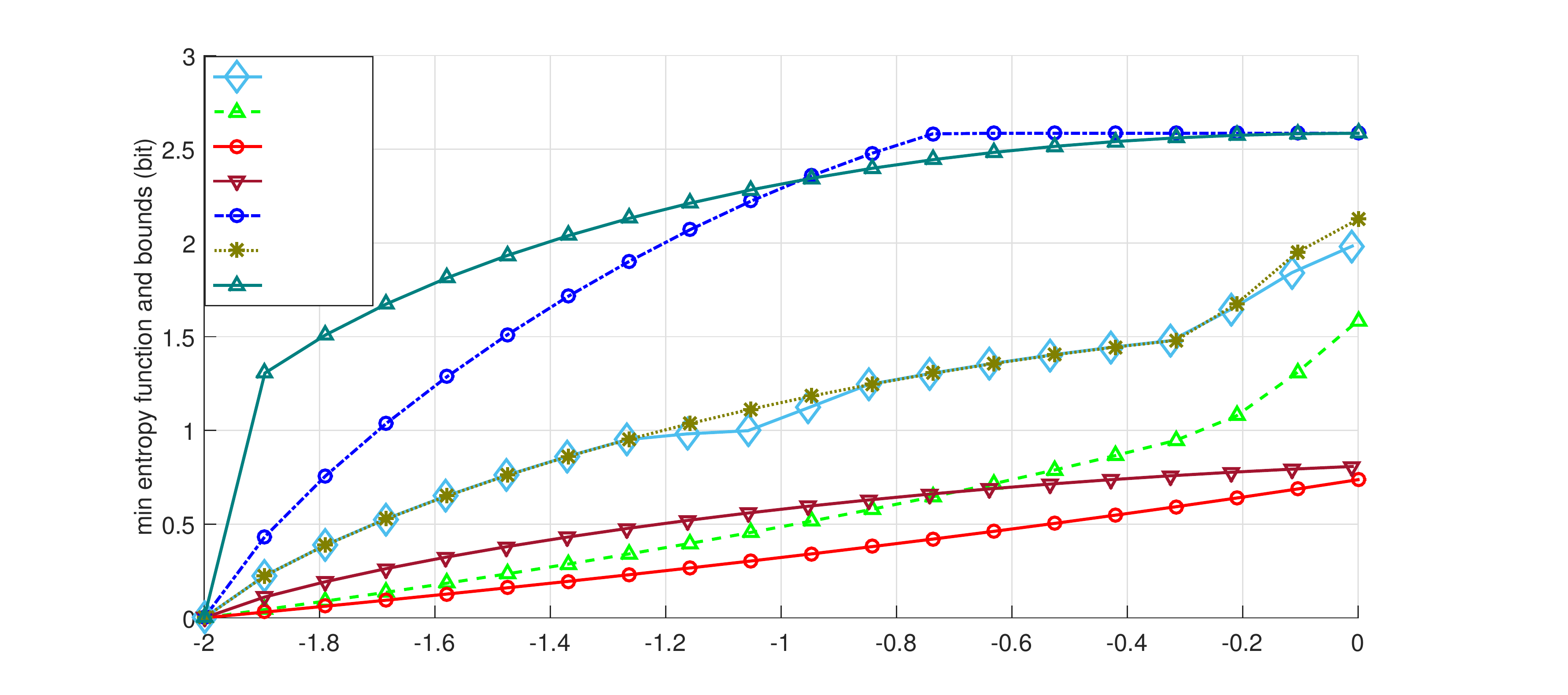}}
\put(8.3,0){$w$}
\put(3,6.8){\tiny $F_{\mathtt U}(w)$\normalsize}
\put(3,6.35){\tiny $G_{\mathtt U}^{(1)}(w)$\normalsize}
\put(3,5.95){\tiny $G_{\mathtt U}^{(2)}(w)$\normalsize}
\put(3,5.6){\tiny $G_{\mathtt U}^{(4)}(w)$\normalsize}
\put(3,5.25){\tiny $Q_{\mathtt U}^{(1)}(w)$\normalsize}
\put(3,4.85){\tiny $Q_{\mathtt U}^{(2)}(w)$\normalsize}
\put(3,4.45){\tiny $Q_{\mathtt U}^{(3)}(w)$\normalsize}
\end{picture}
\caption{Illustration of the bounds on the min-entropy function for the game $\mathtt U$ defined in Example~\ref{Ex:Game}. The horizontal line depicts $w$ and vertical line depicts the value of bounds.}\label{F:bounds}
\end{figure}

\begin{figure}
\centering
\begin{picture}(17,7)
\put(0,.1){\includegraphics[width=17cm]{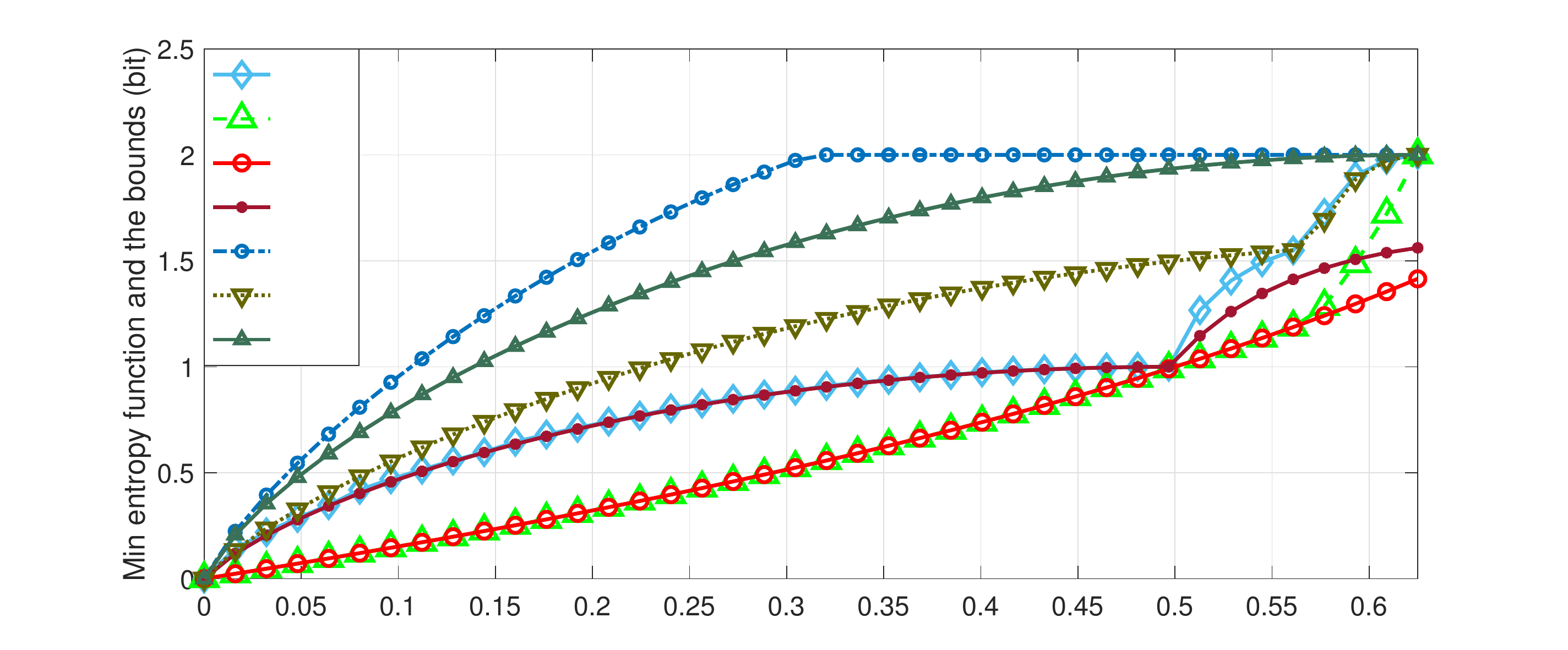}}
\put(8.3,0){$w$}
\put(3,6.3){\tiny $F_{\mathtt U'}(w)$\normalsize}
\put(3,5.82){\tiny $G_{\mathtt U'}^{(1)}(w)$\normalsize}
\put(3,5.34){\tiny $G_{\mathtt U'}^{(2)}(w)$\normalsize}
\put(3,4.86){\tiny $G_{\mathtt U'}^{(4)}(w)$\normalsize}
\put(3,4.38){\tiny $Q_{\mathtt U'}^{(1)}(w)$\normalsize}
\put(3,3.9){\tiny $Q_{\mathtt U'}^{(2)}(w)$\normalsize}
\put(3,3.42){\tiny $Q_{\mathtt U'}^{(3)}(w)$\normalsize}
\end{picture}
\caption{Illustration of the bounds on the min-entropy function for the game $\mathtt U'$ defined in Example~\ref{Ex:Game}. The horizontal line depicts $w$ and vertical line depicts the value of bounds.}\label{F:tight_bounds}
\end{figure}

\section{Proofs}\label{S:proofs}
\subsection{Proof of Proposition~\ref{Pro:mapping}} \label{S:mapping}
To prove Proposition~\ref{Pro:mapping} we make use of Lemmas~\ref{L:yasayi} and \ref{L:han}. A brief discussion on randomness extraction and proof of Lemma~\ref{L:yasayi} is provided in~\ref{sec:randomness_extraction}.
\begin{lemma} \label{L:yasayi}
Consider the correlated random sequences $X_p^L$ and $Y_p^L$ drawn i.i.d.\ from respective spaces $\mathcal{X}$ and $\mathcal{Y}$ by joint probability distribution $p_{XY}$. Let $Q^{(L)}$ be a random variable independent of $Y_p^L$ and uniformly distributed on $\{1,2,\dots,2^{RL}\}$, where $R<H(X|Y)$ is a real number; then, there exist mappings $\mathsf B_L:\mathcal{X}^L\to \{1,2,\dots,2^{RL}\}$ such that
$$\lim_{L\to \infty} \|p_{\mathsf B_L(X_p^L)Y_p^L}-p_{Q^{(L)}Y_p^L}\|_{TV} =0.$$
\end{lemma}

Choose a real number $R$ such that $\gamma H(p_A^{(1)})+(1-\gamma)H(p_A^{(2)})<R<H(X|Y)$ and $LR$ is a natural number (for a sufficiently large $L$).

We define random variables $Q^{(L)}$ and $\hat{A}^{L}$ that are mutually independent of each other and of $Y_p^L$ with the following marginal distributions: let $Q^{(L)}$ be a uniformly distributed random variable on $\{1,2,\dots,2^{RL}\}$ and $\hat{A}^{L}$ be distributed as follows:
$$p_{\hat{A}^{L}}(a^{L})=\prod_{t=1}^{\lceil\gamma L\rceil}p_A^{(1)}(a_t) \prod_{t=\lceil\gamma L )\rceil+1}^{L}p_A^{(2)}(a_t).$$
Note that $R<H(X|Y)$; thus, according to Lemma~\ref{L:yasayi}, there exist mappings $\mathsf B_L:\mathcal{X}^{L}\to \{1,2,\dots,2^{RL}\}$ satisfying
\begin{equation} \label{E:bining1}
\lim_{L\to \infty} \|p_{\mathsf B_L\left(X_p^{L}\right)Y_p^{L}}-p_{Q^{(L)}Y_p^{L}}\|_{TV} =0.
\end{equation}
On the other hand, we have
\begin{align*}
 \text{p-}\limsup_{L\to\infty}\frac{1}{L}\log\frac{1}{p_{\hat{A}^{L}}(\hat{A}^{L})} &=\gamma H(p_A^{(1)})+(1-\gamma)H(p_A^{(2)}) \\
& < R = \text{p-}\liminf_{L\to\infty}\frac{1}{L}\log\frac{1}{p_{Q^{L}}(Q^{L})}.
\end{align*}
Therefore, according to Lemma~\ref{L:han}, there exist mappings $\varphi_L:\{1,2,\dots,2^{RL}\}\to \mathcal{A}^{L}$ such that
$$ \lim_{L\to \infty} \|p_{\varphi_L\left(Q^{(L)}\right)}-p_{\hat{A}^{L}}\|_{TV}=0.$$
Considering the fact that $\varphi_L(Q^{(L)})$ and $\hat{A}^{L}$ both are independent of $Y_p^{L}$, the above equation along with the third property of the total variation distance in Lemma~\ref{lemma:tv_p} results in
\begin{equation}\label{E:han1}
\lim_{L\to \infty} \|p_{\varphi_L\left(Q^{(L)}\right)Y_p^{L}}-p_{\hat{A}^{L}Y_p^{L}}\|_{TV} =0.
\end{equation}
Furthermore, note that
\begin{align}
\|p_{\varphi_L\left(\mathsf B_L(X_p^L)\right)Y_p^{L}}-p_{\hat{A}^{L}Y_p^{L}}\|_{TV} &\leq \|p_{\varphi_L\left(\mathsf B_L(X_p^L)\right)Y_p^{L}}-p_{\varphi_L\left(Q^{(L)}\right)Y_p^{L}}\|_{TV} + \|p_{\varphi_L\left(Q^{(L)}\right)Y_p^{L}}-p_{\hat{A}^{L}Y_p^{L}}\|_{TV} \notag \\
&\leq \|p_{\mathsf B_L(X_p^L)Y_p^{L}}-p_{Q^{(L)}Y_p^{L}}\|_{TV} + \|p_{\varphi_L\left(Q^{(L)}\right)Y_p^{L}}-p_{\hat{A}^{L}Y_p^{L}}\|_{TV},
\label{E:triangular1}
\end{align}
where the first inequality follows from the triangle inequality for the total variation distance, and the second inequality follows from the fourth property of total variation in Lemma~\ref{lemma:tv_p}. Let $A_c^L=\psi_L(X_p^L)=\varphi_L\left(\mathsf B_L\left(X_p^{L}\right)\right)$. Then, by combining Equations \eqref{E:bining1}, \eqref{E:han1} and \eqref{E:triangular1}, we have
\begin{equation}\label{E:final}
\lim_{L\to \infty} \|p_{A_c^LY_p^{L}}-p_{\hat{A}^{L}Y_p^{L}}\|_{TV} =0.
\end{equation}
Thus, $\psi_L(\cdot)=\varphi_L(\mathsf B_L(\cdot))$ is the desired mapping.

\subsection{Proof of Proposition~\ref{Pro:mapping2}} \label{S:mapping2}
Choose real numbers $R_1$ and $R_2$ such that
$$R_1<H(X|Y,U=0),\quad R_2<H(X|Y,U=1),\quad p_U(0)R_1+p_U(1)R_2>H(Z|V),$$
and for sufficiently large $L$, $\lceil p_U(0)L\rceil R_1$ and $\lceil p_U(1)L\rceil R_2$ are natural numbers. Since
$$H(X|YU)=p_U(0)H(X|Y,U=0)+p_U(1)H(X|Y,U=1)>H(Z|V),$$
real numbers $R_1$ and $R_2$ with the above properties exist.

Let $\hat{L}=\lceil p_U(0)L\rceil$ and $\tilde{L}=L-\lceil p_U(0)L\rceil$. Observe that $L=\hat{L}+\tilde{L}$. We will consider three sets of random variables:
\begin{itemize}
\item (Set 1): Random variables $(X_1, \cdots, X_L, Y_1, \cdots, Y_L, W)$ defined in the statement of the proposition with the joint distribution $p_{X^LY^LW}$ satisfying
\begin{equation}\label{eq:real_variables}
\left|\left|p_{X^LY^LW}(x^L,y^L,w)-p_W(w)\prod_{t=1}^{\hat{L}} p_{XY|U}(x_t,y_t|0)\prod_{t=\hat{L}+1}^{L} p_{XY|U}(x_t,y_t|1)\right|\right|_{TV}\leq \delta_1.
\end{equation}
\item (Set 2): Mutually independent random variables $W'$, $(\hat{X}^{\hat L},\hat Y^{\hat L})$ and $(\tilde{X}^{\tilde L},\tilde Y^{\tilde L})$. Here $p_{W'}=p_W$. Furthermore, $(\hat{X}^{\hat L},\hat Y^{\hat L})$ and $(\tilde{X}^{\tilde L},\tilde Y^{\tilde L})$ are i.i.d.\ according to $p_{XY|U}(\hat x,\hat y |0)$ and $p_{XY|U}(\tilde x,\tilde y |1)$, respectively. \color{black} In other words, for $X'^L=(\hat{X}^{\hat L},\tilde X^{\tilde L})$ and $Y'^L=(\hat{Y}^{\hat L},\tilde Y^{\tilde L})$ we have,
\begin{equation*}
p_{X'^LY'^L}(x^L,y^L)=\prod_{t=1}^{\hat{L}} p_{XY|U}(x_t,y_t|0)\prod_{t=\hat{L}+1}^{L} p_{XY|U}(x_t,y_t|1).
\end{equation*}
Equation~\eqref{eq:real_variables} along with the above equation and $p_{W'}=p_W$ shows that the random variables in Set 2 are related to the random variables in Set 1 as follows:
\begin{equation}\label{eq:ideal_variables}
\|p_{X^LY^LW}-p_{X'^LY'^L}p_{W'}\|_{TV} \leq \delta_1.
\end{equation}
\item (Set 3): Two independent random variables $\hat Q^{(\hat L)}$ and $\tilde Q^{(\tilde L)}$ with uniform distributions on $\{1,2,\ldots,2^{\hat L R_1}\}$ and $\{1,2,\ldots,2^{\tilde L R_2}\}$, respectively.
\end{itemize}

Since $R_1<H(X|Y,U=0)$ and $R_2<H(X|Y,U=1)$, according to Lemma~\ref{L:yasayi}, for arbitrary $\delta'>0$, there exist mappings $\hat{\mathsf B}_{\hat L}:\mathcal{X}^{\hat L}\to \{1,2,\dots,2^{R_1\hat L}\}$ and $\tilde{\mathsf B}_{\tilde L}:\mathcal{X}^{\tilde L}\to \{1,2,\dots,2^{R_2\tilde L}\}$ satisfying
\begin{equation} \label{E:bining}
\|p_{\hat{\mathsf B}_{\hat L}\left(\hat X^{\hat L}\right)\hat Y^{\hat L}}-p_{\hat Q^{(\hat L)}}p_{\hat Y^{\hat L}}\|_{TV} \leq \delta', \quad
\|p_{\tilde{\mathsf B}_{\tilde L}\left(\tilde X^{\tilde L}\right)\tilde Y^{\tilde L}}-p_{\tilde Q^{(\tilde L)}}p_{\tilde Y^{\tilde L}}\|_{TV} \leq \delta',
\end{equation}
for sufficiently large $\hat L$ and $\tilde L$. Next, let $Q^{(L)}=(\hat Q^{(\hat L)},\tilde Q^{(\tilde L)})$, and $\mathsf B_L:\mathcal X^L\to \{1,2,\dots,2^{R_1\hat L+R_2 \tilde L}\}$ be the mapping $\mathsf B_L(\hat X^{\hat L},\tilde X^{\tilde L})=(\hat{\mathsf B}_{\hat L}(\hat X^{\hat L}),\tilde{\mathsf B}_{\tilde L}(\tilde X^{\tilde L}))$. Using Equation \eqref{E:bining},
independence of $(\hat{X}^{\hat L},\hat Y^{\hat L})$ from $(\tilde{X}^{\tilde L},\tilde Y^{\tilde L})$, and independence of $\hat Q^{(\hat L)}$ from $\tilde Q^{(\tilde L)}$, we have
\begin{equation*}
\|p_{\mathsf B_L\left(\hat X^{\hat L},\tilde X^{\tilde L}\right)\hat Y^{\hat L}\tilde Y^{\tilde L}}-p_{Q^{(L)}}p_{\hat Y^{\hat L}}p_{\tilde Y^{\tilde L}}\|_{TV} \leq 2\delta',
\end{equation*}
where we used the third property of total variation in Lemma~\ref{lemma:tv_p} for random variables $E_1=(\hat{\mathsf B}_{\hat L}(\hat X^{\hat L}),\hat Y^{\hat L})$, $F_1=(\tilde{\mathsf B}_{\tilde L}(\tilde X^{\tilde L}),\tilde Y^{\tilde L})$, $E_2=(\hat Q^{(\hat L)},\hat Y^{\hat L})$ and $F_2=(\tilde Q^{(\tilde L)},\tilde Y^{\tilde L})$.
Then, using the fifth property of total variation in Lemma~\ref{lemma:tv_p} for random variables $E_1=(\mathsf B_L(\hat X^{\hat L},\tilde X^{\tilde L}),\hat Y^{\hat L},\tilde Y^{\tilde L})$, $E_2=(Q^{(L)},\hat Y^{\hat L},\tilde Y^{\tilde L})$ and $F=W'$, we get
 \begin{equation*}
\|p_{W'}p_{\mathsf B_L\left(\hat X^{\hat L},\tilde X^{\tilde L}\right)\hat Y^{\hat L}\tilde Y^{\tilde L}}-p_{W'}p_{Q^{(L)}}p_{\hat Y^{\hat L}}p_{\tilde Y^{\tilde L}}\|_{TV} \leq 2\delta'.
\end{equation*}
\color{black} Furthermore, by using the notation $X'^L=(\hat{X}^{\hat L},\tilde X^{\tilde L})$ and $Y'^L=(\hat{Y}^{\hat L},\tilde Y^{\tilde L})$, the above inequality is simplified as follows:
\begin{equation}\label{eq:bining2}
\|p_{W'}p_{\mathsf B_L\left(X'^L\right)Y'^L}-p_{W'}p_{Q^{(L)}}p_{Y'^L}\|_{TV} \leq 2\delta'.
\end{equation}
The above inequality shows that the mapping $\mathsf B_L$ simulates random variable $Q^{(L)}$ (in Set 3 of random variables) from $X'^L$ (in Set 2 of random variables) within the total variation distance of $2\delta'$. \color{black} Now, we consider the Set 1 of random variables.
By applying the mapping $\mathsf B_L$ on sequence $X^L$ we have that
\begin{align}\label{eq:bining3}
\|p_{W\mathsf B_L\left(X^L\right)Y^L}-p_Wp_{Q^{(L)}}p_{Y^L}\|_{TV} &\leq \|p_{W\mathsf B_L\left(X^L\right)Y^L}- p_{W'}p_{\mathsf B_L\left(X'^L\right)Y'^L}\|_{TV} + \notag \\
&\qquad\|p_{W'}p_{\mathsf B_L\left(X'^L\right)Y'^L} - p_{W'}p_{Q^{(L)}}p_{Y'^L}\|_{TV} + \notag\\
&\qquad\quad\|p_{W'}p_{Q^{(L)}}p_{Y'^L}-p_Wp_{Q^{(L)}}p_{Y^L}\|_{TV} \notag \\
&\leq 2\delta_1+2\delta',
\end{align}
where the first inequality follows from the triangle inequality for the total variation distance, and the second inequality results from Equation \eqref{eq:bining2} and the following two facts:
\begin{itemize}
\item \color{black}$\|p_{W\mathsf B_L\left(X^L\right)Y^L}- p_{W'}p_{\mathsf B_L\left(X'^L\right)Y'^L}\|_{TV}\leq \delta_1$. This is a consequence of Equation~\eqref{eq:ideal_variables}, combined with the fourth property of the total variation distance in Lemma~\ref{lemma:tv_p}.

\item$\|p_{W'}p_{Q^{(L)}}p_{Y'^L}-p_Wp_{Q^{(L)}}p_{Y^L}\|_{TV}\leq \delta_1$. To see that this inequality is correct, first, note that Equation~\eqref{eq:ideal_variables} along with the second property of total variation in Lemma~\ref{lemma:tv_p} implies that $Y^L$ is in $\delta_1$ distance of $Y'^L$. Then, the fact that $p_{W'}=p_{W}$ along with the fifth property of total variation in Lemma~\ref{lemma:tv_p} gives the desired inequality.
\end{itemize}

Now, Let $\hat Z^L$ be an ideal sequence with distribution
$$p_{\hat Z^L}(\hat z^L)=\prod_{t=1}^{\lceil p_V(0)L\rceil} p_{Z|V}(\hat z_t|0)\prod_{t=\lceil p_V(0)L\rceil+1}^{L} p_{Z|V}(\hat z_t|1).$$
Then, we have
\begin{align*}
 \text{p-}\limsup_{L\to\infty}\frac{1}{L}\log\frac{1}{p_{\hat{Z}^{L}}(\hat{Z}^{L})} &= H(Z|V)\\
& < p_U(0)R_1+p_U(1)R_2 = \text{p-}\liminf_{L\to\infty}\frac{1}{L}\log\frac{1}{p_{Q^{L}}(Q^{L})}.
\end{align*}
Therefore, according to Lemma~\ref{L:han}, there exist mappings $\varphi_L:\{1,2,\dots,2^{R_1\hat L+R_2\tilde L}\}\to \mathcal{Z}^{L}$ such that for sufficiently large $L$
\begin{equation}\label{eq:han1}
\|p_{\varphi_L\left(Q^{(L)}\right)}-p_{\hat{Z}^{L}}\|_{TV} \leq \delta'.
\end{equation}
Next, by replacing $p_{E_1}$, $p_{E_2}$ and $p_{F}$ with respective distributions $p_{\varphi_L(Q^{(L)})}$, $p_{\hat Z^L}$ and $p_Wp_{Y^L}$, the fifth property of total variation in Lemma~\ref{lemma:tv_p} along with Equation~\eqref{eq:han1} implies
\begin{equation}\label{E:han2}
\|p_{\varphi_L\left(Q^{(L)}\right)}p_{Y^{L}}p_W-p_{\hat{Z}^{L}}p_{Y^{L}}p_W\|_{TV} \leq \delta'.
\end{equation}
Thus we have
\begin{align} \label{E:triangular}
\|p_{\varphi_L\left(\mathsf B_L(X^L)\right)Y^{L}W}-p_{\hat{Z}^{L}}p_{Y^L}p_W\|_{TV} &\leq \|p_{\varphi_L\left(\mathsf B_L(X^L)\right)Y^{L}W}-p_{\varphi_L\left(Q^{(L)}\right)}p_{Y^{L}}p_W\|_{TV} + \notag \\
&\quad\quad\quad\quad\quad\qquad \|p_{\varphi_L\left(Q^{(L)}\right)}p_{Y^{L}}p_W-p_{\hat{Z}^{L}}p_{Y^L}p_W\|_{TV} \notag \\
&\leq \|p_{\mathsf B_L(X^L)Y^{L}W}-p_{Q^{(L)}}p_{Y^{L}}p_W\|_{TV} + \delta'\notag\\
&\leq 2\delta_1+3\delta',
\end{align}
where the first inequality follows from the triangle inequality for the total variation distance; the second inequality is a consequent of Equation \eqref{E:han2} and the fourth property of total variation in Lemma~\ref{lemma:tv_p} where $E_1$ and $E_2$ are replaced with $\mathsf B_L(X^L)Y^{L}W$ and $Q^{(L)}Y^{L}W$, respectively; finally, the third inequality results from Equation \eqref{eq:bining3}.

Let $Z^L= \varphi_L\left(\mathsf B_L(X^L)\right)$ and take $\delta'$ small enough that $3\delta'\leq \delta_2$. Then, from Equation~\eqref{E:triangular} we have
$$ \|p_{Z^LY^{L}W}-p_{\hat{Z}^{L}}p_{Y^L}p_W\|_{TV}\leq 2\delta_1+\delta_2,$$
and the proof is complete.

\subsection{Proof of Theorems~\ref{thm:TVw1w2} and \ref{thm:mainw1w2}} \label{proof:mainw1w2}

Take arbitrary $\mathbf{p}\in \mathcal{P}^c_{\mathtt U}(w_2)$ and $\mathbf{q}\in \mathcal{P}_{\mathtt U}(w_1)$. Since $\mathbf{p}\in \mathcal{P}^c_{\mathtt U}(w_2)$, when Alice plays according to $\mathbf{p}$, Bob has a (pure) strategy $j\in\mathcal{B}$ that reduces Alice's expected payoff to a number less than $w_2$. We can define a random variable $\tilde{W}$ that represents the payoff of Alice when Alice plays $\mathbf{p}$ and Bob plays $j$. Then, the alphabet set of $\tilde{W}$ is $\{u_{1,j}, u_{2,j}, \ldots, u_{n,j}\}$, and its probability distribution over this set is $\mathbf{p}$. We must have $\mathbb{E}[\tilde{W}]<w_2$.

Let us assume that Alice adopts $\mathbf{q}$, but Bob keeps playing the same $j\in\mathcal{B}$. Let $W$ be the random variable describing the payoff of Alice when she plays according to $\mathbf{q}$. Because $\mathbf{q}\in \mathcal{P}_{\mathtt U}(w_1)$, $\mathbb{E}[W]\geq w_1$. We can think of $W$ as taking $u_{i,j}$ with probability $q_i$.

Theorem~\ref{thm:TVw1w2} follows from the following chain of inequalities:
\begin{align}
|w_1-w_2|&\leq |\mathbb{E}[W]-\mathbb{E}[\tilde{W}]|\notag\\
&=\left|\sum_{i}(p_i-q_i)u_{i,j}\right|\nonumber
\\&=\left|\sum_{i}(p_i-q_i)\left(u_{i,j}-\frac{\overline{m}+\underline{m}}{2}\right)\right|\label{eqn:mmm1}
\\&\leq\big(\sum_{i}|p_i-q_i|\big)\max_{i}\bigg|u_{i,j}-\frac{\overline{m}+\underline{m}}{2}\bigg|\nonumber
\\&\leq (2d_1(\mathbf{p},\mathbf{q}))\left(\frac{|\overline{m}-\underline{m}|}{2}\right),\label{eqn:tot-var-dist1}
\end{align}
where \eqref{eqn:mmm1} follows from the fact that $(\overline{m}+\underline{m})/2$ is constant and $\sum_i p_i=\sum_i q_i=1$.

To prove Theorem~\ref{thm:mainw1w2}, the key step is to relate the R\'enyi divergence to \emph{variance}. Remember that $d_2(\mathbf{p},\mathbf{q})=\log(1+\chi^2(\mathbf{p}, \mathbf{q}))$. The following lemma follows from the Chapman--Robbins bound:

\begin{lemma} \label{lemmaVar}
Take an arbitrary sequence of real numbers $(x_1, x_2, \ldots, x_n)$ and two probability distributions $\mathbf{p}=(p_1, p_2, \ldots, p_n)$ and $\mathbf{q}=(q_1, q_2, \ldots, q_n)$. Let $W$ be a random variable that takes value $x_i$ with probability $q_i$, and $\tilde{W}$ be a random variable that takes value $x_i$ with probability $p_i$. Then, we have
\begin{align}\chi^2(\mathbf{p}, \mathbf{q})\geq \frac{(\mathbb{E}[{W}]-\mathbb{E}[\tilde{W}])^2}{\mathsf{Var}[W]}.\label{eqn:DTforml}\end{align}
Furthermore, the above inequality becomes an equality if we set $x_i=(p_i-q_i)/q_i$.
\end{lemma}
Observe that the left hand side of \eqref{eqn:DTforml} depends only on the probability values $p_i$ and $q_i$, while the right hand side depends not only on the probabilities, but also the values that $W$ and $\tilde{W}$ take.
Using this lemma and the fact that $\mathbb{E}[\tilde{W}]<w_2<w_1\leq \mathbb{E}[W]$, we can conclude
$$\chi^2(\mathbf{p}, \mathbf{q})\geq \frac{(\mathbb{E}[{W}]-w_2)^2}{\mathsf{Var}[W]}.$$
Observe that $\underline{m}\leq W\leq \overline{m}$ holds with probability one. The proof is finished by the following lemma.
\begin{lemma} \label{lemmaVar2}
For any $w_1>w_2$, we have
$$\frac{(\mathbb{E}[{W}]-w_2)^2}{\mathsf{Var}[W]}\geq \frac{( w_1 -w_2)^2}{( w_1 -\underbar{m})(\overline{m}- w_1 )},$$
provided that $\mathbb{E}[W]\geq w_1$, and $\underline{m}\leq W\leq \overline{m}$.
\end{lemma}
The proof of this lemma is given in \ref{appndixB}.

\subsection{Proof of Theorem~\ref{theorem-sec4-1}} \label{proof-theorem-sec4-1}
In this section we provide additional details and build a geometric picture. This picture implies Theorem~\ref{theorem-sec4-1}, but also gives a geometric interpretation of Nash strategies.

Let
$$L(a_1, a_2, \ldots, a_n)= \Val(\mathtt U+\mathbf{a}\mathbf{1}^T), \qquad \forall \mathbf{a}\in\mathbb{R}^n.$$
Observe that the table $\tilde{\mathtt{U}}$ can be intuitively understood as giving an additional incentive $a_i$ to Alice for playing her $i$-th action (it is actually a disincentive if $a_i<0$). Also, since
$$L(a_1+c, a_2+c, \ldots, a_n+c)= L(a_1, a_2, \ldots, a_n)+c,$$ we only need to understand $L$ when the sum of the incentives $a_i$ is zero.

We need the following definition:
\begin{definition}
Let
$$K(p_1, p_2, \ldots, p_n)=\min_j \sum_{i}p_iu_{i,j}$$
be the payoff that Alice can guarantee with playing distribution $(p_1, \ldots, p_n)$ with table $\mathtt U$. We extend the definition of $K(\cdot)$ to arbitrary $(p_1, p_2, \ldots, p_n)\in\mathbb{R}^n$ by setting
$$K(p_1, p_2, \ldots, p_n)=-\infty,$$
when the tuple $(p_1, p_2, \ldots, p_n)$ does not form a valid probability distribution, \emph{i.e.,} when any of the $p_i$'s becomes negative, or $\sum_i p_i\neq 1$.
\end{definition}
Note that
$$\mathcal{P}_{\mathtt U}(w)=\left\{\mathbf{p}\Big\vert K(\mathbf{p})\geq w\right\}.$$
A full geometric picture of $\mathcal{P}_{\mathtt U}(w)$ as well as Nash strategies are provided in the following theorem:
\begin{theorem}\label{theorem-sec4-2}We have
\begin{enumerate}
\item The function $L(a_1, a_2, \ldots, a_n)$ is the convex conjugate dual of $K(p_1, p_2, \ldots, p_n)$ in the following sense:
\begin{align}L(\mathbf{a})&=\max_{\mathbf{p}\in\mathbb{R}^n}\left[K(\mathbf{p})+\sum_{i=1}^np_ia_i\right], \qquad\forall \mathbf{a}\in\mathbb{R}^n.
\end{align}
The function $L(a_1, a_2, \ldots, a_n)$ is jointly convex in $(a_1, a_2, \ldots, a_n)$, while $K(p_1, p_2, \ldots, p_n)$ is jointly concave in $(p_1, p_2, \ldots, p_n)$.

Furthermore, the supporting hyperplanes to the convex curve $\mathbf{a}\mapsto L(\mathbf{a})$ characterize Alice's Nash strategies as follows: for any arbitrary $\mathbf{a}$, $\mathbf{p}$ is a Nash strategy of Alice for table $\tilde{\mathtt{U}}=\mathtt{U}+\mathbf{a}\mathbf{1}^T$ if and only if $\mathbf{p}$ is a subgradient of the function $L$ at $\mathbf{a}$. In other words, take some arbitrary vector $\mathbf{a}$. Then, \begin{align}L(b_1, b_2, \ldots, b_n)\geq L(a_1,a_2,\ldots, a_n)+\sum_{i}(b_i-a_i)p_i, \qquad \forall \mathbf{b}\in\mathbb{R}^n,\label{eqn:supporthyp}\end{align}
if and only if $\mathbf{p}$ is a Nash strategy of Alice for the payoff table $\tilde{\mathtt{U}}=\mathtt U+\mathbf{a}\mathbf{1}^T$.

\item Given a probability vector $\mathbf{p}$, we have $$L(a_1, a_2, \ldots, a_n)\geq w+\sum_{i}(a_i-b_i)p_i, \qquad \forall \mathbf{a},$$
if and only if $\mathbf{p}$ guarantees a payoff of at least $w$ for game $\mathtt U+\mathbf{b}\mathbf{1}^T$. In particular, setting $b_i=0$, \begin{align}L(a_1, a_2, \ldots, a_n)\geq w+\sum_{i}a_ip_i, \qquad \forall \mathbf{a},\label{eqn:p2T5}\end{align}
if and only if $\mathbf{p}$ guarantees a payoff of at least $w$ for game $\mathtt U$, \emph{i.e.,} $\mathbf{p}\in\mathcal{P}_{\mathtt U}(w)$.
Thus, having a payoff $w$, we look for hyperplanes of the form $w +\sum_{i}p_ia_i$ that pass through $w$ at $(a_1, a_2, \ldots, a_n)=(0,0,\ldots, 0)$, and lie below the curve of $L$.
\end{enumerate}
\end{theorem}
Observe that the second part of Theorem~\ref{theorem-sec4-2} is equivalent with Theorem~\ref{theorem-sec4-1}.

\begin{proof}[Proof of Theorem~\ref{theorem-sec4-2}]

We begin with the first part of the theorem. Using the max-min formulation for the value of a game $\tilde{\mathtt{U}}$, we have
\begin{align}L(a_1, a_2, \ldots, a_n)&=\max_{{p}_i\geq 0, \sum_i p_i=1}\min_j \sum_{i}p_i(u_{i,j}+a_i)\nonumber
\\&=\max_{{p}_i\geq 0, \sum_i p_i=1}\left[\left(\min_j \sum_{i}p_iu_{i,j}\right)+\sum_{i}p_ia_i\right]\nonumber
\\&=\max_{{p}_i\geq 0, \sum_i p_i=1}\left[K(p_1, p_2, \ldots, p_n)+\sum_{i}p_ia_i\right]\nonumber
\\&=\max_{{p}_i\in\mathbb{R}}\left[K(p_1, p_2, \ldots, p_n)+\sum_{i}p_ia_i\right],\label{eqn:Kneg}
\end{align}
where \eqref{eqn:Kneg} follows from the fact that $K(\mathbf{p})$ is minus infinity when $\mathbf{p}$ is not a probability distribution. This shows the duality of $L(\cdot)$ and $K(\cdot)$. Next, note that $K(\cdot)$ is a minimum of linear functions; hence it is a concave function. Convexity of $L(\cdot)$ can be directly seen from \eqref{eqn:supporthyp} which implies that at least one supporting hyperplane to its curve exists at any given point (since at least one Nash strategy exists for any arbitrary game). Thus, it remains to prove \eqref{eqn:supporthyp}.

Without loss of generality, it suffices to prove \eqref{eqn:supporthyp} for $\mathbf{b}=\mathbf{0}$ ($b_i=0$), and get the result for arbitrary $\mathbf{b}$ by changing variable $\mathtt{U} \rightarrow \mathtt{U}+\mathbf{b}\mathbf{1}^T$.
The inequality
\begin{align}L(a_1, a_2, \ldots, a_n)\geq L(0,0,\ldots, 0)+\sum_{i}a_ip_i, \qquad \forall \mathbf{a},\label{eqnLap}\end{align}
can be also expressed as
$$\min_{\mathbf{a}}\left(L(a_1, a_2, \ldots, a_n)-\sum_{i}a_ip_i\right)\geq L(0,0,\ldots, 0).$$
From the duality relation \eqref{eqn:Kneg} and utilizing the Fenchel's duality theorem, the left hand side is $K(p_1, p_2, \ldots, p_n)$. Thus, the expression is equivalent with \eqref{eqnLap} can be written as
$$K(p_1, p_2, \ldots, p_n)\geq L(0,0,\ldots, 0)=\Val(\mathtt{U}),$$
which is equivalent with $\mathbf{p}$ being a Nash strategy.

The proof for the second part of the theorem is similar. As before $b_i$ can be set to zero. Then, we can express \eqref{eqn:p2T5} as
$$\min_{\mathbf{a}}\left(L(a_1, a_2, \ldots, a_n)-\sum_{i}a_ip_i\right)\geq w.$$
From the duality relation \eqref{eqn:Kneg} and the Fenchel's duality theorem, this is equivalent with
$K(p_1, p_2, \ldots, p_n)\geq w,$
or $\mathbf{p}\in\mathcal{P}_{\mathtt U}(w)$.
\end{proof}

\subsection{Proof of Theorem~\ref{T:direct_sum}} \label{S:direct_sum}
Let $p_{A}(a)$ be an arbitrary distribution that secures payoff $w/2$ in game $\mathtt U$. Observe that $p_{A_1A_2}(a_1,a_2)=\mathbbm{1}(a_1=a_2)p_A(a_1)$ secures payoff $w$ in game $\mathtt U\oplus \mathtt U$, where $\mathbbm{1}(\cdot)$ is the indicator function. In this case, $H(p_{A_1A_2})=H(p_A)$; thus, the function $F_{\mathtt U\oplus \mathtt U}(w)$ is bounded above as follows:
 \begin{equation}\label{E:lessthan}
	F_{\mathtt U\oplus \mathtt U}(w) \leq F_{\mathtt U}(\frac{w}{2}).
\end{equation}
On the other hand, let $p^*_{A_1A_2}\in \mathcal{P}_{\mathtt U\oplus \mathtt U}(w)$ be a distribution with minimum entropy that secures arbitrary payoff $w$ in game $\mathtt U \oplus \mathtt U$, \emph{i.e.,} $H(p^*_{A_1A_2})=F_{\mathtt U \oplus \mathtt U}(w)$. Note that such a distribution $p^*_{A_1A_2}$ exists because $\mathcal{P}_{\mathtt U\oplus \mathtt U}(w)$ is compact. We have
\begin{align}
w&\leq \min_{j_1,j_2\in \mathcal{B}} \E_{p^*_{A_1A_2}} [u_{A_1,j_1}+u_{A_2,j_2}]\notag\\
&= \min_{j_1,j_2\in \mathcal{B}} \left(\E_{p^*_{A_1A_2}} [u_{A_1,j_1}]+ \E_{p^*_{A_1A_2}}[u_{A_2,j_2}]\right)\notag\\
& =\min_{j_1\in \mathcal{B}} \E_{p^*_{A_1}} [u_{A_1,j_1}] + \min_{j_2\in \mathcal{B}} \E_{p^*_{A_2}} [u_{A_2,j_2}],\label{eq:min_entropy_1}
\end{align}
where $p^*_{A_1}$ and $p^*_{A_2}$ are the marginal distributions of $p^*_{A_1A_2}$. Equation \eqref{eq:min_entropy_1} implies that there exists $ k\in \{1,2\}$ such that
\begin{equation*}
\min_{j_k\in \mathcal{B}} \E_{p^*_{A_k}} [u_{A_k,j_k}] \geq\frac{ w}{2}.
\end{equation*}
Thus, $p^*_{A_k}$ secures payoff $w/2$ in game $\mathtt U$. Hence,
\begin{equation}\label{E:greeterthan}
F_{\mathtt U\oplus \mathtt U}(w)=H(p^*_{A_1A_2}) \overset{(a)}{\geq} H(p^*_{A_k}) \geq F_{\mathtt U}(\frac{w}{2}),
\end{equation}
where $(a)$ follows from the properties of entropy.

Equations \eqref{E:lessthan} and \eqref{E:greeterthan} conclude $F_{\mathtt U\oplus \mathtt U}(w)=F_{\mathtt U}(w/2)$. The above line of proof can be extended for every natural number $k$ to prove $F_{\oplus^k \mathtt U}( w )=F_{\mathtt U}( w/k)$.

\subsection{Proof of Theorem~\ref{theorem:second_idea}} \label{sec:second_idea}
 If $v=\overline m$, then there exists a pure strategy $a\in \mathcal{A}$ that guarantees the maximum achievable payoff $\overline m$ in any permissible payoff table $\mathtt U'$. Pure strategies have zero entropy and $G_{\mathtt U}^{(4)}(w)=0$. The explicit expression that we wish to prove for $G_{\mathtt U}^{(4)}(w)$ also vanishes as $v$ tends to $\overline m$, hence, the claim holds when $v=\overline m$. Therefore, for the rest of the proof we suppose that $v<\overline m$.

We claim that an optimal table for the minimization problem in the statement of the theorem is as follows:
\begin{equation}
\mathtt U^* =\left[\begin{array}{cccccc}
\underline{m} & \overline{m} & \overline{m} &\overline{m} & \cdots& \overline{m}\\
\overline{m} & v &\overline{m} & \overline{m} & \cdots& \overline{m}\\
\overline{m} & \overline{m} &v & \overline{m} & \cdots&\overline{m} \\
\overline{m} & \overline{m} &\overline{m} & v & \cdots &\overline{m}\\
\vdots & \vdots &\vdots & \vdots & \ddots& \vdots\\
\overline{m} & \overline{m} &\overline{m} & \cdots & \cdots &v
\end{array}\right]_{m\times m}, \label{eqn:tablemm}
\end{equation}
where the above square table is of size $m=\lfloor (\overline m - v)/(\overline m-w)\rfloor+2$. We first show that there is an optimal payoff table with the structure given in \eqref{eqn:tablemm} for some value of $m$ and then optimize over $m$, the size of the table.

Take some arbitrary feasible table $\mathtt U'=[u'_{i,j}]$ and distribution $\mathbf{p}\in \mathcal P_{\mathtt U'}(w)$. Since $\min_{i\in \mathcal A} \min_{j\in \mathcal B} u'_{i,j}=\underline m$, there exists a row $i$ and a column $j$ such that $u'_{i,j}=\underline m$. By reordering the rows and columns of the table, we may assume that $u'_{1,1}=\underline m$. Furthermore, since $\max_{i\in \mathcal A} \min_{j\in \mathcal B} u'_{i,j}=v$, for all $i$, there exists a $j_i$ such that $u_{i,j_i}'\leq v$. Thus, for all $i\in \mathcal A$ and $j\in \mathcal B$ we have $u'_{i,j}\leq u''_{i,j}$, where $\mathtt U''=[u''_{i,j}]$ is defined as follows:
$$u''_{i,j}=\begin{cases}\underline m, &i=1, j=1,\\ \overline m, &i=1, j\neq 1,\\ \overline m,&i\neq 1, j\neq j_i,\\v,&i\neq 1,j=j_i.\end{cases}$$
The table $\mathtt U''$ has higher payoffs for Alice than the table $\mathtt U'$ in each cell. Hence, $\mathbf p\in \mathcal P_{\mathtt U''}(w)$. The table $\mathtt U''$ also belongs to the feasible set of the minimization problem because it has one row containing one entry with value $\underline m$ and all other rows contain an entry with value $v$.

Next, we can do two more operations: if there is a column with all $\overline{m}$ entries, we delete that column. Because we are restricting the action set of Bob, Alice can still secure payoff $w$ with the same distribution $\mathbf{p}$. If two rows of the table
$\mathtt U''$
 are the same, we can merge them together into one row. In other words, for instance the second and third rows are the same, \emph{i.e.,} $j_2=j_3$, we can merge the second and third rows together and assign the sum probability $p_2+p_3$ to this merged row. This will reduce the entropy of the probability vector $\mathbf p$ while leaving the payoff of Alice unchanged. With these two operations and row/column permutation, we either obtain the table given in \eqref{eqn:tablemm} for some value of $m$, or the following $m\times (m-1)$ table
\begin{equation}
\left[\begin{array}{cccccc}
\underline{m} & \overline{m} & \overline{m} &\overline{m} & \cdots& \overline{m}\\
v & \overline{m} &\overline{m} & \overline{m} & \cdots& \overline{m}\\
\overline{m} & v &\overline{m} & \overline{m} & \cdots&\overline{m} \\
\overline{m} & \overline{m} &v & \overline{m} & \cdots &\overline{m}\\
\vdots & \vdots &\vdots & \vdots & \ddots& \vdots\\
\overline{m} & \overline{m} &\overline{m} & \cdots & \cdots &v
\end{array}\right]_{m\times (m-1)}.\label{table-bad}
\end{equation}
One can directly verify that any probability distribution $\mathbf{p}$ for the payoff of Alice from the table in \eqref{table-bad} is less than or equal to the payoff of Alice from the table in \eqref{eqn:tablemm}.

Having showed that a table with the structure given in \eqref{eqn:tablemm} is optimal, we now optimize over $m$, the size of the table, and compute $G_{\mathtt U}^{(4)}(w)$. The constraint $\mathbf p=(p_1,p_2,\dots,p_m)\in \mathcal P_{\mathtt U}(w)$ is equivalent to:
\begin{align*}
&\underline m p_1+\overline m (1-p_1)\geq w,\\
&v p_i+\overline m (1-p_i)\geq w, \quad\forall i>1,
\end{align*}
which simplifies to
\begin{equation}\label{eq:p_constr}
p_1\leq \frac{\overline m-w}{\overline m - \underline m},\quad
p_i\leq \frac{\overline m-w}{\overline m - v}, \quad\forall i>1.
\end{equation}
We claim that without loss of generality, we can assume that $p_1=0$. Since $\frac{\overline m-w}{\overline m - v}\geq\frac{\overline m-w}{\overline m - \underline m}$, we have
$p_1\leq \frac{\overline m-w}{\overline m - v}.$ Now, if we have a distribution
 $(p_1,p_2,\dots,p_m)$ on a table of the form \eqref{eqn:tablemm} for some $m$, we can consider the distribution
 $(0,p_1,p_2,\dots,p_m)$ on a table with size $m+1$; the entropy of $(0,p_1,p_2,\dots,p_m)$ is the same as $(p_1,p_2,\dots,p_m)$ and achieves the payoff of $w$. Therefore, the problem essentially reduces to the following: what is the minimum possible entropy for the set of all probability distributions $(p_1,p_2,\dots,p_m)$ satisfying
\begin{equation}\label{eq:p_constr2}
p_1=0,\quad p_i\leq \frac{\overline m-w}{\overline m - v}, \quad\forall i=2,\dots,m.
\end{equation}
\color{black}
Note that we should choose $m$ large enough to ensure that the set of probability distributions satisfying \eqref{eq:p_constr2} is non-empty. \color{black} An optimal choice is to take $m=m^*$ for any $m^*$ satisfying $m^*\geq \lfloor (\overline m - v)/(\overline m-w)\rfloor+2$. The optimal pmf $\mathbf p^*$ is as follows:
$$p^*_i=\begin{cases}(\overline m-w)/(\overline m - v), &i=2,\dots,\lfloor (\overline m - v)/(\overline m-w)\rfloor+1,\\
1-\lfloor (\overline m - v)/(\overline m-w)\rfloor (\overline m-w)/(\overline m - v), &i=\lfloor (\overline m - v)/(\overline m-w)\rfloor+2\\0,&\textrm{otherwise}. \end{cases}$$
The distribution $\mathbf p^*$ is the pmf that includes as much elements with value $(\overline m-w)/(\overline m - v)$ as possible. This distribution has minimum entropy because it majorizes all other pmfs in $\mathcal P_{\mathtt U}(w)$ \cite[Theorem 8.0.1]{Nielsen2002}. In other words, for any distribution $\mathbf{p}$ satisfying \eqref{eq:p_constr2}, and for all $i\leq m^*$, we have $\sum_{\ell=1}^i p_{\ell}^{*\downarrow} \geq \sum_{\ell=1}^i p_{\ell}^{\downarrow}$, where for a probability mass function $\mathbf p =(p_1,p_2,\dots,p_m)$, we let $\mathbf p^{\downarrow}=(p_1^{\downarrow},p_2^{\downarrow},\dots,p_m^{\downarrow})$ be the vector whose elements are the elements of $\mathbf p$ reordered into non-decreasing order.

The entropy of $\mathbf{p}^*$ equals
\begin{align*}
H(\mathbf p^*) = -\left\lfloor\frac{\overline m-v}{\overline m-w}\right\rfloor &\frac{\overline m-w}{\overline m-v}\log \left(\frac{\overline m-w}{\overline m-v}\right)- \\
& \left(1-\left\lfloor\frac{\overline m-v}{\overline m-w}\right\rfloor \frac{\overline m-w}{\overline m-v}\right)\log\left(1-\left\lfloor\frac{\overline m-v}{\overline m-w}\right\rfloor \frac{\overline m-w}{\overline m-v}\right)
\end{align*}
as desired.

\subsection{Proof of Theorem~\ref{T:bounds}}\label{proofTbound}
Let $\mathbf{p}^*$ be a Nash strategy for player A where $H(\mathbf{p}^*)=\mathsf{h}^*$ and $i^*\in \arg\max_{i\in \mathcal{A}} \min_{j\in \mathcal{B}}u_{i,j}$ be the pure strategy that guarantees payoff $v$ for player A. For $v\leq w\leq w^*$ define $\alpha=(w^*-w)/(w^*-v)$ and $\mathbf{p}=\alpha \mathbf{e}_{i^*} + (1-\alpha)\mathbf{p}^*$ where $\mathbf{e}_{i^*}$ is a vector of all zero elements except for its $i^*$-th element which is $1$. The pmf $\mathbf{p}$ guarantees payoff $w$ since
$$ \mathbf{p}^T U=\alpha \mathbf{e}_{i^*}^T U + (1-\alpha){\mathbf{p}^*}^T U\geq \alpha v \mathbf{1} + (1-\alpha) w^* \mathbf{1}=w \mathbf{1}.$$
Using the properties of entropy one can inspect that:
$$H(\mathbf{p})\leq \alpha H(\mathbf{e}_{i^*})+(1-\alpha)H(\mathbf{p}^*)+H(\alpha,1-\alpha)=(1-\alpha)\mathsf{h}^*-\alpha \log_2(\alpha)-(1-\alpha)\log_2(1-\alpha).$$
Therefore, $F_{\mathtt U}(w)\leq H(\mathbf{p})\leq (1-\alpha)\mathsf{h}^*-\alpha \log_2(\alpha)-(1-\alpha)\log_2(1-\alpha)$. In addition, as $\mathsf{h}^*$ secures payoff $w^*\geq w$, we have $F_{\mathtt U}(w)\leq \mathsf{h}^*$. These two facts imply $F_{\mathtt U}(w)\leq Q_{\mathtt U}^{(1)}(w)$ once we substitute the value of $\alpha$.

From the definition of $\mathbf{p}^*_{\text{max},i}$, it follows that $\mathbf{p}^*_{\text{max},i}\in \mathcal{P}_{\mathtt U}( w )$. Therefore, for every $i\in \mathcal{A}$, $F_{\mathtt U}(w)\leq H(\mathbf{p}^*_{\text{max},i})$ and $F_{\mathtt U}(w)\leq Q_{\mathtt U}^{(2)}(w)$ follows.

For an arbitrary $j\in \mathcal{B}$, if $\mathbf{p}'\in \{\mathbf{p}\in\mathcal{P}_{\mathtt U}( w ): \sum_{i\in \mathcal A}p_iu_{i,j}= w \}$, then, since $\mathbf{p}'$ also belongs to $\mathcal{P}_{\mathtt U}( w )$, we have $F_{\mathtt U}(w)\leq H(\mathbf{p}')$. Hence, for every $j\in \mathcal{B}$:
$$F_{\mathtt U}(w)\leq \underset{\mathbf{p}\in\mathcal{P}_{\mathtt U}( w ): \sum_{i\in \mathcal A}p_iu_{i,j}= w}{\max}H(\mathbf{p}).$$
The above inequality is correct for every $j\in \mathcal{B}$, thus, $F_{\mathtt U}(w)\leq Q_{\mathtt U}^{(3)}(w)$.

\appendix

\setcounter{equation}{0}
\renewcommand{\theequation}{\Alph{section}.\arabic{equation}}
\renewcommand{\thesection}{Appendix \Alph{section}}

\section{Review of the randomness extraction literature}\label{sec:randomness_extraction}
In this appendix, we review some results about randomness extraction which plays a key role in the construction of optimal strategies for the maximizer in the repeated games studied in this paper. The results of this section are adopted from \cite{renner} who considered the randomness extraction problem in the framework of quantum information theory.

Randomness extraction is the process of deriving some almost fair random bits from a given source $X$. More precisely, let $X\in \mathcal{X}$ and $Y\in\mathcal{Y}$ be two random variables with joint distribution $p_{XY}$ and let $f:\mathcal X \to \{1,2,3,\ldots,2^\ell\}$ be a deterministic mapping such that for $B=f(X)$ and $\epsilon>0$ we have
$$\Big \|p_{BY}-p^{U}_Bp_Y\Big\|_{TV}\leq \epsilon,$$
where $p_{BY}$ is the joint distribution of $B$ and $Y$ and $p_B^{U}$ is the uniform distribution on $\{1,2,3,\ldots,2^{\ell}\}$. Then, we say $f$ extracts $\ell$ fair bits independent of $Y$ from $X$ with precision $\epsilon$.

To precede, we need the definition of conditional collision entropy for a given pair of random variables. The following subsection is devoted for this definition.

\subsection{\qquad\qquad\quad Conditional collision entropy}
Let $X\in \mathcal{X}$ and $Y\in \mathcal Y$ be two random variables with joint probability distribution $p_{XY}$, and let $q_Y$ be an arbitrary probability distribution on the sample space $\mathcal Y$. The conditional collision entropy of $p_{XY}$ given $q_Y$ is defined as follows:\footnote{In this and upcoming definitions, we set $q_Y(y)/p_{XY}(x,y)$ to be zero if $q_Y(y)=p_{XY}(x,y)=0$, and infinity if $q_Y(y)>0$ while $p_{XY}(x,y)=0$.}
$$H_2(p_{XY}|q_Y)=-\log \sum_{(x,y) \in \mathcal X\times \mathcal Y} \frac{p_{XY}(x,y)^2}{q_Y(y)},$$
and the conditional collision entropy of $p_{XY}$ given $Y$ is defined as:
$$H_{2}(p_{XY}|Y)=\max_{q_Y\in \Delta(\mathcal Y) } H_{2}(p_{XY}|q_Y).$$

The conditional collision entropy, as defined above, is not a continuous function of the input probability distribution and a slight change of the probability distribution might result in a large deviation in its conditional collision entropy. The smoothed version of conditional collision entropy resolves this drawback. Let $\epsilon$ be an arbitrary positive number and $p_{XY}$ and $q_Y$ be some probability distributions on $\mathcal{X}\times \mathcal{Y}$ and $\mathcal{Y}$, respectively. The $\epsilon$-smooth conditional collision entropy of $p_{XY}$ given $q_Y$ is
$$H^{\epsilon}_{2}(p_{XY}|q_Y)=\max_{\overline{p}_{XY}\in \mathcal B_{\epsilon}(p_{XY})} H_{2}(\overline p_{XY}|q_Y),$$
where
$$\mathcal B_{\epsilon}(p_{XY})=\{\overline{p}_{XY}\in \Delta(\mathcal X\times \mathcal Y); \| \overline{p}_{XY}-p_{XY}\|_{TV} \leq \epsilon\}.$$
As before, the $\epsilon$-smooth conditional collision entropy of $p_{XY}$ given $Y$ is defined as
$$H^{\epsilon}_{2}(p_{XY}|Y) = \max_{q_Y\in \Delta(\mathcal Y)}H^{\epsilon}_{2}(p_{XY}|q_Y).$$

In the following remark, we show that for i.i.d.\ sources the smoothed conditional collision entropy is related to the Shannon entropy of the source.
\begin{remark}\label{remark:iid_min}
Let $(X^n,Y^n)\in \mathcal X^n\times \mathcal Y^n$ be drawn i.i.d.\ from probability distribution $p_{XY}=p_Yp_{X|Y}$ and for arbitrary $\epsilon >0$, the set of typical sequences, $\mathcal \tau_{\epsilon}^{(n)}$, be defined as
$$\mathcal \tau_{\epsilon}^{(n)}=\left\{(x^n,y^n)\in \mathcal X^n\times \mathcal Y^n: \left|-\frac{1}{n}\log p_{X^n|Y^n}(x^n|y^n)-H(X|Y)\right|\leq \epsilon\right\},$$
where $p_{X^n|Y^n}(x^n|y^n)=\prod_{i=1}^n p_{X|Y}(x_i|y_i)$ is the conditional distribution of $X^n$ given $Y^n$. Considering that $\E[-\log p_{X|Y}(X|Y)] = H(X|Y)$ and
$$-\frac{1}{n}\log p_{X^n|Y^n}(X^n|Y^n)=\frac{1}{n}\sum_{i=1}^n-\log p_{X|Y}(X_i|Y_i),$$
the weak law of large numbers implies that for sufficiently large $n$,
\begin{equation} \label{eq:pr_typical}
\Pr \left[(X^n,Y^n)\in \tau_{\epsilon}^{(n)}\right]\geq 1-\epsilon.
\end{equation}
Next, let $\overline p_{X^nY^n}$ be a probability distribution on $\mathcal X^n\times \mathcal Y^n$ defined as
$$\overline p_{X^nY^n}(x^n,y^n)=\begin{cases}\frac{p_{X^nY^n}(x^n,y^n)}{\Pr \left[(X^n,Y^n)\in \tau_{\epsilon}^{(n)}\right]}&(x^n,y^n)\in \tau_{\epsilon}^{(n)},\\0&\textrm{otherwise.}\end{cases}$$
The definition above along with Equation \eqref{eq:pr_typical} implies that $\|\overline p_{X^nY^n}-p_{X^nY^n}\|_{TV} \leq \epsilon$ holds for sufficiently large $n$. Therefore, for sufficiently large $n$, we have
\begin{align*}
H^{\epsilon}_{2}(p_{X^nY^n}|Y^n)&\geq H_{2}(\overline p_{X^nY^n}|\overline p_{Y^n})\\
&= -\log \left(\sum_{(x^n,y^n)\in \mathcal X^n\times \mathcal Y^n} \overline p_{X^nY^n}(x^n,y^n)\frac{\overline p_{X^nY^n}(x^n,y^n)}{\overline p_{Y^n}(y^n)}\right)\\
&\geq -\log \left(\max_{(x^n,y^n)\in \mathcal X^n\times \mathcal Y^n} \overline p_{X^n|Y^n}(x^n|y^n)\right)\\
&=-\log \left(\max_{(x^n,y^n)\in \tau_{\epsilon}^{(n)}} p_{X^n|Y^n}(x^n|y^n)\right)\\
&\geq n(H(X|Y)-\epsilon).
\end{align*}
\end{remark}

\subsection{\qquad\qquad\quad Randomness extraction}\label{sec:re}
In this subsection, we present the main tools regarding the extraction of randomness from a given source with known distribution.
\begin{theorem}[Theorem 5.5.1, \cite{renner}]\label{theorem:leftover}
Let $X\in \mathcal X$ and $Y\in \mathcal Y$ be two random variables with joint distribution $p_{XY}$. For arbitrary $\ell$ let $F:\mathcal X\to \{1,2,3,\dots,2^{\ell}\}$ be a random mapping constructed by assigning to $F(x)$ uniformly at random one element of $\{1,2,3,\dots,2^{\ell}\}$ and independently for distinct inputs $x,x'\in \mathcal X$. Then, we have
$$\sum_{f} p_F(f)\left\|p_{f(X)Y}-p^{U}_Bp_Y\right\|_{TV}\leq \frac{1}{2}2^{-\frac{1}{2}\left(H_{2}(p_{XY}|Y)-\ell\right)},$$
where the summation is over all deterministic mappings $f:\mathcal X\to \{1,2,3,\dots,2^{\ell}\}$, $p_F$ is the probability distribution of the random mapping $F$, $p_{f(X)Y}$ is the joint distribution of $Y$ and $f(X)$, and $p^{U}_B$ is the uniform distribution on the set $\{1,2,3,\ldots,2^{\ell}\}$. Therefore, there exists a deterministic mapping $f:\mathcal X\to \{1,2,3,\dots,2^{\ell}\}$ such that
$$\left\|p_{f(X)Y}-p^{U}_Bp_Y\right\|_{TV}\leq \frac{1}{2}2^{-\frac{1}{2}\left(H_{2}(p_{XY}|Y)-\ell\right)}.$$
\end{theorem}
\begin{proof}
To begin the proof we need a corollary of the Cauchy-Schwarz inequality stated in the following lemma.
\begin{lemma}\label{lemma:cauchy}
Let $s_1,s_2,\dots,s_n$ and $\gamma_1,\gamma_2,\dots,\gamma_n$ be two sequences of real numbers such that for all $i=1,2,\dots,n$, $\gamma_i\geq 0$ and $\sum_{i=1}^n\gamma_i\leq k$, where $k>0$. Then,
$$\sum_{i=1}^n |s_i|\leq \sqrt{k(\sum_{i=1}^n\frac{s_i^2}{\gamma_i})},$$
where $0/0$ is taken as being zero and $1/0=\infty$.
\end{lemma}
\begin{proof}
If for some $i$, $|s_i|>0$ and $\gamma_i=0$ then the right hand side of the inequality in the statement of the lemma becomes infinity and hence it holds. Otherwise, let for all $i=1,\dots,n$, if $|s_i|>0$, then $\gamma_i>0$; hence, by using Cauchy-Schwarz inequality we have
$$\sum_{i=1}^n |s_i|=\sum_{i=1}^n \sqrt{\gamma_i}\frac{|s_i|}{\sqrt{\gamma_i}}\leq \sqrt{(\sum_{i=1}^n\gamma_i)(\sum_{i=1}^n\frac{s_i^2}{\gamma_i})}\leq \sqrt{k(\sum_{i=1}^n\frac{s_i^2}{\gamma_i})},$$
where we take $0/0=0$ in accordance with the statement of the lemma.
\end{proof}
Let $q_Y$ be an arbitrary probability distribution on $\mathcal Y$ and $f$ be an arbitrary realization of $F$, then, from Lemma~\ref{lemma:cauchy} we have
\begin{align}\label{eq:chain}
\|p_{f(X)Y}-p_{B}^{U}p_Y\|_{TV} &= \frac{1}{2}\sum_{b=1}^{2^{\ell}}\sum_{y\in \mathcal Y}|p_{f(X)Y}(b,y)-p_{B}^{U}(b)p_Y(y)| \notag\\
&\leq \frac{1}{2}\sqrt{2^{\ell}\left(\sum_{b=1}^{2^{\ell}}\sum_{y\in \mathcal Y}\frac{|p_{f(X)Y}(b,y)-p_{B}^{U}(b)p_Y(y)|^2}{q_Y(y)}\right)}\notag\\
&=\frac{1}{2}\sqrt{2^{\ell}\left(\sum_{y\in \mathcal Y}\frac{p_Y(y)^2}{q_Y(y)}\sum_{b=1}^{2^{\ell}}\left|p_{f(X)|Y}(b|y)-p_{B}^{U}(b)\right|^2\right)}.
\end{align}
The term $\sum_{b=1}^{2^{\ell}}|p_{f(X)|Y}(b|y)-p_{B}^{U}(b)|^2$ is characterized as
\begin{align}
\sum_{b=1}^{2^{\ell}}\Big|p_{f(X)|Y}(b|y)&-p_{B}^{U}(b)\Big|^2\notag\\
&=\sum_{b=1}^{2^{\ell}}\left|p_{f(X)|Y}(b|y)-2^{-\ell}\right|^2\notag\\
&=\sum_{b=1}^{2^{\ell}}\left(p_{f(X)|Y}(b|y)^2-2\times2^{-\ell}p_{f(X)|Y}(b|y)+2^{-2\ell}\right)\notag\\
&=\left(\sum_{b=1}^{2^{\ell}} p_{f(X)|Y}(b|y)^2\right)-2^{-\ell}\label{eq:detail}\\
&=\sum_{x,x'\in \mathcal X} p_{X|Y}(x|y)p_{X|Y}(x'|y)\left(\mathbbm{1}(f(x)=f(x'))-2^{-\ell}\right)\label{eq:detail2}\\
&=\sum_{x\in \mathcal X} p_{X|Y}(x|y)^2\left(1-2^{-\ell}\right)\notag\\
&\qquad\qquad+ \sum_{x\neq x'} p_{X|Y}(x|y)p_{X|Y}(x'|y)\left(\mathbbm{1}(f(x)=f(x'))-2^{-\ell}\right)\notag\\
&\leq \sum_{x\in \mathcal X} p_{X|Y}(x|y)^2 + \sum_{x\neq x'} p_{X|Y}(x|y)p_{X|Y}(x'|y)\left(\mathbbm{1}(f(x)=f(x'))-2^{-\ell}\right),\label{eq:chain2}
\end{align}
where $\mathbbm 1(\cdot)$ is the indicator function \color{black} and Equation~\eqref{eq:detail} is implied by the following fact:
$$\sum_{b=1}^{2^{\ell}} (-2\times2^{-\ell}p_{f(X)|Y}(b|y)+2^{-2\ell}) = -2\times 2^{-\ell}\left(\sum_{b=1}^{2^{\ell}}p_{f(X)|Y}(b|y)\right) +2^{-2\ell}\times 2^{\ell} =-2\times 2^{-\ell}+2^{-\ell}=-2^{-\ell} ,$$
where we used $\sum_{b=1}^{2^{\ell}}p_{f(X)|Y}(b|y)=1$. To justify Equation~\eqref{eq:detail2}, note that $\sum_{b=1}^{2^{\ell}} p_{f(X)|Y}(b|y)^2$ is the probability of the event $f(X')=f(X'')$, where $X'$ and $X''$ are i.i.d.\ random variables with distribution $p_{X|Y=y}$, \emph{i.e.}, $p_{X'X''}(x,x')=p_{X|Y}(x|y)p_{X|Y}(x'|y)$. On the other hand, the probability of the event $f(X')=f(X'')$ is also characterized as
$$\sum_{x,x'\in \mathcal X} p_{X|Y}(x|y)p_{X|Y}(x'|y)\mathbbm{1}(f(x)=f(x')).$$
Thus,
$$\sum_{b=1}^{2^{\ell}} p_{f(X)|Y}(b|y)^2=\sum_{x,x'\in \mathcal X} p_{X|Y}(x|y)p_{X|Y}(x'|y)\mathbbm{1}(f(x)=f(x')).$$
The above equation along with the fact that $\sum_{x,x'\in \mathcal X} p_{X|Y}(x|y)p_{X|Y}(x'|y)=1$ implies Equation~\eqref{eq:detail2}. \color{black}

Next, by taking the average of $\sum_{b=1}^{2^{\ell}}\big|p_{f(X)|Y}(b|y)-p_{B}^{U}(b)\big|^2$ with respect to $F$ we have
\begin{align}\label{eq:chain3}
\sum_{f}p_F(f)&\sum_{b=1}^{2^{\ell}}\Big|p_{f(X)|Y}(b|y)-p_{B}^{U}(b)\Big|^2\notag\\
&\leq \sum_{x\in \mathcal X} p_{X|Y}(x|y)^2 + \sum_{x\neq x'} p_{X|Y}(x|y)p_{X|Y}(x'|y)\left(\sum_{f}p_F(f)\mathbbm{1}(f(x)=f(x'))-2^{-\ell}\right)\notag\\
&= \sum_{x\in \mathcal X} p_{X|Y}(x|y)^2 + \sum_{x\neq x'} p_{X|Y}(x|y)p_{X|Y}(x'|y)\left(\Pr[F(x)=F(x')]-2^{-\ell}\right)\notag\\
&= \sum_{x\in \mathcal X} p_{X|Y}(x|y)^2 ,
\end{align}
where the inequality follows from \eqref{eq:chain2}, and the last equality holds because for distinct $x,x'\in \mathcal X$, $F(x)$ coincides with $F(x')$ with probability $2^{-\ell}$. Therefore,
\begin{align*}
\sum_{f} p_F(f)\left\|p_{f(X)Y}-p^{U}_Bp_Y\right\|_{TV} &\leq \frac{1}{2}\sum_{f} p_F(f)\sqrt{2^{\ell}\left(\sum_{y\in \mathcal Y}\frac{p_Y(y)^2}{q_Y(y)}\sum_{b=1}^{2^{\ell}}\left|p_{f(X)|Y}(b|y)-p_{B}^{U}(b)\right|^2\right)}\\
&\leq \frac{1}{2}\sqrt{2^{\ell}\left(\sum_{y\in \mathcal Y}\frac{p_Y(y)^2}{q_Y(y)}\sum_{f} p_F(f)\sum_{b=1}^{2^{\ell}}\left|p_{f(X)|Y}(b|y)-p_{B}^{U}(b)\right|^2\right)}\\
&\leq \frac{1}{2}\sqrt{2^{\ell}\left(\sum_{y\in \mathcal Y}\frac{p_Y(y)^2}{q_Y(y)}\sum_{x\in \mathcal X} p_{X|Y}(x|y)^2\right)}\\
&= \frac{1}{2}\sqrt{2^{\ell}\left(\sum_{(x,y)\in \mathcal X\times \mathcal Y}\frac{p_{XY}(x,y)^2}{q_Y(y)}\right)}\\
&=\frac{1}{2}2^{-\frac{1}{2}\left(H_{2}(p_{XY}|q_Y)-\ell\right)},
\end{align*}
where the first inequality follows from \eqref{eq:chain}, the second inequality follows from utilizing the Jensen's inequality for concave function $\sqrt{\cdot}$, the third inequality follows from \eqref{eq:chain3}, and the second equality follows from the definition of $H_2(p_{XY}|q_Y)$. Since $q_Y$ is arbitrary, the claim of the theorem is implied by the above inequality.
\end{proof}
Now, we can utilize Theorem~\ref{theorem:leftover} to obtain another bound in terms of the smoothed conditional collision entropy. Let $\epsilon>0$ be an arbitrary real number, and $X\in \mathcal X$ and $Y\in \mathcal Y$ be two random variables with joint probability distribution $p_{XY}$. Let $\overline p_{XY} \in \mathcal B_{\epsilon}(p_{XY})$ be the probability distribution such that
$$H_{2}^{\epsilon} (p_{XY}|Y)= H_{2}(\overline p_{XY}|Y),$$
and $f:\mathcal X\to \{1,2,3,\ldots,2^{\ell}\}$ be the mapping of Theorem~\ref{theorem:leftover} satisfying
\begin{equation}\label{eq:leftover1}
\left\|\overline p_{f(X)Y}-p^{U}_B\overline p_Y\right\|_{TV}\leq \frac{1}{2}2^{-\frac{1}{2}\left(H_{2}(\overline p_{XY}|Y)-\ell\right)} = \frac{1}{2}2^{-\frac{1}{2}\left(H^{\epsilon}_{2}(p_{XY}|Y)-\ell\right)},
\end{equation}
where $\overline p_{f(X)Y}$ is the joint distribution of $(f(X),Y)$ when $(X,Y)$ is distributed according to $\overline p_{XY}$, and $\overline p_Y$ is the marginal distribution of $\overline p_{XY}$ with respect to $Y$. Then, we have
\begin{align*}
\left\|p_{f(X)Y}-p^{U}_B p_Y\right\|_{TV}&\leq \left\|p_{f(X)Y}-\overline p_{f(X)Y}\right\|_{TV}+\left\|\overline p_{f(X)Y} - p^{U}_B \overline p_Y \right\|_{TV} + \left\|p^{U}_B \overline p_Y - p^{U}_B p_Y\right\|_{TV}\\
&\leq 2\epsilon + \frac{1}{2}2^{-\frac{1}{2}\left(H^{\epsilon}_{2}(p_{XY}|Y)-\ell\right)},
\end{align*}
where the first inequality results from the triangular inequality for the total variation distance, and the second inequality follows from \eqref{eq:leftover1} and the following two facts:
\begin{itemize}
\item $\left\|p_{f(X)Y}-\overline p_{f(X)Y}\right\|_{TV}\leq \left\|p_{XY}-\overline p_{XY}\right\|_{TV}\leq \epsilon$ : This is a consequence of the forth property of the total variation distance in Lemma~\ref{lemma:tv_p}.
\item $\left\|p^{U}_B \overline p_Y - p^{U}_B p_Y\right\|_{TV} = \left\|\overline p_Y -p_Y\right\|_{TV}\leq \left\|p_{XY}-\overline p_{XY}\right\|_{TV}\leq \epsilon$ : The equality and the first inequality follow from the fifth and the second property (respectively) of the total variation distance in Lemma~\ref{lemma:tv_p}.
\end{itemize}

Therefore, the following corollary of Theorem~\ref{theorem:leftover} is concluded.
\begin{corollary}\label{corollary:smooth}
For random variables $X\in\mathcal X$ and $Y\in \mathcal Y$ with joint distribution $p_{XY}$ and $\epsilon>0$, there exists a deterministic mapping $f:\mathcal X\to \{1,2,.\ldots,2^{\ell}\}$ such that
$$\left\|p_{f(X)Y}-p^{U}_B p_Y\right\|_{TV} \leq 2\epsilon + \frac{1}{2}2^{-\frac{1}{2}\left(H^{\epsilon}_{2}(p_{XY}|Y)-\ell\right)},$$
where $p_B^{U}$ is the uniform distribution on $\{1,2,\dots,2^{\ell}\}.$
\end{corollary}
Next, we utilize Corollary~\ref{corollary:smooth} in combination with Remark~\ref{remark:iid_min} to conclude that for the random sequence $(X^n,Y^n)$ drawn i.i.d.\ from $p_{XY}$ there exists a mapping that extracts almost $nH(X|Y)$ bits from $X^n$ with desired precision given that $n$ is sufficiently large. \color{black} This is the claim of Lemma~\ref{L:yasayi} in Section~\ref{S:mapping}. We finish this section by restating Lemma~\ref{L:yasayi} and providing its formal proof. \color{black}
\begingroup
\def\thetheorem{\ref{L:yasayi}}
\begin{lemma}
Consider the correlated random sequences $X^n$ and $Y^n$ drawn i.i.d.\ from respective spaces $\mathcal{X}$ and $\mathcal{Y}$ by joint probability distribution $p_{XY}$. Let $Q^{(n)}$ be a random variable independent of $Y^n$ and uniformly distributed on $\{1,2,\dots,2^{Rn}\}$, where $R<H(X|Y)$ is a real number; then, there exist mappings $\mathsf B_n:\mathcal{X}^n\to \{1,2,\dots,2^{Rn}\}$ such that
$$\lim_{n\to \infty} \|p_{\mathsf B_n(X^n)Y^n}-p_{Q^{(n)}}p_{Y^n}\|_{TV} =0.$$
\end{lemma}
\addtocounter{theorem}{-1}
\endgroup
\begin{proof}
For arbitrary $\delta>0$ choose some positive number $\epsilon$ such that $\epsilon<\delta/2$ and $\epsilon <H(X|Y)-R$. Let $\mathsf B_n$ be the mapping of Corollary~\ref{corollary:smooth} satisfying
$$\|p_{\mathsf B_n(X^n)Y^n}-p_{Q^{(n)}}p_{Y^n}\|_{TV}\leq 2\epsilon + \frac{1}{2}2^{-\frac{1}{2}\left(H^{\epsilon}_{2}(p_{X^nY^n}|Y^n)-nR\right)}.$$
Then, by using Remark~\ref{remark:iid_min}, for sufficiently large $n$ we have
$$\|p_{\mathsf B_n(X^n)Y^n}-p_{Q^{(n)}}p_{Y^n}\|_{TV}\leq 2\epsilon + \frac{1}{2}2^{-\frac{n}{2}\left(H(X|Y)-R-\epsilon\right)}.$$
Since $\epsilon<\delta/2$ and $H(X|Y)-R-\epsilon>0$, there exists a natural number $n_0$ such that for all $n>n_0$, $2\epsilon + \frac{1}{2}2^{-\frac{n}{2}\left(H(X|Y)-R-\epsilon\right)}\leq\delta$; hence, for sufficiently large $n$,
$$\|p_{\mathsf B_n(X^n)Y^n}-p_{Q^{(n)}}p_{Y^n}\|_{TV}\leq\delta,$$
which finishes the proof.
\end{proof}

\setcounter{equation}{0}
\color{black}
\section{Proof of Proposition~\ref{pro:cardinality}}\label{sec:cardinality}
First, we show that the cardinality of random variable $R$ can be reduced to two. Without loss of generality, we may assume that
$p_{R}(r)>0$ for all $r\in\mathcal R$. We will alter the joint distribution $p_{RQ\bold AS}$ and construct a new joint distribution $p'_{RQ\bold AS}$ on the same alphabet sets such that
\begin{itemize}
\item $p'_{RQ\bold AS}$ still satisfies the constraints \eqref{eqnE1}-\eqref{eqnE3} and \eqref{eq:pi};
\item $p'_{R}(r)>0$ for only two values of $r$, \emph{i.e.,} the support of $R$ under distribution $p'_{RQ\bold AS}$ is of cardinality at most two.
\end{itemize}
Let
$$p'_{Q{\bold A}{S}|R}(q,\bold{a},s|r)=p_{Q{\bold A}{S}|R}(q,\bold{a},s|r), \qquad \forall (q,\bold{a},s,r).$$
In other words, the distributions $p$ and $p'$ differ only in the marginal distribution of $R$. Observe that with this choice of $p'_{Q{\bold A}{S}|R}=p_{Q{\bold A}{S}|R}$, Equations \eqref{eqnE1} and \eqref{eqnE2} are satisfied for $p'_{RQ{\bold A}{S}}$ regardless of how we choose the marginal distribution of $R$; these properties are inherited from $p_{RQ{\bold A}{S}}$.

 To specify the joint distribution $p'_{RQ{\bold A}{S}}$, we need to specify $p'_R(r)$ for different values of $r\in\mathcal{R}$. We can think of the marginal distribution $p'_R$ as a vector of size $|\mathcal{R}|$ of non-negative numbers adding up to one, \emph{i.e.,} a vector $[p'_R(r), r\in\mathcal{R}]$ in the probability simplex
\begin{align}&p'_R(r)\geq 0,\quad, \forall r\in\mathcal R, \label{Eqwder1}\\
&\sum_{r\in \mathcal R}p'_R(r)=1. \label{Eqwder2}\end{align}
Let $H(.|.)$ and $H'(.|.)$ denote the entropy function under distributions $p_{RQ{\bold A}{S}}$ and $p'_{RQ{\bold A}{S}}$, respectively. Similarly, let $\pi(\bold A|R)$ and $\pi'(\bold A|R)$ represent the security level of $\bold A$ given $R$ under distributions $p_{RQ{\bold A}{S}}$ and $p'_{RQ{\bold A}{S}}$, respectively. The expression
$$H'(Q{\bold A}|SR)-H'(Q|R)=\sum_{r\in \mathcal R}p'_R(r) \left[ H'(Q{\bold A}|S, R=r)-H'(Q|R=r) \right]$$
is linear in $p'_R(r)$. We impose the following linear constraint on $p'_R(r)$:
\begin{align}\sum_{r\in \mathcal R}p'_R(r) \left[ H'(Q{\bold A}|S, R=r)-H'(Q|R=r) \right]=H(Q{\bold A}|SR)-H(Q|R).\label{Eqwder3}\end{align}
This linear constraint implies that $H'(Q{\bold A}|SR)-H'(Q|R)=H(Q{\bold A}|SR)-H(Q|R)$. This would ensure \eqref{eqnE3} for $p'_{RQ{\bold A}{S}}$. Now, consider the polytope $\mathsf{P}$ formed by real vectors $[p'_R(r), r\in\mathcal{R}]$ of size $|\mathcal{R}|$ satisfying \eqref{Eqwder1}-\eqref{Eqwder3}. This polytope is non-empty as it includes the vector corresponding to the marginal distribution $[p_R(r)]$. The expression $\pi'(\bold A|R)$ is also a linear expression in $p'_R(r)$ for $r\in\mathcal{R}$. The maximum of the linear function $\pi'(\bold A|R)$ over polytope $\mathsf{P}$ occurs at a vertex of $\mathsf{P}$. We choose the distribution corresponding to this vertex. Thus, $\pi'(\bold A|R)\geq \pi(\bold A|R)$ and \eqref{eq:pi} is satisfied under $p'_{RQ{\bold A}{S}}$. Next, it suffices to show that each vertex of $\mathsf{P}$ corresponds to a joint distribution $p'_R(r)$ with at most two non-negative entries. The polytope $\mathsf{P}$ is defined via $|\mathcal{R}|+2$ hyperplanes given in \eqref{Eqwder1}-\eqref{Eqwder3}. Since $\mathsf{P}$ lies in a space of dimension $|\mathcal{R}|$, each of its vertices must lie at the intersection of at least $|\mathcal{R}|$ hyperplanes defining $\mathsf{P}$. In other words, each vertex must lie on at least $|\mathcal{R}|$ of the $|\mathcal{R}|+2$ hyperplanes given in \eqref{Eqwder1}-\eqref{Eqwder3}. This implies that each vertex must lie on at least $|\mathcal{R}|-2$ hyperplanes of the type given in \eqref{Eqwder1}. In other words, each vertex must have at least $|\mathcal{R}|-2$ zero entries, and hence at most two non-zero entries.

Next, we alter the joint distribution $p'_{RQ\bold AS}$ and construct a new joint distribution $p''_{RQ\bold AS}$ on the same alphabet sets such that
\begin{itemize}
\item $p''_{RQ\bold AS}$ still satisfies the constraints \eqref{eqnE1}-\eqref{eqnE3} and \eqref{eq:pi};
\item $p''_{Q}(q)>0$ for at most $2|\mathcal A|$ values of $q$, \emph{i.e.,} the size of the support of $Q$ under distribution $p''_{RQ\bold AS}$ is at most $2|\mathcal A|$.
\end{itemize}
Let
\begin{equation}\label{eq:prime0}
p''_R(r)=p'_R(r),\quad p''_{\bold AS|RQ}(\bold{a},s|r,q)=p'_{\bold AS|RQ}(\bold{a},s|r,q), \qquad \forall (\bold{a},s,r,q).
\end{equation}
In other words, the marginal distribution on $R$ and the conditional distribution on $\bold AS$ given $RQ$ are preserved.
We will now choose $p''_{Q|R}$ to fulfill the definition of the joint distribution $p''$. Equation \eqref{eq:prime0} implies that Equations \eqref{eqnE1} and \eqref{eqnE2} are satisfied under $p''$. For arbitrary $r\in\mathcal R$ satisfying $p''_R(r)>0$, as before, we interpret $p''_{Q|R=r}$ as a real vector of size $|\mathcal Q|$ in the probability simplex
\begin{align}&p''_{Q|R=r}(q)\geq 0,\quad \forall q\in \mathcal Q, \label{eq:prime1}\\
&\sum_{q\in \mathcal Q}p''_{Q|R=r}(q)=1. \label{eq:prime2}\end{align}
Note that $p''_{\bold A|R=r}(\bold a)=\sum_{q\in \mathcal Q}p''_{Q|R=r}(q)p''_{\bold A|Q,R=r}(\bold a|q)$ is a linear function of $p''_{Q|R=r}$, so we impose the following linear constraints on $p''_{Q|R=r}$:
\begin{equation}\label{eq:prime3}
p''_{\bold A|R=r}(\bold a)=p'_{A|R=r}(\bold a),\quad \forall \bold a\in \mathcal A-\{\bold a_0\},
\end{equation}
where $\bold a_0$ is an arbitrary element of $\mathcal A$ and $\mathcal A-\{\bold a_0\}$ is the set $\mathcal A$ from which $\bold a_0$ is excluded. Consider that as $p''_{\bold A|R=r}$ and $p'_{\bold A|R=r}$ both sum up to one, Equation~\eqref{eq:prime3} guarantees that $p''_{\bold A|R=r}(\bold a_0)=p'_{A|R=r}(\bold a_0)$, thus $p''_{\bold A|R=r}=p'_{\bold A|R=r}$. The fact that $p''_{\bold A|R=r}=p'_{\bold A|R=r}$, along with Equation~\eqref{eq:prime0}, implies
\begin{align}
\pi''(\bold A|R)=\pi'(\bold A|R),\label{eq:prime4}\\
H''(S|R)=H'(S|R),\label{eq:prime5}
\end{align}
where $H''(.|.)$ denotes the entropy function and $\pi''(\bold A|R)$ denotes the security level of $\bold A$ given $R$, both under the distribution $p''_{RQ\bold AS}$. Equation~\eqref{eq:prime4} implies that Equation~\eqref{eq:pi} is satisfied under $p''_{RQ\bold AS}$. Let $\mathsf P'$ denote the polytope of all conditional distributions $p''_{Q|R=r}$ satisfying \eqref{eq:prime1}-\eqref{eq:prime3}. Note that $\mathsf P'$ is non-empty since it includes $p'_{Q|R=r}$. We choose $p''_{Q|R=r}$ to be a vertex of $\mathsf P'$ that maximizes $H''(\bold AS|Q,R=r)$ (linear in $p''_{Q|R=r}$). Therefore, we have $H''(\bold AS|Q,R=r)\geq H'(\bold AS|Q,R=r)$ and considering that $p''_R=p'_R$, we conclude that
\begin{equation}\label{eq:prime6}
H''(\bold AS|QR)\geq H'(\bold AS|QR).
\end{equation}
Then, using Equations~\eqref{eq:prime5} and \eqref{eq:prime6}, we have
\begin{align*}
H''(Q\bold A|RS)-H''(Q|R)&=H''(\bold AS|RQ)-H''(S|R)\\
&\geq H'(\bold AS|RQ)-H'(S|R)\\
&= H'(Q\bold A|RS)-H'(Q|R)\\
&\geq 0,
\end{align*}
where the equations follow from the properties of the entropy function, the first inequality follows from Equations~\eqref{eq:prime5} and \eqref{eq:prime6}, and the last inequality follows from the fact that $p'$ satisfies \eqref{eqnE3}. Thus, Equation~\eqref{eqnE3} is satisfied under $p''_{RQ\bold AS}$. Next, we complete the proof by showing that the support of $Q$ under the distribution $p''_{RQ\bold AS}$ has cardinality of at most $2|\mathcal A|$. To do this, it suffices to show that for arbitrary $r\in \mathcal R$ for which $p''_R(r)>0$, the conditional distribution $p''_{Q|R=r}$ has cardinality at most $|\mathcal A|$. Then, from the fact that $p''_R(r)>0$ for only two values of $r\in \mathcal R$, we can conclude that under $p''_{RQ\bold AS}$, the support of $Q$ has cardinality at most $2|\mathcal A|$. Note that the polytope $\mathsf P'$ belongs to a space of dimension $|Q|$, thus, every vertex of $\mathsf P'$ lies in the intersection of at least $|Q|$ hyperplanes defining $\mathsf P'$ (Equations~\eqref{eq:prime1}-\eqref{eq:prime3}). On the other hand, $\mathsf P'$ is defined by $|\mathcal Q|+|\mathcal A|$ hyperplanes of which $|\mathcal Q|$ hyperplanes are of the type given in \eqref{eq:prime1}; thus, every vertex of $\mathsf P'$ lies in at least $|\mathcal Q|-|\mathcal A|$ hyperplanes of the type given in \eqref{eq:prime1}. Therefore, $p''_{Q|R=r}$ has at least $|\mathcal Q|-|\mathcal A|$ zero-valued entries and hence at most $|\mathcal A|$ non-zero entries.
\color{black}

\setcounter{equation}{0}
\section{Proof of Lemma~\ref{lemmaVar2}}\label{appndixB}
Let us minimize the expression $(\E[W]-w_2)^2/\Var[W]$ over all random variables $W$ that satisfy $\underline{m}\leq W \leq \overline{m}$ and $\E[W]\geq w_1$ :
\begin{align}
\min_{\substack{\underline{m}\leq W \leq \overline{m}\\ \E[W]\geq w_1}} \frac{\left( \E[W]-w_2 \right)^2}{\Var[W]} &= \min_{w_1\leq \mu\leq \overline{m}} \left( \min_{\substack{\underline{m}\leq W \leq \overline{m}\\ \E[W]=\mu }} \frac{\left( \E[W]-w_2 \right)^2}{\Var[W]}\right)
= \min_{w_1\leq \mu \leq \overline{m}} \left(\frac{(\mu -w_2)^2}{\underset{\substack{\underline{m}\leq W \leq \overline{m}\\ \E[W]=\mu }}{\max} \Var[W]}\right)\nonumber
\\&= \min_{w_1\leq \mu \leq \overline{m}} \left(\frac{(\mu -w_2)^2}{-\mu^2+\underset{\substack{\underline{m}\leq W \leq \overline{m}\\ \E[W]=\mu }}{\max} \E[W^2]}\right).
\label{E:min_min}
\end{align}
We claim that
\begin{equation}\label{E:max_var}
\underset{\substack{\underline{m}\leq W \leq \overline{m}\\ \E[W]=\mu }}{\max} \mathbb{E}[W^2]=(\overline{m}-\mu )(\mu -\underline{m})+\mu^2.
\end{equation}
Observe that if $W^*$ is the following binary random variable
\begin{align}\mathbb{P}[W^*=\underline{m}]=\frac{\overline{m}-\mu }{\overline{m}-\underline{m}},\quad \mathbb{P}[W^*=\overline{m}]=\frac{\mu -\underline{m}}{\overline{m}-\underline{m}},
\end{align}
we have $\mathbb{E}[W^*]=\mu $ and $\E[(W^*)^2]=(\overline{m}-\mu )(\mu -\underline{m})+\mu^2$. As a result,
\begin{equation}\label{E:max_var2}
\underset{\substack{\underline{m}\leq W \leq \overline{m}\\ \E[W]=\mu }}{\max} \mathbb{E}[W^2]\geq (\overline{m}-\mu )(\mu -\underline{m})+\mu^2.
\end{equation}
On the other hand, the function $f(x)=x^2$ is convex and lies below the line that connects the two points $(\underline{m}, \underline{m}^2)$ and $(\overline{m}, \overline{m}^2)$ for any $x\in[\underline{m},\overline{m}]$, \emph{i.e.,}
$$x^2\leq \underline{m}^2+(x-\underline{m})(\overline{m}+\underline{m}), \quad \forall x\in [\underline{m}, \overline{m}].$$
Thus,
$$\mathbb{E}[W^2]\leq \underline{m}^2+(\mathbb{E}[W]-\underline{m})(\overline{m}+\underline{m})=(\overline{m}-\mu )(\mu -\underline{m})+\mu^2.$$
Thus, Equation \eqref{E:max_var} holds. As a result, Equation \eqref{E:min_min} becomes
\begin{equation}\label{E:min_g}
\min_{\substack{\underline{m}\leq W \leq \overline{m}\\ \E[W]\geq w_1}} \frac{\left( \E[W]-w_2 \right)^2}{\Var[W]} = \min_{w_1\leq \mu \leq \overline{m}} g(\mu),
\end{equation}
where
$$g(\mu)=\frac{(\mu -w_2)^2}{(\overline{m}-\mu )(\mu -\underline{m})}.$$
Note that the function $g(\cdot)$ is increasing for any $\mu$ satisfying $\underline{m}\leq w_2\leq \mu \leq \overline{m}$ as
\begin{align*}
\frac{d g(\mu )}{d\mu }&=\frac{(\mu -w_2)\left[(\overline{m}-\mu)(\mu+w_2-2\underline{m})+(\mu-\underline{m})(\mu- w_2)\right]}{(\overline{m}-\mu )^2(\mu -\underline{m})^2} \geq 0.
\end{align*}
Since $w_2\leq w_1$, this would then imply that the minimum on the right hand side of \eqref{E:min_g} is obtained at $\mu=w_1$. This will complete the proof.

\section*{Acknowledgements}
The authors would like to thank Yury Polyanskiy for pointing out the connection with the Chapman–Robbins bound. The authors also want to thank the anonymous reviewers for their helpful suggestions.
This work was supported by the Sharif University of Technology under Grant QB950607.

\bibliographystyle{elsarticle-harv}
\bibliography{reference}






\end{document}